%% file: main.tex
\documentclass[a4paper]{llncs}

\usepackage{amsmath}
\usepackage{amssymb}
\usepackage{graphicx}
\usepackage{hyperref}

\usepackage{alltt}
\usepackage{paralist}
\usepackage{booktabs}
\usepackage{multirow}
\usepackage{colortbl}
\usepackage{arydshln}
\usepackage{wrapfig}
\usepackage[export]{adjustbox}
\usepackage{enumitem}
\usepackage{float}
\usepackage{caption}
\usepackage{subcaption}
\usepackage{tcolorbox}
\usepackage{nicefrac}

\usepackage[table,dvipsnames,svgnames]{xcolor}

\usepackage{listings}
\usepackage{verbatim}

\usepackage{cite}
\usepackage[capitalize]{cleveref} 

\usepackage{mathtools}
\usepackage{mathpartir}
\usepackage{wasysym}
\usepackage{pifont}
\usepackage{stmaryrd}
\usepackage{bussproofs}
\usepackage{gensymb}
\usepackage[per-mode=symbol,detect-all]{siunitx}

\usepackage{xspace}
\usepackage{url}
\usepackage{svg}
\usepackage{ifthen}
\usepackage{etoolbox} 
\usepackage[outline]{contour}

\usepackage{bussproofs}
\usepackage{gensymb}
\usepackage[dvipsnames]{xcolor} 
\usepackage[svgnames]{xcolor}

\usepackage{etoolbox}

\usepackage{tikz}
\usetikzlibrary{shapes.geometric,positioning,arrows}

\newcommand{\cmark}{\checkmark}
\definecolor{brightyellow}{RGB}{255,247,192}
\definecolor{lightchacki}{RGB}{245,237,205}
\usepackage{mathtools,mathpartir}
\usepackage{amssymb}
\newcommand{\cfg}[3]{\langle #1,\; #2,\; #3\rangle}
\newcommand{\update}[3]{#1[#2 \mapsto #3]}
\newcommand{\step}{\mathrel{\rightarrow}}
\newcommand{\pstep}{\mathrel{\Rightarrow}}

\newcommand\blfootnote[1]{%
	\begingroup
	\renewcommand\thefootnote{}\footnote{#1}%
	\addtocounter{footnote}{-1}%
	\endgroup
}

\makeatletter

\patchcmd{\@sect}
{\@tempskipa #5\relax}
{\refstepcounter{#1}%
	\Hy@raisedlink{\hyper@anchorstart{#1.\csname the#1\endcsname}\hyper@anchorend}%
	\@tempskipa #5\relax}
{}{}
\makeatother

\definecolor{darkorange}{RGB}{255,140,0}

\newcommand{\xmark}{\textcolor{red}{\ding{55}}}
\usetikzlibrary{shapes.geometric,positioning,arrows} 
\definecolor{headerbg}{RGB}{220,230,241}

\newcommand{\Pre}{\mathsf{pre}}
\newcommand{\Post}{\mathsf{post}}

\newcommand{\toolname}{\textsc{Ser}}

\newcommand{\kw}[1]{\textbf{#1}}
\newcommand{\nondet}{\kw{?}}
\newcommand{\ifkw}{\kw{if}}
\newcommand{\elsekw}{\kw{else}}
\newcommand{\whilekw}{\kw{while}}
\newcommand{\yieldkw}{\kw{yield}}
\newcommand{\requestkw}{\kw{request}}

\newcommand{\sat}{\texttt{SAT}}
\newcommand{\unsat}{\texttt{UNSAT}}

\newcommand{\greencmark}{\textcolor{green}{\ding{51}}}

\lstdefinelanguage{CustomPseudoCode}{
	morekeywords={request, yield, return, if, else, while, and, or},
	morecomment=[l]{//},
	morestring=[b]",
	sensitive=true
}

	\title{Deciding Serializability in Network Systems}

	\author{
			Guy Amir\inst{1} \and
			Mark Barbone\inst{1} \and
			Nicolas Amat\inst{2} \and
			Jules Jacobs\inst{1,3}
		}

\institute{
	Cornell University, Ithaca, USA \\
	\email{ \{gda42, mlb494, jj758\}@cornell.edu} \\
	\and
	DTIS, ONERA, Universit\'e de Toulouse, Toulouse, France\\
	\email{ nicolas.amat@onera.fr}\\
	\and
	Jane Street Capital, New York City, USA \\
}

\let\oldmaketitle\maketitle
\renewcommand{\maketitle}{
  \oldmaketitle
  \pagestyle{plain}  
  \thispagestyle{plain}  
}

\crefformat{section}{\S#2#1#3}
\Crefformat{section}{\S#2#1#3}

\crefformat{subsection}{\S#2#1#3}
\Crefformat{subsection}{\S#2#1#3}
\crefformat{subsubsection}{\S#2#1#3}
\Crefformat{subsubsection}{\S#2#1#3}

\newlength{\subfigheight}

\AtBeginDocument{%
	\settoheight{\subfigheight}{%
		\includegraphics[width=0.23\textwidth]{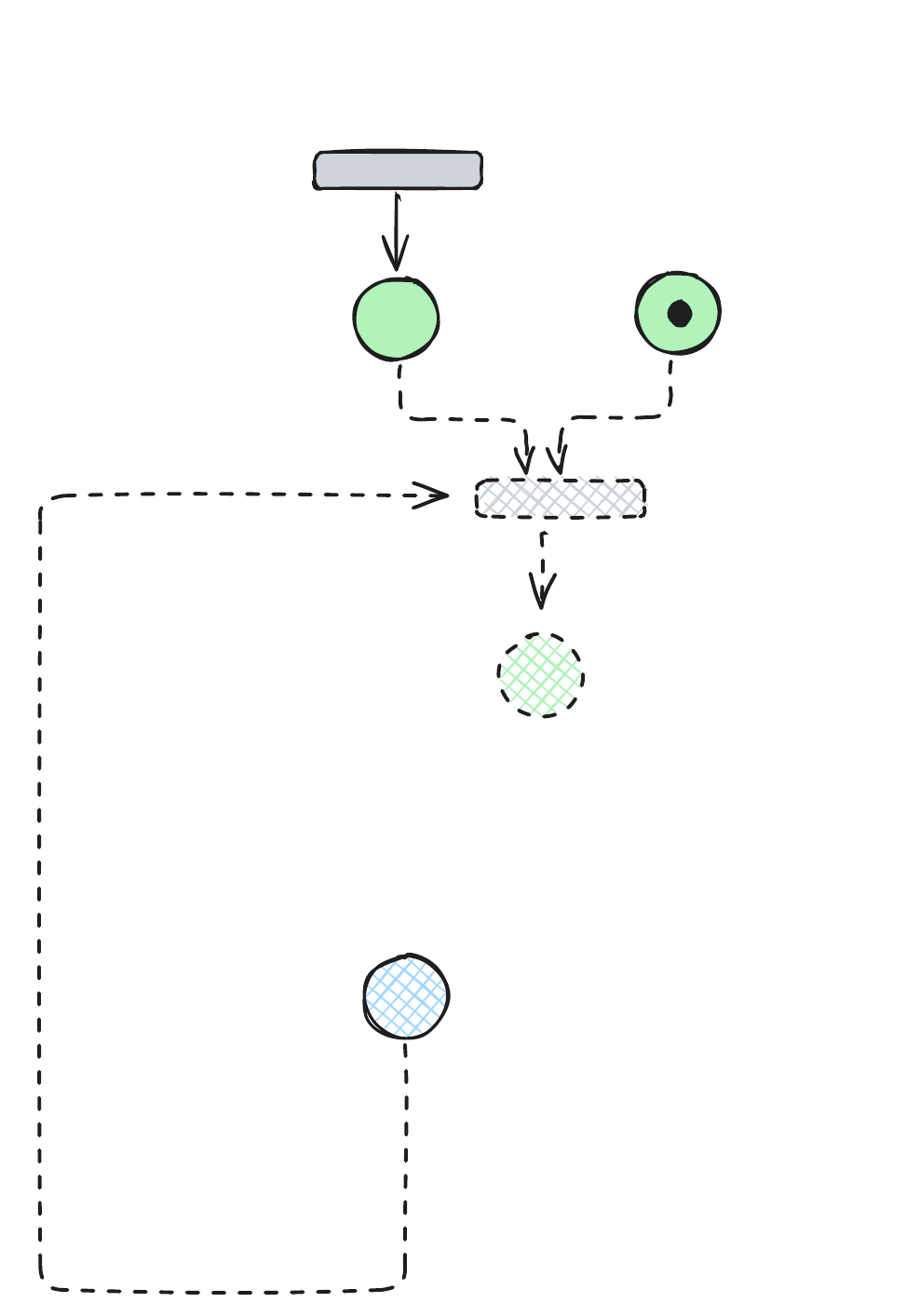}%
	}%
}

\begin{document}

\raggedbottom


\maketitle

\begin{abstract}
	\blfootnote{[*] This paper is an extended version of a paper with the same title presented at the \texttt{TACAS 2026} conference. See \url{https://etaps.org/2026/}.}
	We present the \toolname{} modeling language for automatically verifying \textit{serializability} of concurrent programs, i.e., whether every concurrent execution of the program is equivalent to some serial execution.
	\toolname{} programs are suitably restricted to make this problem decidable, while still allowing for an \textit{unbounded} number of concurrent threads of execution, each potentially running for an \textit{unbounded} number of steps.
	Building on prior theoretical results, we give the first automated end-to-end decision procedure that either proves serializability by producing a checkable certificate, or refutes it by producing a counterexample trace.
	We also present a network-system abstraction to which \toolname{} programs compile. Our decision procedure then reduces serializability in this setting to a Petri net reachability query.
	Furthermore, in order to scale, we curtail the search space via multiple optimizations, including Petri net slicing, semilinear-set compression, and Presburger-formula manipulation.
	We extensively evaluate our framework and show that, despite the theoretical hardness of the problem, it can successfully handle various models of real-world programs, including stateful firewalls, BGP routers, and more.
\end{abstract}


\input{sections/1_introduction}
\input{sections/2_problem_definition}

\input{sections/3_formal_results}

\input{sections/4_implementation}
\input{sections/5_evaluation}

\input{sections/6_related_work}

\input{sections/7_discussion}
\section*{Data and Software Availability}
The data and software necessary to reproduce the experiments in this paper are available as part of the accompanying artifact~\cite{ArtifactRepository}.

\newpage
%
\section*{Acknowledgements}
The work of Amir was
partially supported by a Rothschild Fellowship from Yad Hanadiv (The Rothschild Foundation).
We thank Nate Foster, Fred B. Schneider, Lorin Hochstein, Petr Jancar, and Wolfgang Reisig for their contributions to this project.


{
	\bibliographystyle{splncs04}
	\bibliography{references}
}

\newpage
\appendix
\input{sections/8_appendix_tour}

\input{sections/8_appendix_toy_petri_net.tex}
\input{sections/8_appendix_ser_semantics}
\input{sections/8_appendix_more_NS_examples}
\input{sections/8_appendix_NS_to_PN_formulation}
\input{sections/8_appendix_serializable_program_proof_example}
\input{sections/8_appendix_non_serializable_execution_counterexample}
\input{sections/8_appendix_bidirectional_optimization_proof.tex}
\input{sections/8_appendix_big_table_and_more_results.tex}
\input{sections/8_appendix_SMPT.tex}

\end{document}

%% file: sections/1_introduction.tex
\section{Introduction}
\label{sec:introduction}

In the domain of concurrent systems, from databases to software-defined networks (SDNs)~\cite{KrRaVePaRoAzUh14,XiWeFoNiXi15}, a cornerstone correctness criterion is \emph{serializability}: every concurrent execution must produce outcomes equivalent to some serial ordering of requests. Violations of serializability can lead to subtle anomalies, such as lost updates in databases or routing cycles in SDNs.
While we can check serializability for a fixed number of requests with known execution traces (e.g., by enumerating all possible interleavings), the problem is undecidable for general programs, requiring techniques such as runtime verification or incomplete bounded model checking~\cite{WaSt06a,WaSt06b,FlFrYi08,FaMa08,SiMaWaGu11a,SiMaWaGu11b,Pa79,AlMcPe96,BiEn19}.

However, Bouajjani et al.~\cite{BoEmEnHa13} have shown (as a special case of bounded-barrier linearizability) that for programs with bounded-size state, this problem is decidable even for an \emph{unbounded} number of in-flight requests, each performing an \emph{unbounded} number of steps. The purpose of this paper is to make this theoretical decidability result a reality by designing the first decision procedure and putting forth practical algorithms that either prove serializability (with a proof certificate) or prove non-serializability (with a counterexample trace).
We illustrate the problem by example:


\noindent
\begin{minipage}[t]{0.55\textwidth}
	\begin{minipage}[t]{\textwidth}
		\begin{lstlisting}[caption={Without yielding (serializable)},
			label={lst:MotivatingExample1Ser}]
  // request handler           
  request main: 
      X := 1 // X is global
      y := X // y is local
      X := 0
      return y 
		\end{lstlisting}
	\end{minipage}
	\vspace{1em}
	\begin{minipage}[t]{\textwidth}
		\begin{lstlisting}[caption={With yielding (not serializable)},
			label={lst:MotivatingExample2NonSer}]
  request main: 
      X := 1 
      yield // another request
      y := X // can read 0!
      X := 0
      return y 	
		\end{lstlisting}
	\end{minipage}
\end{minipage}%
\hfill
\begin{minipage}[t]{0.35\textwidth}
	\begin{lstlisting}[caption={With yielding and a spin-lock (serializable)},
		label={lst:MotivatingExample3Ser}]
  request main: 
      // lock
      while (L == 1): 
          yield
      L := 1 

      X := 1
      yield
      y := X 
      X := 0

      // unlock    
      L := 0
      return y 
	\end{lstlisting}
\end{minipage}

These examples are written in our modeling language called \toolname.
A \toolname{} program has a set of named \textbf{request handlers} (one handler, \texttt{main}, in the examples) that are arbitrarily invoked concurrently by the external environment.
Each incoming request processes its request handler's body until it returns a value as its \textbf{response}. Concurrency is managed by the \(\yieldkw\)
 statement, which pauses the current request and gives other requests a chance to run. \toolname{} programs have uppercase \textbf{global shared variables} (\texttt{X} in the examples) and lowercase \textbf{request-local variables} (\texttt{y} in the examples).
The first program (Listing~\ref{lst:MotivatingExample1Ser}) is clearly serializable because there are no yields, and hence, no interleavings: each \texttt{main} request returns 1.
In the second program (Listing~\ref{lst:MotivatingExample2NonSer}), the \texttt{yield} allows interleavings that make the program \emph{non-serializable}. For instance, consider two concurrent requests to \texttt{main}:
Request A executes \texttt{[X := 1]} then yields to Request B; which then
executes \texttt{[X := 1]}, yields to itself, reads \texttt{X} (getting 1), sets \texttt{[X := 0]}, and returns 1.
Finally, 
Request A resumes, reads \texttt{X} (now 0), and returns 0.
This produces the multiset \{(\texttt{main}, 0), (\texttt{main}, 1)\} of (request, response) pairs, which is impossible in any serial execution (where all \texttt{main} requests return 1 and never 0).
Of course, having \texttt{yield}s does not guarantee that an execution is necessarily not serializable, as observed in the third snippet (Listing~\ref{lst:MotivatingExample3Ser}). This program uses an additional lock variable (\texttt{L}), which guarantees that even if an interleaving occurs, the program is semantically equivalent to the first one.
These examples demonstrate that reasoning about serializability can be complex even for very simple programs with few requests running concurrently.
For a tour of additional examples, we refer the reader to Appendix~\ref{appendix:tour}.

\smallskip
\noindent
\textbf{Problem Definition.}
Formally, we define the \textbf{observable execution} of a \toolname{} program as a multiset of (request, response) pairs induced by a specific interleaving. The \textbf{observable behavior} of a \toolname{} program is the set of all possible observable executions that can occur such that the requests are executed concurrently to obtain their paired responses.
A \toolname{} program is \textbf{serializable} if every observable behavior is achievable serially (without interleavings). Differently put, removing all \texttt{yield} statements does not change the program's semantics.
\emph{This paper aims to present the \toolname{} language and decision procedure for this problem.} In particular, \toolname{} is the first toolchain to \textbf{automatically} prove serializability without requiring manual work by the user.

\smallskip
\noindent
\textbf{Challenges.}
To our knowledge, no prior implementation exists that can automatically generate proof certificates for this class of concurrent systems.
Why not?
Our decision procedure builds on Bouajjani et al.'s reduction from serializability to Petri net (PN) reachability~\cite{BoEmEnHa13}. However, since PN reachability is \texttt{Ackermann}-complete~\cite{CzWo22,Le22}, a naive implementation would fail on all but simple programs. 

\smallskip
\noindent
\textbf{Our Approach.}
To address this, we first introduce the abstraction of \textit{network systems} (NS) --- modeling concurrent programs where users send \textit{requests} that manipulate local and shared state before returning \textit{responses}. A \toolname{} program is compiled into a network system, on which our decision procedure operates via reduction to Petri net reachability and semilinear set analysis.
%
We note that while our approach is sound (never incorrectly claims serializability), the underlying reachability query may time out on complex instances, limiting completeness in practice (this is unavoidable, given the \texttt{Ackermann}-hardness of the problem).
Towards this end, we developed multiple optimizations to make the approach practical, including Petri net slicing, semilinear set compression, and additional manipulations with Presburger formulas.
As we demonstrate, these optimizations reduce the search space by \textit{orders of magnitude}, enabling us to scale to non-trivial programs.
Finally, 
we extensively evaluated our \toolname{} toolchain on various programs, covering a broad spectrum of features such as loops, branching, locks, and nondeterminism; as well as SDN-inspired examples such as stateful firewalls, BGP routers, and more.
To our knowledge, this leads to the first \emph{implemented} decision procedure that: (i) automatically \textit{proves} serializability for \textit{unbounded} executions; (ii) generates \textit{proof certificates}; and (iii) handles \textit{non-trivial programs}.

\smallskip
\noindent
\textbf{Contributions.}
We introduce in  \Cref{sec:problem-definition} our \toolname{} language and the Network System program abstraction that captures the essence of concurrent systems.
    In \Cref{sec:formal-results} we present the core decision procedure with proof certificates, and our various optimizations.
    %
    The implementation of the \toolname{} toolchain is covered in \Cref{sec:implementation}, and its evaluation is presented in \Cref{sec:evaluation}. 
We discuss related work in \Cref{sec:related-work} and conclude in \Cref{sec:discussion}.
Our tool, benchmarks, and experiments are available as an accompanying artifact~\cite{ArtifactRepository}.
We also include an appendix with technical details and examples.

%% file: sections/2_problem_definition.tex
\vspace{-3.0pt}
\section{Problem Definition}
\label{sec:problem-definition}

\subsection{Background}
\label{subsec:background}

\textbf{Petri nets.} 
A Petri net is $N=(P,T,\mathsf{pre},\mathsf{post},M_0)$ with a set of places $P$, a set of transitions $T$, flow functions $\mathsf{pre},\mathsf{post}:T\to\mathbb N^P$, and an initial marking (token distribution) $M_0\in\mathbb N^P$.  
A transition $t$ is \textit{enabled} at marking $M \in \mathbb N^P$ if $M\ge \mathsf{pre}(t)$ coordinate-wise, i.e., $M$ provides at least as many tokens as required by $\mathsf{pre}(t)$.
If an enabled transition \(t\) \textit{fires}, it produces the marking $M'$ (denoted $M\xrightarrow{t}M'$), with $M'= M - \mathsf{pre}(t) + \mathsf{post}(t)$ consuming input tokens and producing output tokens.
This can extend naturally to a sequence of firings $\sigma = t_1\cdots t_k$ (denoted $M \xrightarrow{\sigma} M'$), giving rise to a sequence of markings $M_0,\ldots,M_k$ with $M=M_0$, $M'=M_k$, and $M_i \xrightarrow{t_{i}} M_{i+1}$ for all $i$. 
We define the set $R(N)=\{M \mid \exists \sigma\in T^*. M_0 \xrightarrow{\sigma} M\}$ to include all markings reachable from the initial state $M_0$.
The \textit{reachability problem} asks, given a Petri net $N$ and a marking $M$, whether $M\in R(N)$. 
Specifically, we focus on reachability of a linear-constraint formula $\mathcal {F}$; it is \sat\ if some marking $M\in R(N)$ satisfies $\mathcal {F}$ (denoted $M \models \mathcal {F}$), and otherwise \unsat{} (see the toy example in 	Appendix~\ref{appendix:toyPN}).
Surprisingly, even the \textit{unbounded} case, where places hold arbitrarily many tokens, is decidable~\cite{Ma81,Ko82,La92}, although \texttt{Ackermann}-complete~\cite{CzWo22,Le22}.

\medskip
\noindent
\textbf{Verdict proofs.} 
If $\mathcal {F}$ is reachable, a witness sequence $\sigma\in T^*$ with $M_0\xrightarrow{\sigma}M$ and $M\models \mathcal {F}$ serves as a proof, and is verifiable by simulation of the Petri net.  
If $\mathcal {F}$ is unreachable, there exists~\cite{Le09} an inductive Presburger certificate $C$ proving non-reachability: 
(i) $M_0\models C$, (ii) $\forall t\in T \quad (M\models C \wedge M\xrightarrow{t}M')\Rightarrow M'\models C$, and (iii) $C\Rightarrow \neg \mathcal {F}$.

\medskip
\noindent
\textbf{Semilinear sets and Parikh’s theorem.}
A set $S\subseteq\mathbb N^k$ is \textit{semilinear} if 
$S=\bigcup_{i=1}^m \{\mathbf b_i+\sum_{j=1}^{r_i} n_j\mathbf p_{i,j}\mid n_j\in\mathbb N\}$ 
for some $\mathbf b_i,\mathbf p_{i,j}\in\mathbb N^k$. 
Such sets coincide with those definable by \textit{Presburger arithmetic}~\cite{Pr29}. 
By Parikh’s theorem~\cite{Parikh66}, the \textit{Parikh image} of any context-free language (CFL) is semilinear. 
For an alphabet $\Sigma=\{a_1,\dots,a_k\}$ and a word $w\in\Sigma^*$, the \emph{Parikh image} of $w$ is the vector
\(
\mathsf{Parikh}(w)\;=\;\bigl(|w|_{a_1},\dots,|w|_{a_k}\bigr)\in\mathbb{N}^k,
\)
where $|w|_{a_i}$ denotes the number of occurrences of the symbol $a_i$ in $w$.
For a language $L\subseteq\Sigma^*$, we write
\(
\mathsf{Parikh}(L)\;=\;\{\mathsf{Parikh}(w)\mid w\in L\}\subseteq\mathbb{N}^k
\).


\medskip
\noindent
\textbf{Deciding serializability in unbounded systems.} 
Bouajjani et al.~\cite{BoEmEnHa13} have proved that serializability in unbounded systems reduces to Petri net reachability, as a special case of \textit{bounded-barrier linearizability}.

\clearpage
\subsection{The SER Language}
Our \toolname{} syntax is defined as follows: 

\[
\begin{aligned}
	\mathbf{Expression}\quad e ::= &\ 0 \mid 1 \mid 2 \mid \dots            && \text{numeric constant}\\
	&\mid \nondet                            && \text{nondeterministic\ value (0/1)}\\
	&\mid x := e \mid x                      && \text{write/read local variable}\\
	&\mid X := e \mid X                      && \text{write/read global variable}\\
	&\mid e_1 == e_2                         && \text{equality test}\\
	&\mid e_1 ; e_2                          && \text{sequencing}\\
	&\mid \ifkw(e_1)\{e_2\}\elsekw\{e_3\}    && \text{conditional}\\
	&\mid \whilekw(e_1)\{e_2\}               && \text{while loop}\\
	&\mid \yieldkw                           && \text{yield to scheduler}
	\\[0.8em]
	\mathbf{Program}\quad P_0 ::= &\ \requestkw\ name_1\{e_1\}\;\dots\;\requestkw\ name_n\{e_n\}
	&& \text{set of handlers} 
\end{aligned}
\]

Given a program \(P_0\), we write \(name_i\{e_i\}\in P_0\) for each of the handlers \(name_i\{e_i\}\).
Our semantics is standard and fully formalized in Appendix~\ref{appendix:ser-semantics}. 
 In addition, arithmetic extensions are supported in the tool~\cite{ArtifactRepository} but omitted here for brevity.
%
%

\subsection{Network System}    
We now present our abstract network system (NS) model, motivated by software-defined networks. In the networking domain, spawning a request corresponds to sending a \textit{packet}, with each local variable mapped to a unique \textit{packet header field}; global variables correspond to variables on \textit{programmable switches}, as they are shared among all requests visiting the switch. Throughout this paper, we use the term \emph{request} to refer to a concurrent computation unit. 
We define a network system $\mathcal{N}$ as a tuple $(G, L, \mathit{REQ},  \mathit{RESP}, g_0, \delta, \mathit{req}, \mathit{resp})$ where:
\begin{itemize}
\item $G$ is a set of \textit{global network states} (e.g., the values of variables on a switch)

\item $L$ is a set of \textit{local packet states} (e.g., packet header values)

\item $\mathit{REQ}$ is a finite set of \textit{request labels} (each marked {\color{ForestGreen}$\blacklozenge_\text{req}$})

\item $\mathit{RESP}$ is a finite set of \textit{response labels} (each marked {\color{red}$\blacklozenge_\text{resp}$})

\item $g_0 \in G$ is the \textit{initial global state} of the network system

\item $\mathit{req} \subseteq \mathit{REQ} \times  L$ maps each request to its corresponding (initial) local state --- this represents externally spawning a packet matching the request type

\item $\mathit{resp} \subseteq L \times \mathit{RESP}$ maps a (final) local state to its corresponding response (this represents a packet exiting the network and returning the computation)

\item $\delta \subseteq  (L \times G) \times ( L \times G)$ defines atomic execution steps that update both global and local state (this represents a packet doing a single hop in the network)
\end{itemize}


\noindent
\textbf{Request and response.}
A \emph{request} label ${\color{ForestGreen}\blacklozenge_\text{req}}\in\mathit{REQ}$ spawns a request (i.e., a packet/thread) on which a concurrent computation is executed; a \emph{response} label ${\color{red}\blacklozenge_\text{resp}}\in\mathit{RESP}$ is its returned value. The pair $({\color{ForestGreen}\blacklozenge_\text{req}},{\color{red}\blacklozenge_\text{resp}})$ captures the input/output behavior of a single request from a single concurrent execution of the NS.

\smallskip
\noindent
\textbf{States.}
A \emph{network state} is a triple $(g,\mathcal{R},\mathcal{Z})$ where
$g \in G$ is the global network state,
$\mathcal{R} \in \mathrm{Multiset}(L \times \mathit{REQ})$ is a multiset of in-flight requests (i.e., local assignments of each thread in the current timestep, and the original request label that spawned it),
and $\mathcal{Z} \in \mathrm{Multiset}(\mathit{REQ} \times \mathit{RESP})$ is a multiset of completed request/response pairs.
We write $\uplus$ for multiset union.
The initial global state is $(g_0, \varnothing, \varnothing)$.

\smallskip
\noindent
\textbf{Transition rules.}
A transition $\longrightarrow$ either (1) spawns a request; (2) advances one request via $\delta$; or (3) returns a response. When no further steps remain, $\mathcal{Z}$ is the final multiset of request/response pairs that arose during the NS run.



\[
\text{(New Request)}\quad
\infer{({\color{ForestGreen}\blacklozenge_\text{req}},\ell)\in\mathit{req}}
{(g,\mathcal{R},\mathcal{Z}) \rightarrow (g,\; \mathcal{R}\uplus\{(\ell,{\color{ForestGreen}\blacklozenge_\text{req}})\},\; \mathcal{Z})}
\]
\[
\text{(Processing Step)}\quad
\infer{((\ell, g),(\ell', g'))\in\delta}
{(g,\; \mathcal{R}\uplus\{(\ell,{\color{ForestGreen}\blacklozenge_\text{req}})\},\; \mathcal{Z})
	\rightarrow
	(g',\; \mathcal{R}\uplus\{(\ell',{\color{ForestGreen}\blacklozenge_\text{req}})\},\; \mathcal{Z})}
\]
\[
\text{(Response)}\quad
\infer{(\ell,{\color{red}\blacklozenge_\text{resp}})\in\mathit{resp}}
{(g,\; \mathcal{R}\uplus\{(\ell,{\color{ForestGreen}\blacklozenge_\text{req}})\},\; \mathcal{Z})
	\rightarrow
	(g,\; \mathcal{R},\; \mathcal{Z} \uplus \{({\color{ForestGreen}\blacklozenge_\text{req}},{\color{red}\blacklozenge_\text{resp}})\})}
\]

\smallskip
\noindent
\textbf{Serializability.}
An \textit{interleaved run} is a complete execution 
\((g_0,\varnothing,\varnothing)\!\to^*\!(g_n,\varnothing,\mathcal{Z})\):
\[
(g_0,\varnothing,\varnothing) \;\to\; (g_1,\mathcal{R}_1,\mathcal{Z}_1) \;\to\; \cdots \;\to\; 
(g_n,\mathcal{R}_{n},\mathcal{Z}_{n}),
\quad
\text{with } \mathcal{R}_{n}=\varnothing,\mathcal{Z}_n=\mathcal{Z}.
\]
It is \textit{serial} if each $\mathcal{R}_i$ has \textit{at most} one request.
	Intuitively, serial runs have at most one in-flight request at a time.
	%
	%
	Given NS \(\mathcal{S}\), let \(\mathsf{Int}(\mathcal{S})\) and \(\mathsf{Ser}(\mathcal{S})\)  respectively denote the (infinite) sets of request/response multisets, for interleaved and serial runs of \(\mathcal{S}\):
\[
\mathsf{Int}(\mathcal{S})=\{\,\mathcal{Z}\mid \exists\text{ \textit{interleaved} run }(g_0,\varnothing,\varnothing)\rightarrow^{*}(g_n,\varnothing,\mathcal{Z})\,\},
\]
\[
\mathsf{Ser}(\mathcal{S})=\{\,\mathcal{Z}\mid \exists\text{ \textit{serial} run }(g_0,\varnothing,\varnothing)\rightarrow^{*}(g_n,\varnothing,\mathcal{Z})\,\}\subseteq \mathsf{Int}(\mathcal{S}).
\]

An NS $\mathcal{S}$ is \emph{serializable} if \(\mathsf{Int}(\mathcal{S})=\mathsf{Ser}(\mathcal{S})\), i.e., every multiset of request/response pairs obtained by an interleaved execution can also be obtained serially.

\subsection{Translating SER Programs to Network Systems}
\label{subsec:SerToNsTranslation}
The NS abstraction not only captures concurrent behaviors in software-defined networks but also enables a natural translation from \toolname{} programs. 
Given a \toolname{} program \(P_0\) with local variables (\texttt{vars}), global variables (\texttt{VARS}), and mappings \(\rho,g\) from these to a finite value set \texttt{V}, we define the initial local/global states $\rho_0$ and $g_0$ assigning $0$ to all local and global variables, respectively. 
Using the small-step semantics ($\pstep$, defined in full in Appendix~\ref{appendix:ser-semantics}), we construct the NS $(G, L, \mathit{REQ}, \mathit{RESP}, g_0, \delta, \mathit{req}, \mathit{resp})$:


\[
\begin{aligned}
	G \;&=\; \{\, g : \texttt{VARS}\!\to\! \texttt{V} \,\},\\[0.3ex]
	L \;&=\; \bigl\{\, (e,\rho)\ \bigm|\ \rho:\texttt{vars}\!\to\! \texttt{V},\ \exists\,name_i\{e_i\}\!\in\! P_0\ \text{s.t.} \\[-0.2ex]
	&\qquad\quad\qquad\quad e=e_i \text{ or } e \text{ is a suffix of } e_i \text{ starting after a } \yieldkw \text{ statement}\,\bigr\},\\[0.3ex]
	REQ \;&=\; \{\, name_i \mid name_i\{e_i\}\!\in\! P_0 \,\},\quad RESP \;=\; \texttt{V},\\[0.3ex]
	req \;&=\; \{\, (r,\ell) \mid r=name_i\!\in\!REQ,\ name_i\{e_i\}\!\in\! P_0,\ \ell=(e_i,\rho_0)\!\in\!L\},\\[0.3ex]
	resp \;&=\; \{\, (\ell',r') \mid \exists v\!\in\!\texttt{V}.\ \ell'=(v,\rho')\!\in\!L,\ r'=v\!\in\!RESP \,\},\\[0.3ex]
	\delta \;&=\; \bigl\{\, ((e,\rho),g)\!\to\!((e',\rho'),g') \ \bigm|\ (e,\rho),(e',\rho')\!\in\!L,\ g,g'\!\in\!G,\ 
	\cfg{e}{\rho}{g} \pstep \cfg{e'}{\rho'}{g'} \,\bigr\}.
\end{aligned}
\]

\begin{tcolorbox}[colback=black!5!white, colframe=black, boxrule=1pt]
\textbf{Example.}
We construct the NS for the non-serializable example in Listing~\ref{lst:MotivatingExample2NonSer}:
\begin{itemize}
\item 
The set $G$ is defined as $G=\{[\texttt{X=0}], [\texttt{X=1}]\}$.

\item 
The initial global state is defined as $g_0 = [\texttt{X=0}]$.

\item 
The set $L$ is defined as all reachable local states, i.e., pairs of assignments (such as $[\texttt{y=0}], [\texttt{y=1}]$) coupled with all reachable \toolname{} programs (continuations of a program at a point of execution).
For example, the reachable programs for Listing~\ref{lst:MotivatingExample2NonSer} are depicted as code snippets in Fig.~\ref{fig:code2ExampleNS}.

\item 
The set of requests is $REQ = \{{\color{ForestGreen}\blacklozenge_\text{main}}\}$.

\item 
The set of responses is $RESP = \{{\color{red}\blacklozenge_0},{\color{red}\blacklozenge_1}\}$.

\item
The function \(\delta\) is presented in Fig.~\ref{fig:code2ExampleNSSecondPart} (Appendix~\ref{appendix:MoreNsExamples}).

\end{itemize}

We depict in Fig.~\ref{fig:code2ExampleNS} the explicit network system that
serves as a mapping from requests ({\color{ForestGreen}$\blacklozenge_\text{main}$}) to responses ({\color{red}$\blacklozenge_0$}, {\color{red}$\blacklozenge_1$}).
We note that, for simplicity, we depict only \textit{reachable} states.
\end{tcolorbox}

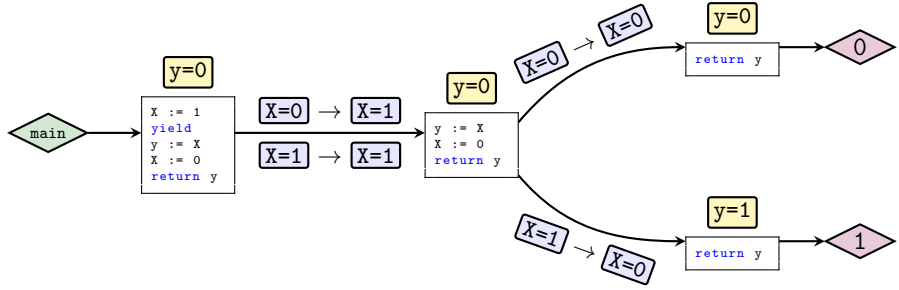
\begin{figure}[!htbp]
	\centering

	\begin{tikzpicture}[
		node distance=1.5cm and 2.5cm,
		>=stealth,
		thick,
		every node/.style={font=\small}
	]
	  
	  \node[
		draw=black,
		line width=0.8pt,
		fill=ForestGreen!20,
		text=black,
		diamond,
		aspect=2,
		inner sep=2pt,
		scale=0.7
	  ] (main) {\texttt{main}};
	  
	  \node[right=0.7cm of main, align=center] (state1) {
		\begin{tikzpicture}[baseline=(ybox.base)]
			\node[
			draw=black,
			line width=0.8pt,
			fill=brightyellow,
			text=black,
			rectangle,
			rounded corners=1pt,
			inner sep=2pt
			] (ybox) {\texttt{y=0}};
		\end{tikzpicture}\\[-2.5pt]
		\begin{minipage}{1.0cm}
			\begin{lstlisting}[language=CustomPseudoCode,numbers=none,basicstyle=\tiny\ttfamily]
X := 1
yield
y := X
X := 0
return y
			\end{lstlisting}
		\end{minipage}
	  };
	  
	  \node[right=of state1, align=center] (state2) {
		\begin{tikzpicture}[baseline=(ybox.base)]
			\node[
			draw=black,
			line width=0.8pt,
			fill=brightyellow,
			text=black,
			rectangle,
			rounded corners=1pt,
			inner sep=2pt
			] (ybox) {\texttt{y=0}};
		\end{tikzpicture}\\[-2.5pt]
		\begin{minipage}{1.0cm}
			\begin{lstlisting}[language=CustomPseudoCode,numbers=none,basicstyle=\tiny\ttfamily]
y := X
X := 0
return y
			\end{lstlisting}
		\end{minipage}
	  };
	  
	  \node[above right=-0.5cm and 2.2cm of state2, align=center] (state3) {
		\begin{tikzpicture}[baseline=(ybox.base)]
			\node[
			draw=black,
			line width=0.8pt,
			fill=brightyellow,
			text=black,
			rectangle,
			rounded corners=1pt,
			inner sep=2pt
			] (ybox) {\texttt{y=0}};
		\end{tikzpicture}\\[-2.5pt]
		\begin{minipage}{1.0cm}
			\begin{lstlisting}[language=CustomPseudoCode,numbers=none,basicstyle=\tiny\ttfamily]
return y
			\end{lstlisting}
		\end{minipage}
	  };
	  
	  \node[below right=-0.2cm and 2.2cm of state2, align=center] (state4) {
		\begin{tikzpicture}[baseline=(ybox.base)]
			\node[
			draw=black,
			line width=0.8pt,
			fill=brightyellow,
			text=black,
			rectangle,
			rounded corners=1pt,
			inner sep=2pt
			] (ybox) {\texttt{y=1}};
		\end{tikzpicture}\\[-2.5pt]
		\begin{minipage}{1.0cm}
			\begin{lstlisting}[language=CustomPseudoCode,numbers=none,basicstyle=\tiny\ttfamily]
return y
			\end{lstlisting}
		\end{minipage}
	  };
	  
	  \node[
		right=0.6cm of state3,
		draw=black,
		line width=0.8pt,
		fill=RedViolet!20,
		text=black,
		diamond,
		aspect=2,
		inner sep=2pt,
		scale=0.7,
		font=\Large
	  ] (resp0) {\texttt{0}};
	  
	  \node[
		right=0.6cm of state4,
		draw=black,
		line width=0.8pt,
		fill=RedViolet!20,
		text=black,
		diamond,
		aspect=2,
		inner sep=2pt,
		scale=0.7,
		font=\Large
	  ] (resp1) {\texttt{1}};
	  
	  \draw[->] (main) -- (state1);
	  
	  \draw[->] (state1) -- node[above] {\begin{tikzpicture}[baseline=(a.base)]\node[draw=black,line width=0.8pt,fill=blue!10,rectangle,rounded corners=1pt,inner sep=2pt] (a) {\texttt{X=0}};\end{tikzpicture} $\to$ \begin{tikzpicture}[baseline=(b.base)]\node[draw=black,line width=0.8pt,fill=blue!10,rectangle,rounded corners=1pt,inner sep=2pt] (b) {\texttt{X=1}};\end{tikzpicture}} node[below] {\begin{tikzpicture}[baseline=(c.base)]\node[draw=black,line width=0.8pt,fill=blue!10,rectangle,rounded corners=1pt,inner sep=2pt] (c) {\texttt{X=1}};\end{tikzpicture} $\to$ \begin{tikzpicture}[baseline=(d.base)]\node[draw=black,line width=0.8pt,fill=blue!10,rectangle,rounded corners=1pt,inner sep=2pt] (d) {\texttt{X=1}};\end{tikzpicture}} (state2);
	  
	  \draw[->] ([yshift=4pt]state2.east) to[out=50,in=180] node[above, sloped] {\begin{tikzpicture}[baseline=(a.base)]\node[draw=black,line width=0.8pt,fill=blue!10,rectangle,rounded corners=1pt,inner sep=2pt] (a) {\texttt{X=0}};\end{tikzpicture} $\to$ \begin{tikzpicture}[baseline=(b.base)]\node[draw=black,line width=0.8pt,fill=blue!10,rectangle,rounded corners=1pt,inner sep=2pt] (b) {\texttt{X=0}};\end{tikzpicture}} (state3.west);
	  \draw[->] ([yshift=-16pt]state2.east) to[out=-50,in=180] node[below, sloped] {\begin{tikzpicture}[baseline=(a.base)]\node[draw=black,line width=0.8pt,fill=blue!10,rectangle,rounded corners=1pt,inner sep=2pt] (a) {\texttt{X=1}};\end{tikzpicture} $\to$ \begin{tikzpicture}[baseline=(b.base)]\node[draw=black,line width=0.8pt,fill=blue!10,rectangle,rounded corners=1pt,inner sep=2pt] (b) {\texttt{X=0}};\end{tikzpicture}} (state4.west);
	  
	  \draw[->] (state3) -- (resp0);
	  \draw[->] (state4) -- (resp1);
	  
	\end{tikzpicture}
	\caption{
	The network system for Listing~\ref{lst:MotivatingExample2NonSer}. 
		Local states show the variable assignments (yellow rectangles \fcolorbox{black}{brightyellow}{\rule{0pt}{7pt}\rule{7pt}{0pt}}) and the remaining code; edges indicate transitions of global states  (blue rectangles \fcolorbox{black}{blue!10}{\rule{0pt}{7pt}\rule{7pt}{0pt}}). 
		Requests and responses appear as
		{\color{ForestGreen!20}$\blacklozenge$} (green) and
		{\color{RedViolet!20}$\blacklozenge$} (red) diamonds, respectively.
		From left to right:  ${\color{ForestGreen}\blacklozenge_{\text{main}}}$ spawns a request with \texttt{[y=0]} and the full program; after yielding, $\delta$ steps with global state \texttt{[X=1]} and local state \texttt{[y=0]}, then updates \texttt{y} based on the global value, returning it as the final response (either {\color{red}$\blacklozenge_0$} or {\color{red}$\blacklozenge_1$}).
		}
		%

	%
	%
\label{fig:code2ExampleNS}
\end{figure}

%% file: sections/3_formal_results.tex
\section{Formal Results}
\label{sec:formal-results}


\subsection{The Algorithm (without Optimizations)}

Given a network system \(\mathcal S= (G, L, \mathit{REQ},  \mathit{RESP}, g_0, \delta, \mathit{req}, \mathit{resp})\) we run the following steps:  

\medskip
	\noindent
	\textbf{Step 1: Serializability automaton.} 
	We define an NFA 
	\(
	\mathcal A_{\mathrm{ser}}(\mathcal S)=(Q,\Sigma,\delta^A,q_0,F)
	\), with
	\( Q=G,\;F=G,\;q_0=g_0,
	\)
	over an alphabet
	\(
	\Sigma=\{({\color{ForestGreen}\blacklozenge_{\mathit{req}}}/{\color{red}\blacklozenge_{\mathit{resp}}})
	\mid {\color{ForestGreen}\blacklozenge_{\mathit{req}}}\in\mathit{REQ},\;
	{\color{red}\blacklozenge_{\mathit{resp}}}\in\mathit{RESP}\}.
	\)
%
%
We let each transition correspond to a request/response pair:
\(
\delta^A \subseteq Q \times \Sigma \times Q,
\quad q \xrightarrow{{\color{ForestGreen}\blacklozenge_{\mathit{req}}}/{\color{red}\blacklozenge_{\mathit{resp}}}} q'
\),
%
iff $\mathcal S$ is in global state $q$ and issues a request
$	{\color{ForestGreen}\blacklozenge_{\mathit{req}}}$, then upon some \textit{full serial execution} it eventually transitions to global state  $q'$ and returns response $	{\color{red}\blacklozenge_{\mathit{resp}}}$.
%
%
%
%
%
%
%
%
	Its language
	\(L(\mathcal A_{\mathrm{ser}}(\mathcal S))\subseteq\Sigma^*\) is exactly the set of serial
	request/response traces.
	Hence, by definition, it holds that applying the Parikh image gives the set of all multisets of request/response pairs obtained by serial executions:
	\(
	\mathsf{Ser}(\mathcal S)
	\;=\;
	\mathsf{Parikh}\bigl(L(\mathcal A_{\mathrm{ser}}(\mathcal S))\bigr)
	\;\subseteq\;\mathbb N^{\Sigma}.
	\)

\begin{tcolorbox}[colback=black!5!white, colframe=black, boxrule=1pt]
	\textbf{Example.} For Listing~\ref{lst:MotivatingExample2NonSer}, the NS in Fig.~\ref{fig:code2ExampleNS} gives rise to the Serial NFA in Fig.~\ref{fig:code2ExampleNFA}.
	A trace of request/response pairs is accepted by the NFA iff some serial execution of the program induces it.
Here, serial runs produce only ({\color{ForestGreen}$\blacklozenge_\text{main}$}/{\color{red}$\blacklozenge_1$}), and the only reachable global state is [\texttt{X=0}].
%
	%
	\end{tcolorbox}
	\begin{figure}[!htbp]
		\centering

		\begin{tikzpicture}[
			->,>=stealth,
			thick,
			node distance=2.5cm,
			state/.style={
				draw=black,
				line width=0.8pt,
				fill=blue!10,
				rectangle,
				rounded corners=1pt,
				inner sep=2pt,
				font=\small
			},
			every node/.style={font=\small}
			]
			\node[state] (X1) {\texttt{X=1}};
			\node[state, right of=X1] (X0) {\texttt{X=0}};
			
			\draw[->] ([yshift=-0.4cm]X0.south) -- (X0.south);
			
			\draw[->] (X1) -- node[above] {${\color{ForestGreen}\blacklozenge_{\mathrm{main}}}/{\color{red}\blacklozenge_1}$} (X0);
			\draw[->] (X0) edge[loop right] node[right] {${\color{ForestGreen}\blacklozenge_{\mathrm{main}}}/{\color{red}\blacklozenge_1}$} (X0);
		\end{tikzpicture}
		
		\caption{Serial NFA of Listing~\ref{lst:MotivatingExample2NonSer}.}
		\label{fig:code2ExampleNFA}
	\end{figure}
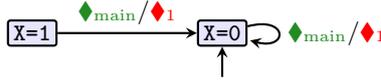

%
%
%
%
	%
	%

\noindent
	\textbf{Step 2: Interleaving Petri net.}
Next, we translate the NS into a Petri net \(N_{\mathrm{int}}(\mathcal S)\). The \textit{non-sink places} of the PN represent either (i) global state assignments, or (ii) local states of in-flight packets. The \textit{sink places} represent request/response pairs of terminated packets.
We define the \textit{transitions} between states to correspond to the \(\delta\),\(req\), and \(resp\) mappings of the NS (the \(req\) transitions can fire without any input tokens in order to correspond to initializing arbitrarily many requests externally).
Finally, we define the initial marking \(M_0\) to be a single token in the place corresponding to the initial global state \(g_0\).
This construction (which is fully formalized in Appendix~\ref{appendix:NS-to-PN-formulation}) guarantees that any reachable marking \(M\) (i.e., \(M_0 \xrightarrow{}^{*} M\)), when projected (\(\pi\)) to the sink places, corresponds to a multiset of  (${{\color{ForestGreen}\blacklozenge_{\mathit{req}}}/{\color{red}\blacklozenge_{\mathit{resp}}}}$) pairs of the NS, as obtained by an interleaving (and vice versa), i.e., \(\mathsf{Int}(\mathcal S)=\{\pi(M) \mid M \in R(N_{\mathrm{int}}(\mathcal S))\}\).

\begin{tcolorbox}[colback=black!5!white, colframe=black, boxrule=1pt]
	\textbf{Example.}
In our running example, the NS gives rise to the PN in Fig.~\ref{fig:code2ExamplePN}, encoding all possible interleavings. The places \textcolor{blue}{ $P_2$} and \textcolor{blue}{$P_3$} represent the global states \textcolor{blue}{[X=1]} and \textcolor{blue}{[X=0]}, respectively, while the places $P_1$, $P_4$, $P_5$, and $P_6$ capture the local states of in-flight requests, i.e., the remaining program code coupled with the assignments to each request’s local variables. Similarly, places \textcolor{red}{$P_7$} and \textcolor{red}{$P_8$} respectively correspond to responses {\color{red}$\blacklozenge_1$} and {\color{red}$\blacklozenge_0$}. Each token either models an active request, a completed request/response pair, or --- when residing in a global-state place --- the current global state of the NS. Finally, transitions implement the network system’s mappings ($\delta/req/resp$): they either spawn a new request (e.g., transition $\textcolor{ForestGreen}{t_1}$, producing {\color{ForestGreen}$\blacklozenge_{\mathrm{main}}$} based on $req$), advance the program by one step (e.g., \(t_2,t_3,t_4,\) and \(t_5\), based on $\delta$),
or return a response (e.g., transitions $\textcolor{red}{t_6}$ and $\textcolor{red}{t_7}$, based on $resp$).
\end{tcolorbox}

%

\begin{figure}[!htbp]
	\centering
	\includegraphics[width=0.7\textwidth]{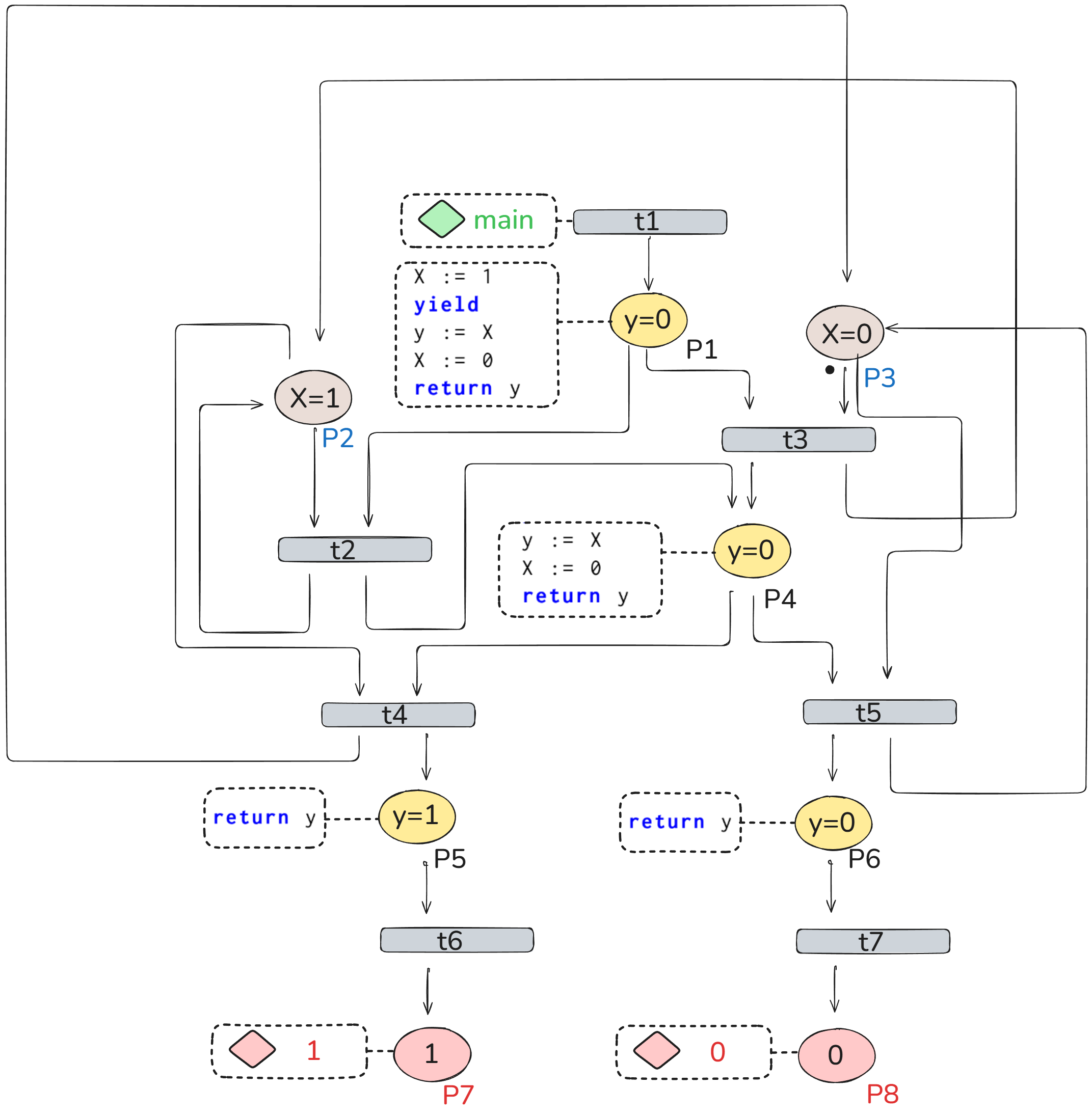}
	\caption{The PN encoding interleaved executions of the program in Listing~\ref{lst:MotivatingExample2NonSer}.}
	\label{fig:code2ExamplePN}
\end{figure}

	\pagebreak
	\noindent
	\textbf{Step 3: Non-serializable set.}  
	Let
	\(\;\mathsf{NonSer}(\mathcal S)=\mathbb N^{|\Sigma|}\setminus \mathsf{Ser}(\mathcal S)\), i.e., all multisets of (${{\color{ForestGreen}\blacklozenge_{\mathit{req}}}/{\color{red}\blacklozenge_{\mathit{resp}}}}$) pairs that \textit{cannot} be obtained via a serial execution of NS $\mathcal S$.
	
	\begin{tcolorbox}[colback=black!5!white, colframe=black, boxrule=1pt]
	\textbf{Example.}
	Regarding the aforementioned program, we automatically generate the following reachability query\footnote{If not for the equality constraints, the problem would have been considered a Petri net \textit{coverability} query, which is easier~\cite{Ra78}.} for the Petri net in Fig.~\ref{fig:code2ExamplePN}, encoding a target semilinear set by imposing the following constraints on the token distribution:
	%
	%
	\[
	\mathcal {F}: \quad P_1 = 0 \wedge 
	\textcolor{blue}{P_2} \ge 0 \wedge \textcolor{blue}{P_3} \ge 0  \wedge P_4 = 0
	\wedge P_5 = 0 \wedge P_6 = 0 \wedge \textcolor{red}{P_7} \ge 0 \wedge \textcolor{red}{P_8} \ge 1.
	\]
	
	This set requires no tokens on $P_1,P_4,P_5,P_6$, at least one token on $\textcolor{red}{P_8}$ (i.e., a response {\color{red}$\blacklozenge_0$}), and any number of tokens on $\textcolor{blue}{P_2},\textcolor{blue}{P_3},\textcolor{red}{P_7}$. 
	%
	\end{tcolorbox}

	
	\noindent
	\textbf{Step 4: Decision \& validation.}
	We ask whether there \textit{exists} a reachable marking \(M\) of \(N_{\mathrm{int}}(\mathcal S)\) such that \(M \models \mathcal {F}\). As \(\mathcal {F}\) encodes \(\mathsf{NonSer}(\mathcal S)\), this is equivalent to a marking \(M\) such that \(	M_0 \xrightarrow{}^{*} M\) and
	\[
	\pi(M) \in \mathsf{Int}(\mathcal S)
	\quad\wedge\quad
	\pi(M)\in \mathsf{NonSer}(\mathcal S).
	\]
	
	\begin{itemize}
		\item [\sat]: yields a counterexample interleaving \(M\) with
		\(\pi(M)\notin \mathsf{Ser}(\mathcal S)\), validated by simulation of the network system \(\mathcal S\).
		
		\item [\unsat]: yields an inductive invariant of
		\(N_{\mathrm{int}}(\mathcal S)\), back-translated to an NS-level proof of serializability 
%
		%
(see an example in Appendix~\ref{appendix:ns-serializable}).
	\end{itemize}

\begin{tcolorbox}[colback=black!5!white, colframe=black, boxrule=1pt]
	\textbf{Example.}
	In our running example, the target semilinear set \(\mathcal {F}\) is, in fact, reachable. For instance, it includes the following marking:
	
	\[
	M^* = \{\textcolor{blue}{P_3}(1),\;\textcolor{red}{P_7}(1),\;\textcolor{red}{P_8}(1)\}
	\]
	
	which is reachable by the PN in Fig.~\ref{fig:code2ExamplePN}. The full firing sequence leading to marking $M^*$ is given in Table~\ref{tab:PetriNetFiringCounterexample} (in Appendix~\ref{appendix:non-serializable-execution-example}).
	%
	%
	Specifically, this reachable marking encodes the outputs $\{{\color{ForestGreen}\blacklozenge_\text{main}}/{\color{red}\blacklozenge_0},{\color{ForestGreen}\blacklozenge_\text{main}}/{\color{red}\blacklozenge_1}\}$ which, indeed, can only be induced by a non-serial execution of Listing~\ref{lst:MotivatingExample2NonSer}.
\end{tcolorbox} 
%
%
%
%
%
	


\medskip
\noindent
\textbf{Complexity analysis.}
The core algorithm reduces serializability checking to Petri net reachability with target semilinear sets. 
Since the serial executions form a regular language (step 1), their Parikh image is effectively semilinear by Parikh's theorem, with size exponential in the NFA.
The interleaving Petri net (step 2) has $O(|G| + (|\mathit{REQ}| \times |L|) + (|\mathit{REQ}| \times |\mathit{RESP}|))$ places and $O(|\mathit{REQ}|\times (1+ |\delta| + |\mathit{RESP}|))$ transitions.
The reachability query (step 3) asks whether the Petri net can reach the complement of a semilinear set, which is decidable but \texttt{Ackermann}-complete~\cite{CzWo22,Le22}.
Without optimizations, even simple examples can generate Petri nets with hundreds of places and exponentially-sized semilinear constraints, making the approach impractical.
Our optimizations (see subsec.~\ref{sec:optimizations}) \textit{drastically} reduce both the Petri net size and the semilinear set complexity, as we elaborate next.

%

\subsection{Optimizations}
\label{sec:optimizations}


We apply four optimizations to the base algorithm to control intermediate blow‐up in the size of both the PN and the constructed semilinear set. 
An extensive empirical evaluation of these optimizations appears in Appendix~\ref{appendix:full_results}.

\medskip
\noindent
\textbf{(1) Bidirectional slicing.}  
When solving Petri net reachability, many places and transitions might be irrelevant to the specific target set~\cite{Ra12}.  
	We slice them before symbolic reasoning by combining forward and backward passes:  
	the forward pass over-approximates the places reachable from $M_0$; and symmetrically,   
	the backward pass traverses in reverse from any place that can influence a target constraint (hence over-approximating the places that can contribute to it).
	We iteratively remove non-forward-reachable and
	non-backward-relevant places and transitions, to a fixed point.  
	Appendix~\ref{appendix:BidirectionalProof} illustrates this (Fig.~\ref{fig:bidirectional_pruning}) and proves soundness (Theorem~\ref{thm:bidirectional-pruning}):


\begin{theorem}[Bidirectional Slicing Soundness]
	\label{thm:bidirectional-pruning}
	Let $N = (P, T, \Pre, \Post, M_0)$ be a Petri net and $S$ a target set.  
	Let $N' = (P',T',\,\Pre|_{P'\times T'},\,\Post|_{P'\times T'},\,M_0|_{P'})$ be the sliced net.  
	Then $S$ is reachable from $N$ iff it is reachable from $N'$.
\end{theorem}
%
%


\medskip
\noindent
\textbf{(2) Semilinear set pruning.}  
A semilinear set $S=\bigcup_{i=1}^m L_i$ with 
	$L_i=\{\,\mathbf{b}_i+\sum_{\mathbf{p}\in P_i}n_p \mathbf{p} \mid n_p\in\mathbb N\}$ may contain redundant 
	period vectors or components.  
	Thus, during construction, we:
	(1) remove any period vector $\mathbf{p}\in P_i$ expressible as a nonnegative combination of $P_i\setminus\{\mathbf{p}\}$; 
	and (2) drop $L_i$ when $L_i\subseteq L_j$ (for \(i \neq j\)).  
	This pruning keeps formulas compact and solver calls tractable.



\medskip
\noindent
\textbf{(3) Generating fewer constraints.}  
When computing the Parikh image of a regular expression as a semilinear set,
	most regex operations can be implemented without an exponential blow-up.
	However, the Kleene star is a notable exception. Given $S=\bigcup_{i=1}^m L_i$,  the Kleene closure $S^\ast$ can be expressed as a semilinear set by: 
	\[
	S^\ast=\bigcup_{I \subseteq \{1,...,m\}} 
	\Big\{\sum_{i \in I} \mathbf{b}_i + \sum_{\mathbf{p} \in \bigcup_{i \in I} (P_i\cup \{\mathbf{b}_i\})} n_p \mathbf{p}\Big\},
	\]
	yielding $2^m$ components. To mitigate this:
	(i) if $L_i=\{\mathbf{b}_i\}$ (period-less component), factor it out, star the rest, then add $\mathbf{b}_i$ as a period;  
	(ii) if $L_i=\{\sum_{\mathbf{p}\in P_i}n_p\mathbf{p}\}$ (zero base), likewise star the rest and add each $\mathbf{p}\in P_i$ as a period vector to the resulting set.  
	Each such case halves the component count and circumvents exponential blow-ups.

\medskip
\noindent
\textbf{(4) Strategic Kleene elimination order.} 
We use \textit{Kleene's algorithm}~\cite{Kl56} to translate the serializability NFA into a regex.
The size of the generated semilinear set is not only impacted by how the
	semilinear set operations are implemented, but also by what \textit{specific} regular
	expression is given as input: a single regular language may be represented by a
	number of equivalent regexes, each of different complexity.
	In particular, as Kleene star can cause a large blow-up in the semilinear set size,
	we are especially sensitive to the \emph{star height} of the generated regex.
	Naive Kleene elimination may introduce many nested stars.  
	We reduce this by strategically choosing to eliminate lower-degree states first:
	\[
	q^*=\arg\min_{q\in Q}\bigl(|\delta^A_{\mathrm{in}}(q)|+|\delta^A_{\mathrm{out}}(q)|\bigr).
	\]

	\smallskip
	\noindent
	As we demonstrate in Appendix~\ref{appendix:full_results}, our optimizations expedite the search procedure and make the representations \textit{significantly} more compact. This, in turn, enables deciding serializability for instances that are otherwise intractable.

%% file: sections/4_implementation.tex
\section{Implementation}
\label{sec:implementation}

\subsection{Code Architecture}

We implemented our approach in \toolname{}~\cite{ArtifactRepository}, a publicly available toolchain 
written mostly in \texttt{Rust}. 
\toolname{} implements an end-to-end serializability checker for a given input program. If the program is serializable, we return a proof thereof; otherwise, if it is not serializable, a counterexample is given to the user for an interleaving that can result in request/response pairs that are unattainable in any serial execution.
Our workflow translates the decidability problem to an equivalent Petri net reachability question (for an unbounded number of tokens), in which (i) the Petri net represents all possible interleavings of the program; and (ii) the reachability query represents a semilinear set (equivalently, a Presburger arithmetic encoding) of all request/response pairs that \textit{cannot} be obtained by any serial execution.
As Petri net reachability is \texttt{Ackermann}-complete~\cite{CzWo22,Le22}, we added various optimizations to expedite the search process, both at the PN level and the property-encoding level.
The pipeline of \toolname{} is depicted in Fig.~\ref{fig:full_program_flow}, and includes:

\begin{enumerate}
	\item \textbf{Input \& parsing.} 
	Our framework receives either a \toolname{} program with the syntax described in~\Cref{sec:problem-definition}, or a \texttt{JSON} file directly encoding a network system. In the case of the former, an additional step takes place, parsing the input 
	to an expression tree that is translated to the equivalent NS.
	
	\item \textbf{Petri net conversion.} The NS is then translated into a Petri net 
which represents all possible interleavings. 
The PN is encoded in the de facto standard \texttt{NET} format, to support off-the-shelf PN model checkers. 
	
	\item \textbf{Semilinear conversion.} 
We generate a semilinear set encoding all non-serializable outputs, via translation of the serialized NFA (e.g., Fig.~\ref{fig:code2ExampleNFA}) to a regex, which is then projected (via the Parikh image) and complemented. 
	At the end of the pipeline, an \texttt{XML}-formatted output encodes a reachability query that encapsulates constraints over the PN token count.

	\item \textbf{Reachability engine.} The PN and the reachability query are fed to a  PN model checker, which combines \textit{bounded model checking} (BMC)~\cite{BiCiClZh99} in search of a counterexample; and \textit{state equation reasoning}~\cite{Mu77} in order to prove non-reachability. 
	%
	In order to expedite the search, ``large'' (PN, query) pairs are replaced with multiple sliced PNs (generated by the reachability engine), each coupled with a sub-query encoding a separate disjunct. 
The disjuncts are solved on the fly, until reaching \texttt{SAT}, in which case, we have a counterexample; otherwise, if all disjuncts are \texttt{UNSAT}, we render the original program as serializable.

	\item \textbf{Proof \& certification.} 
	If \texttt{SAT}, we reconstruct and validate an NS-level counterexample. Otherwise, if all disjuncts are \texttt{UNSAT}, we extract per-disjunct proofs and ``stitch'' these to a single inductive serializability certificate, which we then project to the NS and validate (i) initiation, (ii) inductiveness, and (iii) query refutation.
	%

	\item \textbf{Instrumentation \& logging.} Throughout the pipeline, we record various intermediate representations and performance metrics.  
%
	%
	
	
\end{enumerate}


\begin{figure}[!htbp]
	\centering
	\includegraphics[width=1.0\textwidth]{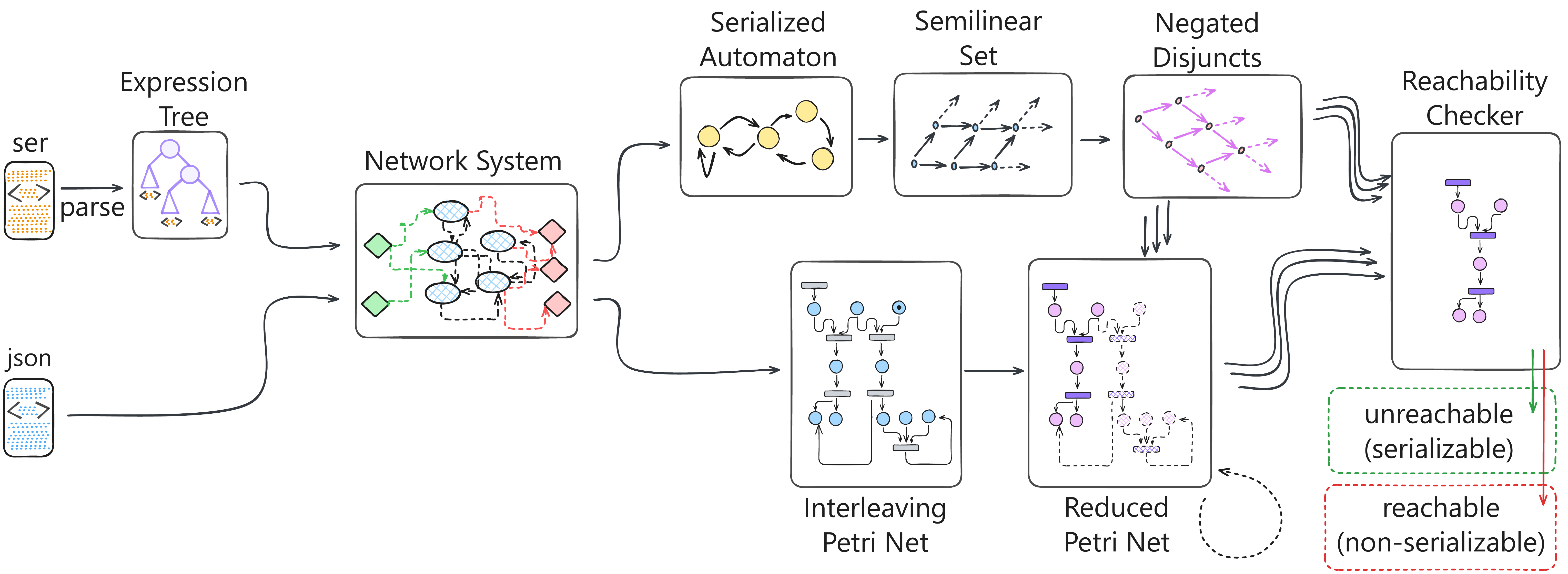}
	\caption{Full program flow 
	(simplified, without backward arrows 
	to the NS level).}
	\label{fig:full_program_flow}
\end{figure}

\subsection{Benchmark Overview} 
\label{subsec:benchmarks}
To the best of our knowledge, ours is the first and only tool to: (i) statically check serializability on \textit{unbounded} programs; and (ii) \textit{prove serializability holds}.
Thus, due to a lack of standard benchmarks for evaluating serializability, we 
assembled a suite of dozens of benchmarks (as part of our accompanying artifact~\cite{ArtifactRepository}).
We believe this is the first benchmark suite for serializability in this setting, aiming to connect \toolname{} programs to practical, real-world analogues.
These include both serializable and non-serializable instances encoded in both \texttt{SER} and \texttt{JSON} formats, and covering a broad range of features, including arithmetic, locks, loops, non-determinism, and more (see an overview of all our benchmarks in Table~\ref{tab:benchmarks-all} of Appendix~\ref{appendix:full_results}).
We note that although the benchmarks themselves are not the main part of the paper, we believe that they have merit on their own, due to their relevance to various real-world systems of interest.
Specifically, we wish to note our suite of benchmarks encoding \textit{network \& system protocols} (see Table~\ref{tab:networking-benchmarks}), which include models of stateful firewalls, BGP routing programs, network monitors, and more --- 
%
as motivated by real-world concurrency problems in this domain. One such example is our \textit{routing-cycle benchmark} in software-defined networks (motivated by~\cite{NaGhSa24}). Another example is our \textit{snapshot isolation benchmark} (see Appendix~\ref{appendix:tour}), which was motivated by a real database bug, namely, duplicate-key errors~\cite{cockroach-issue-14099} in the popular \texttt{CockroachDB} system~\cite{cockroachdb-si-docs}.
In both cases and others, the non-serializable behavior was automatically identified by our toolchain.

%% file: sections/5_evaluation.tex
\section{Evaluation}
\label{sec:evaluation}


\noindent
\textbf{Experimental setup.}
All experiments were run on a Lenovo ThinkPad P16s, with 16 AMD CPU cores and 64 GB of RAM, running Ubuntu 24.04.2.
%
%
We use \texttt{SMPT}~\cite{AmDa23} (built upon \texttt{Z3}~\cite{DeBj08}) as our backend Petri net model checker.
Our code and benchmarks are publicly available~\cite{ArtifactRepository}.

\begin{wraptable}{r}{0.52\textwidth}
	\vspace{-1.2\baselineskip} 
	\centering
	\small
	\input{tables/average_and_mean_values_of_big_table_simplified_compact.tex}
	\caption{Runtime for generating certificates (\texttt{cert.}) and the overall runtime (\texttt{total}), including for validation.}
	\label{tab:stats-summary}
	\vspace{-0.6\baselineskip} 
\end{wraptable}
\medskip
\noindent
\textbf{Results.}
We ran \toolname{} on all 47 benchmarks, out of which 27 are serializable, and the remaining 20 are non-serializable. 
For each benchmark, we measured the time for deciding the reachability query, as well as the overall time, including validation of the invariant proof (if serializable) or of the counterexample (if not serializable). These experiments ran in parallel on $16$ cores with all four optimizations and a \texttt{TIMEOUT} threshold of $500$ seconds.
Within this time limit, \toolname{} solved 26 of the 27 serializable benchmarks and 19 of the 20 non-serializable benchmarks (see summary in Table~\ref{tab:stats-summary} and the full results in Table~\ref{tab:benchmarks-all} of the appendix).
%
%
The \textit{median} total runtime was $1{,}909$ ms across all benchmarks, and $2{,}238.5$ ms ($830$ ms) when  solely focusing on serializable (non-serializable) benchmarks.
%
The \textit{average} total runtime was $32{,}898.38$ ms across all benchmarks, and $25{,}530.69$ ms ($42{,}980.47$ ms) when solely focusing on serializable (non-serializable) benchmarks.
We also observe a clear runtime split based on serializability: 
among non-serializable benchmarks, counterexample generation takes much longer than validation, and dominates the overall runtime; whereas among serializable benchmarks the validation time dominates the overall runtime. This is not surprising, as validating a given counterexample only requires a polynomial-time simulation of the network system to confirm its feasibility.

\begin{table}[htbp]
	\centering
	\input{tables/big_table_summary_only_networking_part.tex}
	\caption{Overview of benchmarks from the \textit{network \(\&\) system protocols} category. 
}
\label{tab:networking-benchmarks}
\end{table}






%% file: tables/average_and_mean_values_of_big_table_simplified_compact.tex
	\begin{tabular}{l r r r r r}
		\toprule
		& \multicolumn{2}{c}{average (ms)} 
		& \multicolumn{2}{c}{median (ms)} \\
		\cmidrule(lr){2-3} \cmidrule(lr){4-5}
		Category
		& \shortstack{\texttt{cert.}}
		& \texttt{total}
		& \shortstack{\texttt{cert.}}
		& \texttt{total} \\
		\midrule
		\textcolor{ForestGreen}{Serializable}      &   2{,}273 &  25{,}531 &  1{,}178 &  2{,}239 \\
		\textcolor{red}{Not serializable}  &  42{,}076 &  42{,}980 &   773 &   830 \\
           &    &    &     &    \\
		All               &  19{,}079 &  32{,}898 &   773 &  1{,}909 \\

		\bottomrule
	\end{tabular}

%% file: tables/big_table_summary_only_networking_part.tex
\begin{table}[H] 
	\centering
	\small
	\setlength{\tabcolsep}{5pt}
	\renewcommand{\arraystretch}{0.9}
	\begin{tabular*}{\textwidth}{@{\extracolsep{\fill}}%
			p{2.5cm}     
			c          
			c c c c c c 
			c       
		}
		\toprule
		\textbf{Benchmark}
		& \textbf{Serializable}
		& \multicolumn{6}{c}{\textbf{Features}}
		& \textbf{Runtime} \\
		\cmidrule(lr){1-1} \cmidrule(lr){2-2} \cmidrule(lr){3-8} \cmidrule(lr){9-9}
		& 
		& If & While & \texttt{?} & Arith & Yield & Multi-req & (ms)
		\\
		\midrule
		\texttt{banking (g1)}          & \xmark      & \cmark & \cmark &        & \cmark & \cmark & \cmark &  74{,}539 \\
		\texttt{banking (g2)}          & \greencmark & \cmark & \cmark &        & \cmark & \cmark & \cmark & \texttt{TIMEOUT} \\
		\texttt{routing (g3)}      & \xmark      & \cmark & \cmark & \cmark & \cmark & \cmark & \cmark &  20{,}954 \\
		\texttt{monitor (g4)}       & \xmark      & \cmark & \cmark & \cmark & \cmark & \cmark & \cmark &  7{,}047 \\
		\texttt{monitor (g5)}       & \greencmark & \cmark & \cmark & \cmark & \cmark & \cmark & \cmark &  12{,}324 \\
		\texttt{firewall (g6)}& \xmark      & \cmark &        & \cmark & \cmark & \cmark &       &  8{,}285 \\
		\texttt{firewall (g7)}& \greencmark & \cmark &        & \cmark & \cmark &       &       & 252{,}752 \\
		\midrule
		\bottomrule
	\end{tabular*}
\end{table}

%% file: sections/6_related_work.tex

\section{Related Work}
\label{sec:related-work}

\textbf{Theoretical results.}
Serializability (or \emph{atomicity}) was
first introduced by Eswaran et al.~\cite{EsGrKoTr76},
later motivating Herlihy and Wing's~\cite{HeWe87,HeWi90} similar notion of \emph{linearizability} for concurrent data structures. 
%
The \emph{membership problem} --- deciding if a specific interleaving is serializable --- is \texttt{NP}-complete~\cite{Pa79}, a result that was later extended to linearizability~\cite{GiKo97}, as well as to other consistency models~\cite{BiEn19}. 
%
The \emph{correctness problem} --- whether \emph{all} executions satisfy this criterion --- is \texttt{EXPSPACE} when threads are bounded~\cite{AlMcPe96} and undecidable otherwise~\cite{BoEmEnHa13}, though decidable for bounded-barrier programs (and hence, for serializability in our setting). Bouajjani et al.~\cite{BoEmEnHa18} further show that unbounded-thread linearizability for certain ADTs reduces to VASS coverability in \texttt{EXPSPACE}~\cite{Ra78}. The \toolname{} toolchain is, to our knowledge, the first to implement Bouajjani et al.’s serializability algorithm~\cite{BoEmEnHa13}, adapting it to distributed transactions, extending it with a proof certificate mechanism, and scaling it with various optimizations to models of real-world programs. 

\medskip
\noindent
\textbf{Model checking and runtime verification.}
Runtime checks for serializability and conflict/view‐serializability were proposed by Wang and Stoller~\cite{WaSt06a,WaSt06b}.  
TLA logic~\cite{La94} can express various serializability forms~\cite{CoOlPnTuZu07}, however, such approaches~\cite{SoVaVi20,Ho24} remain restricted to bounded systems due to finite‐state tools (\texttt{TLC}, \texttt{Apalache}~\cite{YuMaLa99,KoKuTr19}).  
Heuristic or enumeration-based model checkers include \texttt{Line-up}~\cite{BuDeMuTa10} (built upon \texttt{CHESS}~\cite{MuQaBaBaNaNe08}), \texttt{LinTSO}~\cite{BuGoMuYa12}, \texttt{Violat}~\cite{EmEn19} and its schema precursor~\cite{EmEn18}, bridge‐predicate methods~\cite{BuNeSe11,BuSe09}, and PAT‐based refinement checking~\cite{LiChLiSuZhDo12,SuLuDoPa09,LiChLiSu09,Zh11}.  
Recent work includes \texttt{RELINCHE} for bounded linearizability~\cite{GoKoVa25}, \texttt{CDSSpec}~\cite{OuDe17} (for \texttt{C/C++11}), \texttt{Lincheck}~\cite{KoDeSoTsAl23}, and \texttt{SAT}‐based~\cite{BiHeCa09,ViWeMa15} approaches~\cite{BuAlMa07}.  
Symbolic testing~\cite{EmEnHa15} can expose violations of observational refinement~\cite{FiOhRiYa10,BoEmCoHa15}.  
Other checkers, e.g. \texttt{SPIN}/\texttt{PARGLIDER}~\cite{Fl04,VeYaYo09,Ho97,VeYa08}, depend on explicit linearization points, which are difficult to determine~\cite{VeYaYo09}.  
Overall, existing methods are typically incomplete, bounded, or assume prior knowledge.  
By contrast, our method covers unbounded threads and uniquely produces serializability certificates.

\medskip
\noindent
\textbf{Static analysis.}
Static methods prove linearizability for bounded~\cite{AmRiReSaYa07,MaLeSaRaBe08} and unbounded systems~\cite{BeLeMaRaSa08,Va09,Va10}, but usually rely on heuristics or annotated linearization points (e.g.~\cite{DrPe14}).  
	Lian and Feng~\cite{LiFe13} propose a logic for non‐fixed points.  
	However, annotation‐based analyses~\cite{OhRiVeYaYo10,ZhPeHa15,AbJoTr16} may be inconclusive, as failures can stem from incorrect annotations rather than from true violations~\cite{BoEmCoHa15}.

\medskip
\noindent
\textbf{Manual proofs and additional approaches.}
Tasiran~\cite{Ta08} proves serializability for \texttt{Bartok-STM}, while Colvin et al.~\cite{CoGrLuMo06} use I/O automata for list‐set linearizability.  
	Simplifications exist for specific data structures~\cite{BoEmEnMu17,FeEnMoRiSh18}.
	Other notable linearizability results include Wing and Gong's~\cite{WiGo93} on unbounded FIFO/priority queues, Chakraborty et al.~\cite{ChHeSeVa15} on queues, and Cerný et al.’s \texttt{CoLT} for linked heaps (which is complete only under bounded threads~\cite{CeRaZuChAl10}).  
	Bouajjani et al.~\cite{BoEnWa17} introduce a recursive priority‐queue violation detector, akin to their stack/queue methods~\cite{BoEmEnHa18}.
%
Other strategies include testing~\cite{WiGo93,PrGr12,PrGr13,EmEn17,Lo17}, theorem proving~\cite{CoDoGr05,DeScWe11}, and the use of additional verification frameworks~\cite{BoEmEnMu17,FeEnMoRiSh18,EnKo24}.

%% file: sections/7_discussion.tex
\section{Discussion}
\label{sec:discussion}


%

\subsection{Limitations}
While our approach advances the state of the art in verifying unbounded serializability, several limitations remain.
First, the underlying Petri net reachability problem has \texttt{Ackermann}-complete complexity~\cite{CzWo22,Le22}, causing our tool to time out on some complex benchmarks. 
Second, our current implementation relies on \texttt{SMPT}~\cite{AmDa23}, which may fail to find proofs even when they exist, limiting completeness.
Third, our network system model assumes a simple request/response pattern and cannot model more complex interactions, such as streaming, callbacks, or partial responses.
%
Finally, \toolname{} targets finite-state programs: each request must have finite local state and the program must induce finite global state (with an unbounded number of requests). Thus, applying \toolname{} to real systems requires that executions generate only a finite \textit{reachable} state space, in order for the NS construction to terminate.
This setting is akin to model checkers such as \texttt{PRISM}~\cite{KwNoPa02} and \texttt{STORM}~\cite{DeJuKaVo17}.
%
%

\subsection{Future Work}
\textbf{Additional optimizations.}
To improve scalability, we are adapting \textit{polyhedral reductions}~\cite{AmBeDa21,amat2022polyhedral}, a form of structural reduction~\cite{Be87,BeLeDa20} $(N_1, m_1) \vartriangleright_E (N_2, m_2)$ where $N_2$ is a simpler Petri net and $E$ allows reconstruction of $N_1$’s state space. This would allow verification on the reduced net, with proofs lifted back to the original one.
Moreover, we believe a further avenue for optimization lies in using \textit{approximations} to decide serializability. Our approach already leverages this idea via the underlying model checker, which employs the \textit{state equation} abstraction~\cite{Mu77} to over-approximate the reachable state space. We expect that additional approximation-based techniques could yield further scalability gains.
Finally, other potential optimizations involve short-circuiting steps in our algorithmic pipeline. For instance, we currently generate the reachability query $\mathcal{F}$ in three stages: (i) translating the serial NFA into a regular expression via Kleene’s theorem; (ii) translating the regular expression into a semilinear set using Parikh’s construction; and (iii) complementing the resulting semilinear set.
However, there are techniques (e.g., Verma et al.~\cite{VeSeSc05}) that \textit{directly} compute the Parikh image of an automaton. We did not adopt this approach because, although its construction is linear in size, it relies heavily on Boolean logic, which we found \texttt{ISL} (the standard integer set library) handles poorly in practice.

\medskip
\noindent
\textbf{Proof assistants.}
Another natural next step is to formalize certificate checking in a proof assistant (e.g., \texttt{Rocq}~\cite{Rocq}). This would entail (i) developing a verifier for the PN invariants we use; and (ii) proving theorems that connect invariant validity to serializability.
Specifically, this would likely require extending existing tactics such as \texttt{LIA} (Linear Integer Arithmetic), which currently does not support full Presburger arithmetic, as required by our logic.

\subsection{Applicability to Real-World Programs}
Real-world SDN programs typically satisfy our finite-reachable-state requirement due to bounded end-host buffers and limited switch memory. Moreover, we anticipate that \texttt{P4} programs~\cite{BoDaGiIzMcReScTaVaVaWa14} can be translated to \toolname{} based on the following high-level mappings:
(i) packets to requests, (ii) switch registers to global variables, (iii) packet header fields to local variables, and (iv) packet forwarding to yielding. 
%
This translation motivated us to evaluate our toolchain on programs modeling \textit{stateful firewalls}~\cite{KiLiZhKiLeSeSe20,HoLaArReWa22}.
Furthermore, we believe our framework is applicable beyond SDNs. Specifically, our NS model abstracts distributed state with message-passing/RPC-style concurrency, which aligns naturally with database transactions. One such example is our \textit{snapshot-monitoring benchmark} (see Appendix~\ref{appendix:tour}). 

\subsection{Conclusion}
We present the first end-to-end framework that automatically verifies serializability for unbounded concurrent systems and generates proof certificates thereof.
Our approach bridges theory and practice, with the following key contributions:
%
(1) formalizing serializability for network systems, (2) implementing the decision procedure with proof generation, (3) developing optimizations that reduce complexity by orders of magnitude, and (4) demonstrating feasibility on various benchmarks inspired by real-world systems.

%% file: sections/8_appendix_tour.tex
\appendix

\section{Tour of Examples}
\label{appendix:tour}

Next, we will walk through a series of examples, in varying levels of complexity. Each example will demonstrate different aspects of serializable vs. non-serializable programs.
The first examples are relatively basic, while the last examples have higher complexity and are motivated by real-world programs, e.g., BGP routing policy updates.
Each thread is spawned any number of times (and at any point in time) by a \textit{request} from the user, marked {\color{ForestGreen}$\blacklozenge_\text{req}$}. The request executes, and eventually returns a \textit{response}  {\color{red}$\blacklozenge_\text{resp}$}.
For instance, in the three examples presented in \Cref{sec:introduction} (Listings~\ref{lst:MotivatingExample1Ser},~\ref{lst:MotivatingExample2NonSer}, and~\ref{lst:MotivatingExample3Ser}), there is a single type of request {\color{ForestGreen}$\blacklozenge_\text{main}$} and (up to) two types of responses {\color{red}$\blacklozenge_0$}, {\color{red}$\blacklozenge_1$}.
We analyze serializability through the lens of such {\color{ForestGreen}$\blacklozenge_\text{req}$}/{\color{red}$\blacklozenge_\text{resp}$} pairs. Specifically, the programs in Listings~\ref{lst:MotivatingExample1Ser} and~\ref{lst:MotivatingExample3Ser} only induce pairs of the type {\color{ForestGreen}$\blacklozenge_\text{main}$}/{\color{red}$\blacklozenge_1$}, while the program in Listing~\ref{lst:MotivatingExample2NonSer} can also induce {\color{ForestGreen}$\blacklozenge_\text{main}$}/{\color{red}$\blacklozenge_0$}, as formulated by our Network System framework (see \Cref{sec:problem-definition}). 
We depict global variables with upper-case characters, while local variables (for each request) are depicted with lower-case ones.
Unless explicitly stated otherwise, all global and local variables are initialized to 0.
The symbol \texttt{?} depicts a nondeterministic choice between 0 and 1. All other constructs (\texttt{while}, \texttt{yield}, and \texttt{if}) have their standard interpretation, and are based on the \toolname{} semantics covered in Appendix~\ref{appendix:ser-semantics}.

\subsection{Example 1}


We start with a basic example, describing a single request {\color{ForestGreen}$\blacklozenge_\text{A}$}, a single local variable (\texttt{x}) per request, and a single global variable (\texttt{FLAG}) shared among all in-flight requests. 
In Listing~\ref{lst:BasicSer}, an in-flight request assigns to \texttt{x} the value of \texttt{FLAG} (hence, initially, \texttt{[x:=0]}). Then, the request non-deterministically chooses whether to \texttt{yield} or to flip the value of \texttt{x}. Subsequently, \texttt{FLAG} is assigned 1 and the value of \texttt{x} is returned as the response to request {\color{ForestGreen}$\blacklozenge_\text{A}$}. 
Note that the presence of the \texttt{else} branch renders the program serializable, as intuitively, 
for any interleaving that modifies \texttt{x} via the \texttt{if} branch, there exists a corresponding serial execution in which the \texttt{else} branch is taken, yielding an equivalent outcome.
However, this changes in  Listing~\ref{lst:BasicNonSer},
%
%
%
%
%
%
%
%
%
in which there is no \texttt{else} branch --- an update that makes the program non-serializable.
Now, any serial execution will have \textit{at most one} pair of {\color{ForestGreen}$\blacklozenge_\text{A}$}/{\color{red}$\blacklozenge_0$} (this is in fact the first request, returning the original zero-initialized value of \texttt{FLAG}).
As the first request also assigns \texttt{[FLAG:=1]} before terminating, any subsequent request in a serial run will assign \texttt{[x:=1]} and hence, will return only responses of {\color{red}$\blacklozenge_1$}. 
%
%
%
%
%
%
%
%
However, given that the first request can also \texttt{yield}, it is possible for another request to concurrently run the program after the first request yields and before it resumes. This, in turn, will allow more than one request to assign \texttt{[x:=0]}, and hence, for example, we can obtain \textit{multiple} {\color{ForestGreen}$\blacklozenge_\text{A}$}/{\color{red}$\blacklozenge_0$} pairs. Thus, Listing~\ref{lst:BasicNonSer} is not serializable.
%
%
%
%
%



\noindent
\begin{minipage}[t]{0.45\textwidth}
	\begin{lstlisting}[caption={Serializable},
		label={lst:BasicSer},numbers=none]
request A: 
    x := FLAG
    if (?):
        yield
    else:
        x := 1 - x
    FLAG := 1
    return x
	\end{lstlisting}
\end{minipage}
\hfill
\begin{minipage}[t]{0.45\textwidth}
	\begin{lstlisting}[caption={Not serializable},
	label={lst:BasicNonSer},numbers=none]
request A: 
    x := FLAG 
    if (?): 
        yield
    // no else
	
    FLAG := 1 
    return x
		\end{lstlisting}
\end{minipage}%

\subsection{Example 2}

The following program pairs have a single global variable (\texttt{X}), and two requests --- {\color{ForestGreen}$\blacklozenge_\text{incr}$} which increments \texttt{X} by 1, and {\color{ForestGreen}$\blacklozenge_\text{decr}$} which decrements \texttt{X} by 1. Both programs have loops that guarantee that \texttt{X} will always be between 0 and 3, otherwise the \texttt{while} loop will yield ad infinitum. Both requests return the value of \texttt{X} after updating it.
In the first case, Listing~\ref{lst:FredSer} presents a serializable program, due to the absence of any \texttt{yield} between the increment/decrement of \texttt{X} and its return. Equivalently, in each of the requests, the update of \texttt{X} and the returned value can be thought of as \textit{a single atomic execution}.
However, in Listing~\ref{lst:FredNonSer}, we add an additional \texttt{yield} (and a local variable \texttt{y}), occurring in each of the requests, between the update of \texttt{X} and its return.
This change allows requests of the same type to update \texttt{X} to the same value ---  resulting in outputs such as
$\{{\color{ForestGreen}\blacklozenge_\text{incr}}/{\color{red}\blacklozenge_\text{1}},{\color{ForestGreen}\blacklozenge_\text{incr}}/{\color{red}\blacklozenge_\text{2}},{\color{ForestGreen}\blacklozenge_\text{incr}}/{\color{red}\blacklozenge_\text{3}},{\color{ForestGreen}\blacklozenge_\text{decr}}/{\color{red}\blacklozenge_\text{2}},{\color{ForestGreen}\blacklozenge_\text{decr}}/{\color{red}\blacklozenge_\text{2}}\}$ which cannot be obtained in any serial execution.

%
%
%
 %


\noindent
\begin{minipage}[t]{0.45\textwidth}
	\begin{lstlisting}[caption={Serializable},
		label={lst:FredSer},numbers=none]
request incr: 
    while (X == 3):
        yield

    X := X + 1
    return X		

request decr: 
    while (X == 0): 
        yield

    X := X - 1
    return X
		\end{lstlisting}
\end{minipage}
\hfill
\begin{minipage}[t]{0.45\textwidth}
	\begin{lstlisting}[caption={Not serializable},
		label={lst:FredNonSer},numbers=none]
request incr:
    while (X == 3):
        yield
    y := X
    yield
    X := y + 1
    return X		

request decr: 
    while (X == 0):
        yield
    y := X
    yield
    X := y - 1
    return X
		\end{lstlisting}
\end{minipage}
	
\subsection{Example 3}

The next example (see Listing~\ref{lst:ComplexWhileNonSer}) has a global variable \texttt{X} and, for each in-flight request, a local variable \texttt{i}. The {\color{ForestGreen}$\blacklozenge_\text{flip}$} request flips \texttt{X} (initialized to 0); the {\color{ForestGreen}$\blacklozenge_\text{main}$} request attempts to decrement \texttt{i} five times.
%
%
Any serial execution cannot induce a response {\color{red}$\blacklozenge_1$}, 
as it will have a single request in-flight, with \texttt{X} being either 0 or 1. Thus, exactly one of the \texttt{while} loops will run indefinitely, prohibiting any {\color{ForestGreen}$\blacklozenge_\text{main}$}/{\color{red}$\blacklozenge_1$} pairs.
%
%
To prove that the program is not serializable, we show that an interleaving \emph{can} result in a non-empty set of outputs. 
%
%
Specifically, given at least \texttt{[i=5]} interleavings of in-flight {\color{ForestGreen}$\blacklozenge_\text{flip}$} requests, it is possible for a {\color{ForestGreen}$\blacklozenge_\text{main}$} request to terminate and bypass all \texttt{while} loops, something that cannot occur in any serial execution.
%

\vspace{1em}  
\begin{wrapfigure}{r}{0.38\textwidth}
	  \vspace{-1em}  
	\centering
	\begin{lstlisting}[caption={Not serializable},label={lst:ComplexWhileNonSer},numbers=none]
request flip: 
    X := 1 - X 

request main:
    i := 5
    while (i > 0):
        while (X == 0):
            yield
        while (X == 1):
            yield
        i := i - 1

    return 1        
	\end{lstlisting}
\vspace{-0.5em}  
\end{wrapfigure}
\vspace{1em}  




%
%
%
%

\subsection{Example 4}

We illustrate a simple banking system inspired by Chandy and Lamport’s distributed snapshot algorithm~\cite{ChLa85}.  The system manages a client’s funds across multiple accounts; we use two accounts, \texttt{A} and \texttt{B}, but the same pattern extends to any number of accounts.  
Each {\color{ForestGreen}$\blacklozenge_{\mathit{transfer}}$} request transfers \texttt{\$50} from \texttt{A} to \texttt{B}, and each {\color{ForestGreen}$\blacklozenge_{\mathit{interest}}$} request adds an interest rate of \texttt{t\%} to each account (we set \texttt{[t=100\%]} for simplicity).
Both requests return the combined total \texttt{[A + B]}.
In every serial execution with one {\color{ForestGreen}$\blacklozenge_{\mathit{interest}}$} request, and any number of {\color{ForestGreen}$\blacklozenge_{\mathit{transfer}}$} requests, the total balance satisfies the invariant
\texttt{[A$_{\text{after}}$ + B$_{\text{after}}$ = (1 + t\%) \,$\bigl($A$_{\text{before}}$ + B$_{\text{before}}\bigr)$]}
.
%
%
Although the individual balances of \texttt{A} and \texttt{B} depend on the serial order, the \textit{combined} sum always reflects exactly one application of the interest rate.
However, non-serial interleavings can violate this invariant.  For instance, if a {\color{ForestGreen}$\blacklozenge_{\mathit{transfer}}$} request deducts \texttt{\$50} from \texttt{A} (resulting in \texttt{[50,50]}) and then yields, then an {\color{ForestGreen}$\blacklozenge_{\mathit{interest}}$} request may double both balances to \texttt{[100,100]} before the transfer resumes --- resulting in \texttt{[100,150]} and a missing \texttt{\$50}.  By contrast, any serial ordering of these two operations yields \texttt{[A + B = (100+50)$\times$2 = 300]}, with final states \texttt{[150,150]} or \texttt{[100,200]} depending on which request runs first.
%
%
Listing~\ref{lst:BankSer} has a serial version of this banking system (without \texttt{yield}), and Listing~\ref{lst:BankNonSer} includes \texttt{yield} statements between the adjustment of accounts \texttt{A} and \texttt{B} (we note that this is motivated by real-world systems in which accounts can be sharded and partitioned across different nodes).

\noindent
\begin{minipage}[t]{0.45\textwidth}
	\begin{lstlisting}[caption={Serializable},
		label={lst:BankSer},numbers=none]
A := 100, B := 50

request transfer: 
    // transfer $50
    A := A - 50
    // no yield
    B := B + 50
    return A + B
	
request interest: 
    // add a 100% interest
    A := A + A
    // no yield
    B := B + B
    return A + B      
			\end{lstlisting}
\end{minipage}
\hfill
\begin{minipage}[t]{0.45\textwidth}
	\begin{lstlisting}[caption={Not serializable},
		label={lst:BankNonSer},numbers=none]
A := 100, B := 50
	
request transfer: 
    // transfer $50
    A := A - 50
    yield
    B := B + 50
    return A + B

request interest: 
    // add a 100% interest
    A := A + A
    yield
    B := B + B
    return A + B
      		\end{lstlisting}
\end{minipage}

\subsection{Example 5}

The following example is motivated by~\cite{NaGhSa24} and demonstrates how reasoning about serializability corresponds to correctness of routing policies in software‐defined networks (SDNs). In an SDN, switches not only forward packets but can also be programmed in domain‐specific languages (e.g., \texttt{P4}~\cite{BoDaGiIzMcReScTaVaVaWa14}). At runtime, a centralized controller node can adjust the global network policy by periodically sending control packets to each switch, causing it to adjust its routing policy.
An instance of a simple network with two competing policies is shown in Fig.~\ref{fig:BgpRoutingPolicies}. This network consists of four nodes (numbered \texttt{0} through \texttt{3}), with the two middle nodes --- node \texttt{1} (labeled \texttt{WEST}) and node \texttt{2} (labeled \texttt{EAST}), serving as ingress points from where traffic nondeterministically enters the network. The controller selects one of two policies: a \textcolor{NavyBlue}{blue} policy, which routes traffic from West to East, or an \textcolor{darkorange}{orange} policy, which routes it in the opposite direction.

%
%
%
\begin{wrapfigure}{r}{0.45\textwidth}  
	\centering
	\includegraphics[
	width=\linewidth,
	trim=10 15 15 5,   
	clip
	]{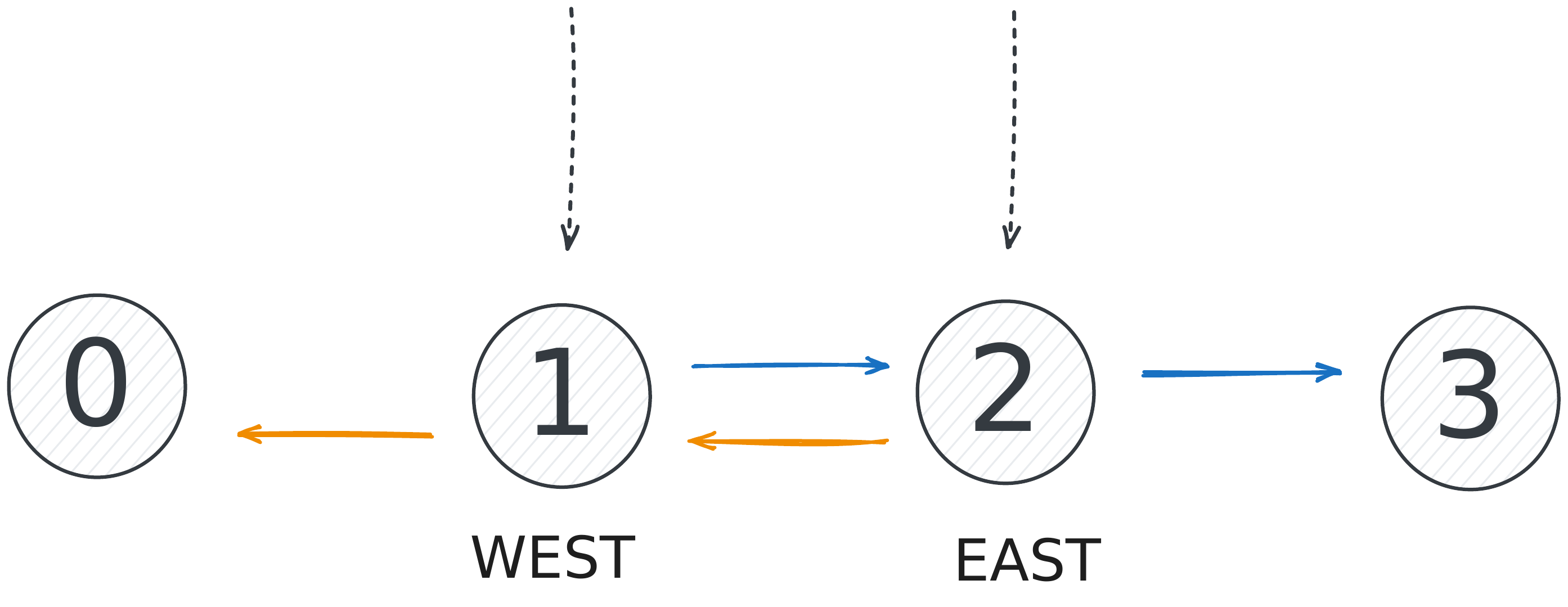}
	\caption{Two routing policies.}
	\label{fig:BgpRoutingPolicies}
\end{wrapfigure}
%
%
This SDN-controlled routing policy is realized in the pseudo-code in Listing~\ref{lst:BgpNonSerializable}.
The program includes a single global variable \texttt{B}, indicating whether the current routing policy is \textcolor{NavyBlue}{blue} (\texttt{[B=1]}) or \textcolor{darkorange}{orange} (\texttt{[B=0]}).
The program has three types of requests:
%
	(i)
	{\color{ForestGreen}$\blacklozenge_\text{policy\_update}$}:
 represents a controller  update, which nondeterministically decides whether to update the policy (i.e., flip the value of  variable \texttt{B}) or not;
%
(ii)
	{\color{ForestGreen}$\blacklozenge_\text{route\_west}$}:
	 a request representing a packet entering the network from the \texttt{WEST} node; and 
%
(iii)
{\color{ForestGreen}$\blacklozenge_\text{route\_east}$}: a request representing a packet entering the network from the \texttt{EAST} node.
	%
%

\begin{center}
\begin{minipage}[!htbp]{1.0\textwidth}
	\begin{lstlisting}[caption={BGP routing (not serializable)},label={lst:BgpNonSerializable},numbers=none]
 request policy_update:
     if (?): // nondeterministically 1 or 0
         B := 1  // blue policy 
     else:
         B := 0 // orange policy
		
 request route_west:
     current := 1 // initial node
     while (current == 1) or (current == 2): // still routing        
         if (current == 1): // west (switch 1)
             if (B == 1): // blue policy
                 current := 2
             else: // orange policy
                 current := 0
         if (current == 2): // east (switch 2)
             visited_east := 1
             if (B == 1): // blue policy
                 current := 3
             else: // orange policy
                 current := 1
         yield
     return current + current + visited_east
     
 request route_east: ... // dual case      
		\end{lstlisting}
\end{minipage}
\end{center}

Each of the routing requests represents a single packet entering the network. The request includes a local \texttt{current} variable representing the index of the current node visited. This variable is initialized as the ingress node value and is updated to emulate the chosen routing path. There is also a \texttt{visited\_east} variable (or a \texttt{visited\_west} variable, depending on the request in question).
The return value of the {\color{ForestGreen}$\blacklozenge_\text{route\_west}$} requests is the sum \texttt{[current+current+visited\_east]}, an identifier encoding all possible \texttt{(current\_switch, visited\_east)} pairs.
The program is not serializable, as witnessed by an interleaving that can give rise to a final return value of \texttt{[current+current+visited\_east=1]} (due to \texttt{[current=0]} and \texttt{[visited\_east=1]}). This represents a \textit{routing cycle} in the network, which is possible only when there is an interleaving between a control packet ({\color{ForestGreen}$\blacklozenge_\text{policy\_update}$}) and a routing packet (e.g., {\color{ForestGreen}$\blacklozenge_\text{route\_west}$}). Specifically, this occurs when a request has already been spawned and has begun routing based on the previous policy, then yields, and eventually returns after the policy was flipped based on another control packet --- hence resulting in a routing cycle.
More formally, this is conveyed by response values that represent these cycles and are obtained only via non-serial executions. For example, acyclic routes of this request have either a return value of 0 (in the case of [\texttt{current=0}, \texttt{visited\_east=0}]) or 7 (in the case of [\texttt{current=3}, \texttt{visited\_east=1}]).
Dually, routing cycles could also occur in the case of {\color{ForestGreen}$\blacklozenge_\text{route\_east}$} interleavings.

\subsection{Example 6}
\label{app:snapshot-isolation-example}

The next program captures serializability through the lens of the \textit{snapshot isolation} consistency model, which is used in various real-world database systems, including \texttt{PostgreSQL}~\cite{postgresql-transaction-iso} and \texttt{CockroachDB}~\cite{cockroachdb-si-docs}, and has been linked to real‐world anomalies (e.g., duplicate‐key errors in the latter~\cite{cockroach-issue-14099}).
The depicted program has two nodes (represented by the global variables \texttt{N$_1$} and \texttt{N$_2$}) which monitor ongoing traffic in the network, and are originally both active, as indicated by their initial values: \texttt{[N$_1$=1]}, \texttt{[N$_2$=1]}.
The {\color{ForestGreen}$\blacklozenge_\text{main}$} request takes a snapshot of the system, i.e., locally records the current activation status of each of the two nodes.
Then, in the first request, and in any future ones in which both nodes are active, each in-flight request non-deterministically decides which of the two nodes to deactivate, i.e., set \texttt{[N$_i$:=0]}, for maintaining overall energy efficiency.
The {\color{ForestGreen}$\blacklozenge_\text{main}$} request eventually returns the current sum of active nodes in the system.
In order for the system to emulate multiple non-trivial interleavings, our setting also includes two additional requests, {\color{ForestGreen}$\blacklozenge_\text{activate\_n1}$} and {\color{ForestGreen}$\blacklozenge_\text{activate\_n2}$}, which activate nodes \texttt{N$_1$} and \texttt{N$_2$}, respectively.
We note that the program is not serializable due to the \texttt{yield} statement that appears immediately after the recorded snapshot of the node activation status. One such example of a non-serializable behavior occurs when two {\color{ForestGreen}$\blacklozenge_\text{main}$} requests are both in-flight, and each of them records two active monitor nodes and then executes \texttt{yield}. Then, each request might turn off a complement monitor node. As a result of each request operating based on its isolated snapshot of the global state, both monitor nodes can be turned off --- inducing a request with {\color{ForestGreen}$\blacklozenge_\text{main}$}/{\color{red}$\blacklozenge_0$} (for \texttt{[N$_1$+N$_2$=0+0=0]}).
We note that in any serial execution, no two {\color{ForestGreen}$\blacklozenge_\text{main}$} requests can \textit{simultaneously} record both monitors as active, and hence, a response of {\color{red}$\blacklozenge_0$} cannot be obtained by serial executions.

\begin{minipage}[htbp]{1.1\textwidth}
	\begin{lstlisting}[caption={Snapshot-based monitor deactivation (not serializable)},numbers=none]
// initialize both monitors to be active
N_1_ACTIVE := 1
N_2_ACTIVE := 1

request main:
    // take snapshot
    n_1_active_snapshot := N_1_ACTIVE
    n_2_active_snapshot := N_2_ACTIVE
    yield

    if (n_1_active_snapshot == 1) and (n_2_active_snapshot == 1):
        // if both nodes active --- choose which one to deactivate 
        if (?): 
            N_1_ACTIVE := 0
        else:
            N_2_ACTIVE := 0

    return N_1_ACTIVE + N_2_ACTIVE  // total active nodes

request activate_n1:
    N_1_ACTIVE := 1

request activate_n2:
    N_2_ACTIVE := 1

	\end{lstlisting}
\end{minipage}




%% file: sections/8_appendix_toy_petri_net.tex
\clearpage

\section{Toy Petri Net Example}
\label{appendix:toyPN}

Observe the toy Petri net in Fig.~\ref{fig:toyPN}.

\begin{figure}[H]
	\centering
	\includegraphics[width=0.3\textwidth]{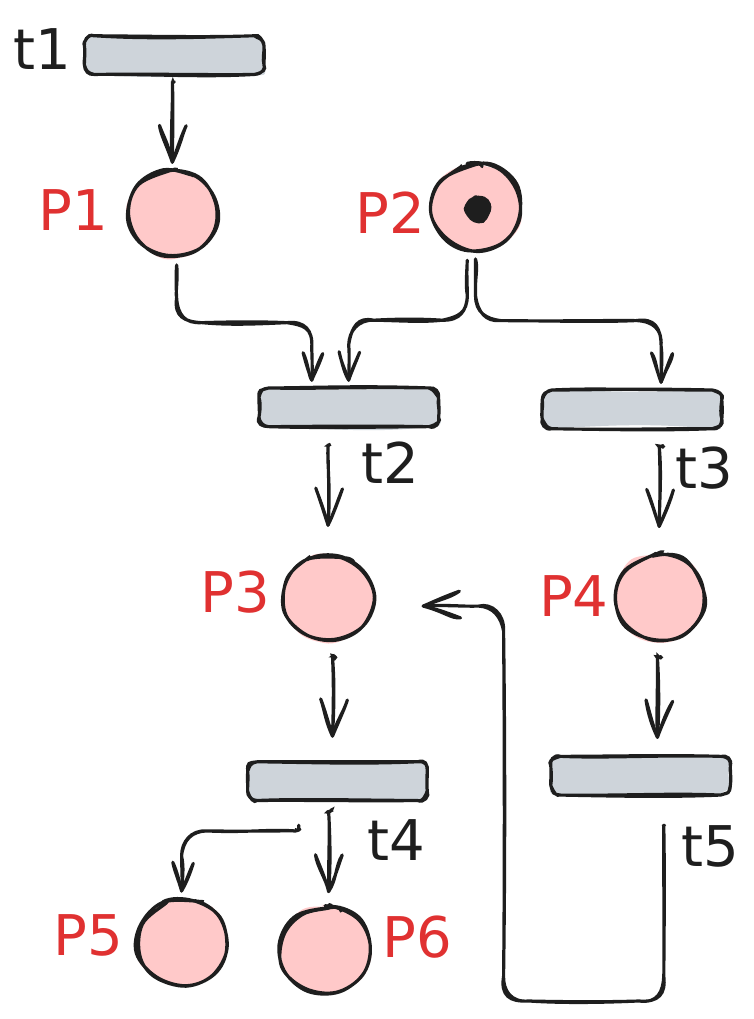}
	\caption{A toy Petri net.}
	\label{fig:toyPN}
\end{figure}

We formally define the net as follows:

 \(N=(P,T,\mathsf{pre},\mathsf{post},M_0)\) with
\[
P=\{P_1,P_2,P_3,P_4,P_5,P_6\},\quad
T=\{t_1,t_2,t_3,t_4,t_5\},
\]
and the flow functions $\mathsf{pre}$,$\mathsf{post}$ are given as
\[
\begin{array}{c|cccccc}
	& P_1 & P_2 & P_3 & P_4 & P_5 & P_6 \\ \hline
	\mathsf{pre}(t_1)  & 0 & 0 & 0 & 0 & 0 & 0 \\
	\mathsf{post}(t_1) & 1 & 0 & 0 & 0 & 0 & 0 \\ \hline
	\mathsf{pre}(t_2)  & 1 & 1 & 0 & 0 & 0 & 0 \\
	\mathsf{post}(t_2) & 0 & 0 & 1 & 0 & 0 & 0 \\ \hline
	\mathsf{pre}(t_3)  & 0 & 1 & 0 & 0 & 0 & 0 \\
	\mathsf{post}(t_3) & 0 & 0 & 0 & 1 & 0 & 0 \\ \hline
	\mathsf{pre}(t_4)  & 0 & 0 & 1 & 0 & 0 & 0 \\
	\mathsf{post}(t_4) & 0 & 0 & 0 & 0 & 1 & 1 \\ \hline
	\mathsf{pre}(t_5)  & 0 & 0 & 0 & 1 & 0 & 0 \\
	\mathsf{post}(t_5) & 0 & 0 & 1 & 0 & 0 & 0
\end{array}
\]
The initial marking is
\[
M_0 = (0,1,0,0,0,0)^\top 
\]

Differently put, there is a single token in place $P_2$.

\begin{itemize}
	\item An examples of a \emph{reachable} marking is
	\[
	M_f = (0,0,0,0,1,1)^\top,
	\]
	reached by the firing sequence
	\[
	M_0 \xrightarrow{t_1} M_1
	\xrightarrow{t_2} M_2
	\xrightarrow{t_4} M_f,
	\]
	where
	\[
	M_1 = (1,1,0,0,0,0)^\top,
	\quad
	M_2 = (0,0,1,0,0,0)^\top.
	\]
	\item An example of a \emph{non-reachable} marking is
	\[
	M_{nr} = (0,1,1,0,0,0)^\top.
	\]
	Since producing a token at \(P_3\) (via \(t_2\)) necessarily consumes the only token in \(P_2\) and, as no transition replenishes \(P_2\), then it is impossible for these two places to \textit{simultaneously} hold a single token in any reachable firing. However, we note that if the initial marking were 
	
	\[
	M_0' = (0,2,0,0,0,0)^\top,
	\]
	then marking $M_{nr}$ \textit{would have} been reachable, by firing a single transition $t_1$, followed by a single transition $t_2$.
\end{itemize}


%% file: sections/8_appendix_ser_semantics.tex
\clearpage
\section{SER Small-Step Semantics}
\label{appendix:ser-semantics}

%
%
The set \(\texttt{V}\) is a finite set of numeric constants; booleans use $0/1$. We respectively denote with \(\texttt{VARS}\) and \(\texttt{vars}\) the (finite) sets of global and local variables. Mappings $\rho:{\texttt{vars}}\to \texttt{V}$ and $g:{\texttt{VARS}}\to \texttt{V}$ respectively map a local or global variable to its current value in \(\texttt{V}\).
Configurations are denoted as $\cfg{e}{\rho}{g}$, with \(e\) being a valid \toolname{} expression. 
Small steps are denoted $(\step)$, while big steps are denoted $(\pstep)$, and may comprise of a sequence of small steps (denoted $\step^{*}$).
%
%

\smallskip
\noindent\textit{Small step} $(\step)$.
\begin{mathpar}
	\inferrule*[right=ND-0]{ }{\cfg{\nondet}{\rho}{g} \step \cfg{0}{\rho}{g}}
	\and
	\inferrule*[right=ND-1]{ }{\cfg{\nondet}{\rho}{g} \step \cfg{1}{\rho}{g}}\\
	
	\inferrule*[right=LOCAL-READ]{\rho(x)=v\quad v\in{\texttt{V}}}{\cfg{x}{\rho}{g} \step \cfg{v}{\rho}{g}}
	\and
	\inferrule*[right=GLOBAL-READ]{g(X)=v \quad v\in{\texttt{V}}}{\cfg{X}{\rho}{g} \step \cfg{v}{\rho}{g}}\\
	
	\inferrule*[right=LOCAL-WRITE-STEP]{\cfg{e}{\rho}{g} \step \cfg{e'}{\rho'}{g'}}{\cfg{x := e}{\rho}{g} \step \cfg{x := e'}{\rho'}{g'}}
	\and
	\inferrule*[right=LOCAL-WRITE-DONE]{v\in{\texttt{V}}}{\cfg{x := v}{\rho}{g} \step \cfg{v}{\update{\rho}{x}{v}}{g}}
	
	\inferrule*[right=GLOBAL-WRITE-STEP]{\cfg{e}{\rho}{g} \step \cfg{e'}{\rho'}{g'}}{\cfg{X := e}{\rho}{g} \step \cfg{X := e'}{\rho'}{g'}}
	\and
	\inferrule*[right=GLOBAL-WRITE-DONE]{v\in{\texttt{V}}}{\cfg{X := v}{\rho}{g} \step \cfg{v}{\rho}{\update{g}{X}{v}}}
	
	\inferrule*[right=EQ-L]{\cfg{e_1}{\rho}{g} \step \cfg{e_1'}{\rho'}{g'}}{\cfg{e_1 == e_2}{\rho}{g} \step \cfg{e_1' == e_2}{\rho'}{g'}}
	\and
	\inferrule*[right=EQ-R]{\cfg{e_2}{\rho}{g} \step \cfg{e_2'}{\rho'}{g'}}{\cfg{v_1 == e_2}{\rho}{g} \step \cfg{v_1 == e_2'}{\rho'}{g'}}
	\and
	\inferrule*[right=EQ-T]{v_1=v_2\quad v_1, v_2\in{\texttt{V}}}{\cfg{v_1 == v_2}{\rho}{g} \step \cfg{1}{\rho}{g}}
	\and
	\inferrule*[right=EQ-F]{v_1\neq v_2 \quad v_1, v_2\in{\texttt{V}}}{\cfg{v_1 == v_2}{\rho}{g} \step \cfg{0}{\rho}{g}}
	\end{mathpar}
	
	\begin{mathpar}
	\inferrule*[right=SEQ-STEP]{\cfg{e_1}{\rho}{g} \step \cfg{e_1'}{\rho'}{g'}}{\cfg{e_1 ; e_2}{\rho}{g} \step \cfg{e_1' ; e_2}{\rho'}{g'}}
	\and
	\inferrule*[right=SEQ-DONE]{v\in{\texttt{V}}}{\cfg{v ; e_2}{\rho}{g} \step \cfg{e_2}{\rho}{g}}

	\inferrule*[right=IF-GUARD]{\cfg{e_1}{\rho}{g} \step \cfg{e_1'}{\rho'}{g'}}{\cfg{\ifkw(e_1)\{e_2\}\elsekw\{e_3\}}{\rho}{g} \step \cfg{\ifkw(e_1')\{e_2\}\elsekw\{e_3\}}{\rho'}{g'}}
	\and
	\inferrule*[right=IF-T]{ }{\cfg{\ifkw(1)\{e_2\}\elsekw\{e_3\}}{\rho}{g} \step \cfg{e_2}{\rho}{g}}
	\and
	\inferrule*[right=IF-F]{ }{\cfg{\ifkw(0)\{e_2\}\elsekw\{e_3\}}{\rho}{g} \step \cfg{e_3}{\rho}{g}}
	
	\inferrule*[right=WHILE-UNFOLD]{ }{\cfg{\whilekw(e_1)\{e_2\}}{\rho}{g}
		\step
		\cfg{\ifkw(e_1)\{\,e_2 ; \whilekw(e_1)\{e_2\}\,\}\elsekw\{0\}}{\rho}{g}}
\end{mathpar}

\noindent\textit{Big step} $(\pstep)$ and scheduling.
\begin{mathpar}
	\inferrule*[right=YIELD]{\cfg{e}{\rho}{g} \step^{*} \cfg{\yieldkw ; e'}{\rho'}{g'}}{\cfg{e}{\rho}{g} \pstep \cfg{e'}{\rho'}{g'}}
	\\
	\inferrule*[right=TERMINATE]{\cfg{e}{\rho}{g} \step^{*} \cfg{v}{\rho'}{g'} \quad v\in{\texttt{V}}}{\cfg{e}{\rho}{g} \pstep \cfg{v}{\rho'}{g'}}	
\end{mathpar}

\smallskip
\noindent
\textbf{Note.}
Instead of defining a \texttt{spawn} instruction, as exists in some languages --- \toolname{} captures \textit{external} spawning via requests.
This setting can equivalently capture self-spawning (by using additional global variables), while translating more naturally to the networking domain --- in which threads are captured by packets sent by an external user.

%% file: sections/8_appendix_more_NS_examples.tex
\clearpage

\section{Additional Network System Examples}
\label{appendix:MoreNsExamples}

\subsection{Translation Example: Listing~\ref{lst:MotivatingExample1Ser}}
\label{appendix:subsec::Ex1A:NS}

For our first motivating example, presented in Listing~\ref{lst:MotivatingExample1Ser}, we depict the NS in Fig.~\ref{fig:code1ExampleNS}, the Serializability NFA in Fig.~\ref{fig:code1ExampleNFA}, and the Interleaving Petri net in Fig.~\ref{fig:code1ExamplePN}.


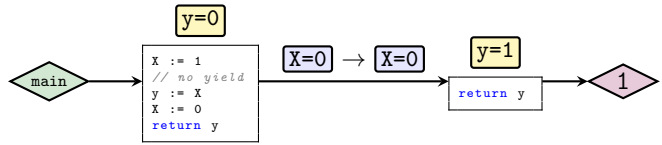
\begin{figure}[!htbp]
	\centering

	\begin{tikzpicture}[
		node distance=1.5cm and 2.5cm,
		>=stealth,
		thick,
		every node/.style={font=\small}
		]
		\node[
		draw=black,
		line width=0.8pt,
		fill=ForestGreen!20,
		text=black,
		diamond,
		aspect=2,
		inner sep=2pt,
		scale=0.7
		] (main) {\texttt{main}};
		
		\node[right=0.7cm of main, align=center] (state1) {
			\begin{tikzpicture}[baseline=(ybox.base)]
				\node[
				draw=black,
				line width=0.8pt,
				fill=brightyellow,
				text=black,
				rectangle,
				rounded corners=1pt,
				inner sep=2pt
				] (ybox) {\texttt{y=0}};
			\end{tikzpicture}\\[-2.5pt]
			\begin{minipage}{1.3cm}
				\begin{lstlisting}[language=CustomPseudoCode,numbers=none,basicstyle=\tiny\ttfamily]
X := 1
// no yield
y := X
X := 0
return y
				\end{lstlisting}
			\end{minipage}
		};
		
		\node[right=of state1, align=center] (state2) {
			\begin{tikzpicture}[baseline=(ybox.base)]
				\node[
				draw=black,
				line width=0.8pt,
				fill=brightyellow,
				text=black,
				rectangle,
				rounded corners=1pt,
				inner sep=2pt
				] (ybox) {\texttt{y=1}};
			\end{tikzpicture}\\[-2.5pt]
			\begin{minipage}{1.0cm}
				\begin{lstlisting}[language=CustomPseudoCode,numbers=none,basicstyle=\tiny\ttfamily]
return y
				\end{lstlisting}
			\end{minipage}
		};
		
		\node[
		right=0.6cm of state2,
		draw=black,
		line width=0.8pt,
		fill=RedViolet!20,
		text=black,
		diamond,
		aspect=2,
		inner sep=2pt,
		scale=0.7,
		font=\Large
		] (resp1) {\texttt{1}};
		
		\draw[->] (main) -- (state1);
		
		\draw[->] (state1) -- node[above] {%
			\begin{tikzpicture}[baseline=(a.base)]
				\node[draw=black,line width=0.8pt,fill=blue!10,rectangle,rounded corners=1pt,inner sep=2pt] (a) {\texttt{X=0}};
			\end{tikzpicture}
			$\to$
			\begin{tikzpicture}[baseline=(b.base)]
				\node[draw=black,line width=0.8pt,fill=blue!10,rectangle,rounded corners=1pt,inner sep=2pt] (b) {\texttt{X=0}};
			\end{tikzpicture}
		} (state2);
		
		\draw[->] (state2) -- (resp1);
		
	\end{tikzpicture}
	
	\caption{The network system for interleaved executions of the program in Listing~\ref{lst:MotivatingExample1Ser}.}
\label{fig:code1ExampleNS}
\end{figure}



\begin{figure}  [!htbp]
	\centering

	\begin{tikzpicture}[
		->,>=stealth,
		thick,
		node distance=2.5cm,
		state/.style={
			draw=black,
			line width=0.8pt,
			fill=blue!10,
			rectangle,
			rounded corners=1pt,
			inner sep=2pt,
			font=\small
		},
		every node/.style={font=\small}
		]
		\node[state] (X0) {\texttt{X=0}};
		
		\draw[->] ([yshift=-0.4cm]X0.south) -- (X0.south);
		
		\draw[->] (X0) edge[loop right]
		node[right] {${\color{ForestGreen}\blacklozenge_{\mathrm{main}}}/{\color{red}\blacklozenge_1}$} (X0);
	\end{tikzpicture}
	
	\caption{The NFA for serial executions of the program in Listing~\ref{lst:MotivatingExample1Ser}.}
\label{fig:code1ExampleNFA}
\end{figure}
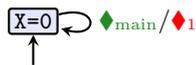

\begin{figure}[H]
	\centering
	\includegraphics[width=0.6\textwidth]{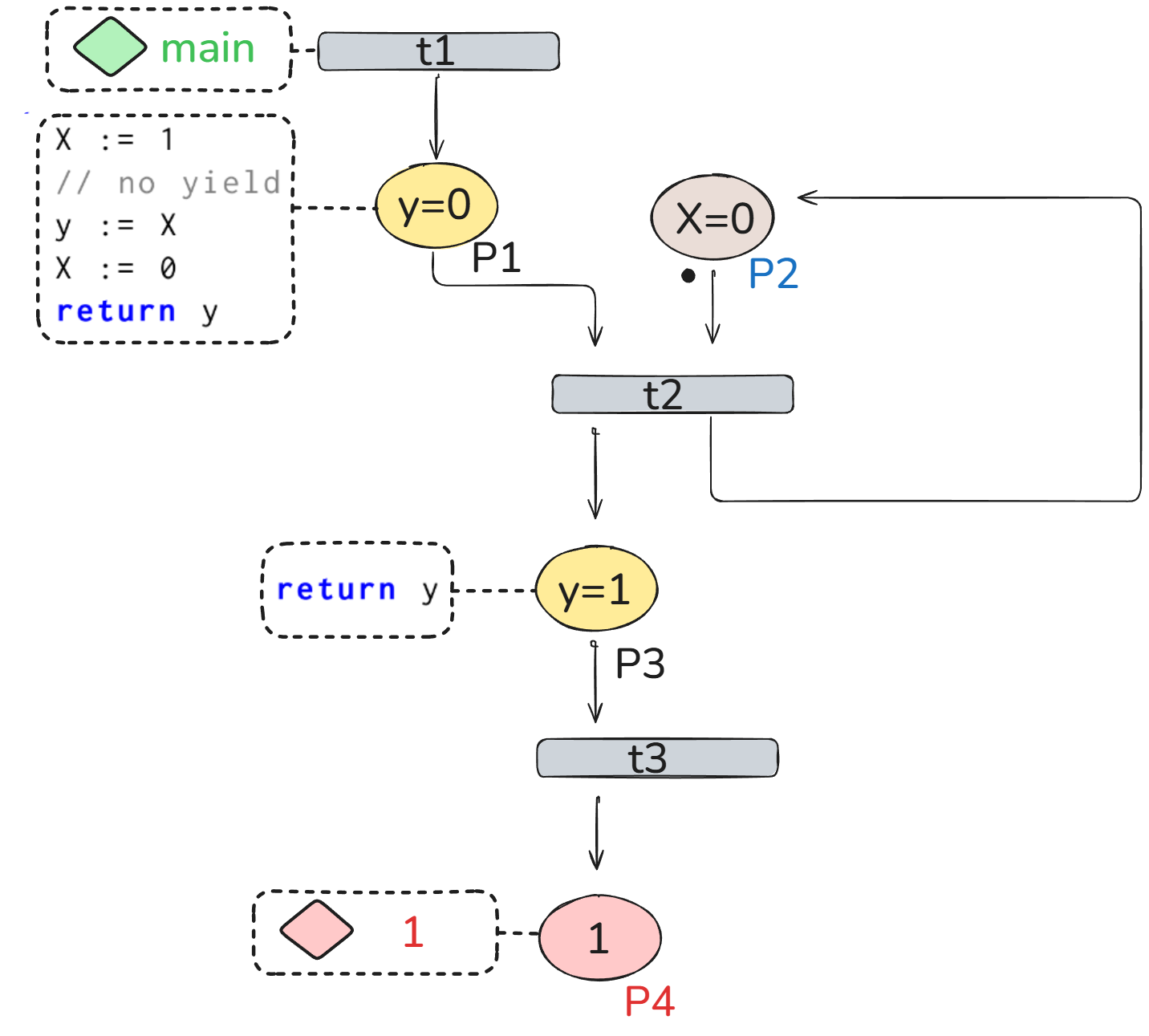}
	\caption{The Petri net for interleaved executions of the program in Listing~\ref{lst:MotivatingExample1Ser}.}
	\label{fig:code1ExamplePN}
\end{figure}

\subsection{Translation Example: Listing~\ref{lst:MotivatingExample2NonSer}}
\label{appendix:subsec::Ex1B:NS}

The NS, Serializability NFA, and Interleaving Petri net of Listing~\ref{lst:MotivatingExample2NonSer} are depicted in the main text (see subsec.~\ref{subsec:SerToNsTranslation}).
We present in Fig.~\ref{fig:code2ExampleNSSecondPart} the mappings \(\delta\), $req$, and $resp$.

\begin{figure}[!htbp]
	\centering
	
	\[
	\begin{array}{@{}r@{\;}l}
		req \coloneq & 
		\big\{
		\big[
		\begin{array}{c c c}
			\begin{tikzpicture}[baseline=(textnode.base)]
				\node[
				draw=black,
				line width=0.8pt,
				fill=ForestGreen!20,
				text=black,
				diamond,
				aspect=2,
				inner sep=2pt,
				scale=0.7
				] (textnode) {\texttt{main}};
			\end{tikzpicture}
			&\!\!\rightarrow\!\!&
			\begin{array}{c}
				\begin{tikzpicture}[baseline=(ybox.base)]
					\node[
					draw=black,
					line width=0.8pt,
					fill=brightyellow,
					text=black,
					rectangle,
					rounded corners=1pt,
					inner sep=2pt
					] (ybox) {\texttt{y=0}};
				\end{tikzpicture}\vspace{-2pt}
				\\
				\begin{minipage}{0.20\linewidth}
					\begin{lstlisting}[language=CustomPseudoCode,numbers=none,basicstyle=\tiny\ttfamily]
X := 1 
yield 
y := X
X := 0
return y
					\end{lstlisting}
				\end{minipage}
			\end{array}
		\end{array}
		\big]
		\big\}
		\\[2em]
		resp \coloneq &
		\big\{
		\big[
		\begin{array}{c c c}
			\begin{array}{c}
				\begin{tikzpicture}[baseline=(ybox.base)]
					\node[
					draw=black,
					line width=0.8pt,
					fill=brightyellow,
					text=black,
					rectangle,
					rounded corners=1pt,
					inner sep=2pt
					] (ybox) {\texttt{y=0}};
				\end{tikzpicture}\vspace{-2pt}
				\\
				\begin{minipage}{0.11\linewidth}
					\begin{lstlisting}[language=CustomPseudoCode,numbers=none,basicstyle=\tiny\ttfamily]
return y
					\end{lstlisting}
				\end{minipage}
			\end{array}
			&\!\!\rightarrow\!\!&
			\begin{tikzpicture}[baseline=(textnode.base),scale=0.7]
				\node[
				draw=black,
				line width=0.8pt,
				fill=RedViolet!20,
				text=black,
				diamond,
				aspect=2,
				inner sep=2pt,
				font=\small
				] (textnode) {\texttt{0}};
			\end{tikzpicture}
		\end{array}
		\big]\,{},
		\big[
		\begin{array}{c c c}
			\begin{array}{c}
				\begin{tikzpicture}[baseline=(ybox.base)]
					\node[
					draw=black,
					line width=0.8pt,
					fill=brightyellow,
					text=black,
					rectangle,
					rounded corners=1pt,
					inner sep=2pt
					] (ybox) {\texttt{y=1}};
				\end{tikzpicture}\vspace{-2pt}
				\\
				\begin{minipage}{0.11\linewidth}
					\begin{lstlisting}[language=CustomPseudoCode,numbers=none,basicstyle=\tiny\ttfamily]
return y
					\end{lstlisting}
				\end{minipage}
			\end{array}
			&\!\!\rightarrow\!\!&
			\begin{tikzpicture}[baseline=(textnode.base),scale=0.7]
				\node[
				draw=black,
				line width=0.8pt,
				fill=RedViolet!20,
				text=black,
				diamond,
				aspect=2,
				inner sep=2pt,
				font=\small
				] (textnode) {\texttt{1}};
			\end{tikzpicture}
		\end{array}
		\big]
		\big\}
		\\[2em]
		\delta \coloneq & 
		\big\{\big[(
		\begin{tikzpicture}[baseline=(ybox.base)]
			\node[
			draw=black,
			line width=0.8pt,
			fill=blue!10,
			text=black,
			rectangle,
			rounded corners=1pt,
			inner sep=2pt
			] (ybox) {\texttt{X=0}};
		\end{tikzpicture}\,{},
		\begin{array}{c}
			\begin{tikzpicture}[baseline=(ybox.base)]
				\node[
				draw=black,
				line width=0.8pt,
				fill=brightyellow,
				text=black,
				rectangle,
				rounded corners=1pt,
				inner sep=2pt
				] (ybox) {\texttt{y=0}};
			\end{tikzpicture}\vspace{-2pt}
			\\
			\begin{minipage}{0.14\linewidth}
				\begin{lstlisting}[language=CustomPseudoCode,numbers=none,basicstyle=\tiny\ttfamily]
X := 1
yield
y := X
X := 0
return y
				\end{lstlisting}
			\end{minipage}
		\end{array}
		)
		\;\rightarrow\;
		(
		\begin{tikzpicture}[baseline=(ybox.base)]
			\node[
			draw=black,
			line width=0.8pt,
			fill=blue!10,
			text=black,
			rectangle,
			rounded corners=1pt,
			inner sep=2pt
			] (ybox) {\texttt{X=1}};
		\end{tikzpicture}\,{},
		\begin{array}{c}
			\begin{tikzpicture}[baseline=(ybox.base)]
				\node[
				draw=black,
				line width=0.8pt,
				fill=brightyellow,
				text=black,
				rectangle,
				rounded corners=1pt,
				inner sep=2pt
				] (ybox) {\texttt{y=0}};
			\end{tikzpicture}\vspace{-2pt}
			\\
			\begin{minipage}{0.14\linewidth}
				\begin{lstlisting}[language=CustomPseudoCode,numbers=none,basicstyle=\tiny\ttfamily]
y := X
X := 0
return y
				\end{lstlisting}
			\end{minipage}
		\end{array}
		)
		\big],
		\\[0.5em]
		& \phantom{\big\{}
		\big[(
		\begin{tikzpicture}[baseline=(ybox.base)]
			\node[
			draw=black,
			line width=0.8pt,
			fill=blue!10,
			text=black,
			rectangle,
			rounded corners=1pt,
			inner sep=2pt
			] (ybox) {\texttt{X=1}};
		\end{tikzpicture}\,{},
		\begin{array}{c}
			\begin{tikzpicture}[baseline=(ybox.base)]
				\node[
				draw=black,
				line width=0.8pt,
				fill=brightyellow,
				text=black,
				rectangle,
				rounded corners=1pt,
				inner sep=2pt
				] (ybox) {\texttt{y=0}};
			\end{tikzpicture}\vspace{-2pt}
			\\
			\begin{minipage}{0.14\linewidth}
				\begin{lstlisting}[language=CustomPseudoCode,numbers=none,basicstyle=\tiny\ttfamily]
y := X
X := 0
return y
				\end{lstlisting}
			\end{minipage}
		\end{array}
		)
		\;\rightarrow\;
		(
		\begin{tikzpicture}[baseline=(ybox.base)]
			\node[
			draw=black,
			line width=0.8pt,
			fill=blue!10,
			text=black,
			rectangle,
			rounded corners=1pt,
			inner sep=2pt
			] (ybox) {\texttt{X=0}};
		\end{tikzpicture}\,{},
		\begin{array}{c}
			\begin{tikzpicture}[baseline=(ybox.base)]
				\node[
				draw=black,
				line width=0.8pt,
				fill=brightyellow,
				text=black,
				rectangle,
				rounded corners=1pt,
				inner sep=2pt
				] (ybox) {\texttt{y=1}};
			\end{tikzpicture}
			\vspace{-2pt}
			\\
			\begin{minipage}{0.11\linewidth}
				\begin{lstlisting}[language=CustomPseudoCode,numbers=none,basicstyle=\tiny\ttfamily]
return y
				\end{lstlisting}
			\end{minipage}
		\end{array}
		)
		\big],
		\ldots
		\big\}
	\end{array}
	\]
	\caption{The \(\delta\) transition function, and the \(req\) and \(resp\) mappings for the program in Listing~\ref{lst:MotivatingExample2NonSer}.}
	\label{fig:code2ExampleNSSecondPart}
\end{figure}
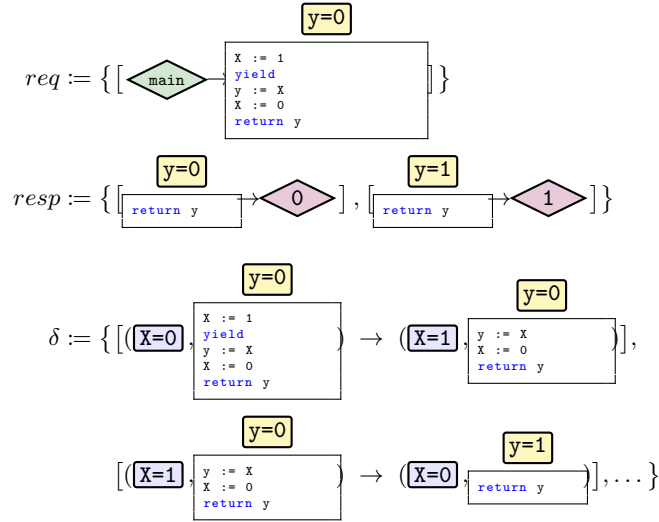

\subsection{Translation Example: Listing~\ref{lst:MotivatingExample3Ser}}
\label{appendix:subsec:Ex1C:NS}

For our third motivating example, presented in Listing~\ref{lst:MotivatingExample3Ser}, we denote the NS in Fig.~\ref{fig:code3ExampleNS}, the Serializability NFA in Fig.~\ref{fig:code3ExampleNFA}, and the Interleaving Petri net in Fig.~\ref{fig:code3ExamplePN}.

%
%
%


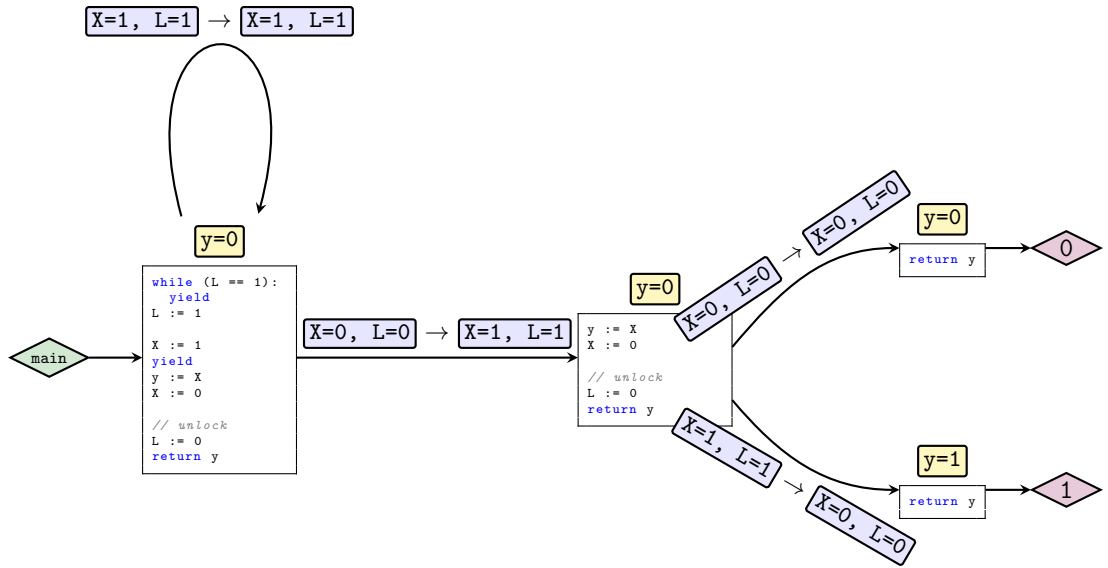
\begin{figure}[!htbp]
	\centering

	\begin{tikzpicture}[
		node distance=1.5cm and 2.5cm,
		>=stealth,
		thick,
		every node/.style={font=\small}
		]
		\node[
		draw=black,
		line width=0.8pt,
		fill=ForestGreen!20,
		text=black,
		diamond,
		aspect=2,
		inner sep=2pt,
		scale=0.7
		] (main) {\texttt{main}};
		
		\node[right=0.7cm of main, align=center] (state1) {
			\begin{tikzpicture}[baseline=(ybox.base)]
				\node[
				draw=black,
				line width=0.8pt,
				fill=brightyellow,
				text=black,
				rectangle,
				rounded corners=1pt,
				inner sep=2pt
				] (ybox) {\texttt{y=0}};
			\end{tikzpicture}\\[-2.5pt]
			\begin{minipage}{1.8cm}
				\begin{lstlisting}[language=CustomPseudoCode,numbers=none,basicstyle=\tiny\ttfamily]
while (L == 1):
	yield
L := 1

X := 1
yield
y := X
X := 0

// unlock
L := 0
return y
				\end{lstlisting}
			\end{minipage}
		};
		
		\node[right=of state1, xshift=12mm, align=center] (state2) {
			\begin{tikzpicture}[baseline=(ybox.base)]
				\node[
				draw=black,
				line width=0.8pt,
				fill=brightyellow,
				text=black,
				rectangle,
				rounded corners=1pt,
				inner sep=2pt
				] (ybox) {\texttt{y=0}};
			\end{tikzpicture}\\[-2.5pt]
			\begin{minipage}{1.8cm}
				\begin{lstlisting}[language=CustomPseudoCode,numbers=none,basicstyle=\tiny\ttfamily]
y := X
X := 0

// unlock
L := 0
return y
				\end{lstlisting}
			\end{minipage}
		};
		
		\node[above right=-0.5cm and 2.2cm of state2, align=center] (state3) {
			\begin{tikzpicture}[baseline=(ybox.base)]
				\node[
				draw=black,
				line width=0.8pt,
				fill=brightyellow,
				text=black,
				rectangle,
				rounded corners=1pt,
				inner sep=2pt
				] (ybox) {\texttt{y=0}};
			\end{tikzpicture}\\[-2.5pt]
			\begin{minipage}{0.9cm}
				\begin{lstlisting}[language=CustomPseudoCode,numbers=none,basicstyle=\tiny\ttfamily]
return y
				\end{lstlisting}
			\end{minipage}
		};
		
		\node[below right=-0.2cm and 2.2cm of state2, align=center] (state4) {
			\begin{tikzpicture}[baseline=(ybox.base)]
				\node[
				draw=black,
				line width=0.8pt,
				fill=brightyellow,
				text=black,
				rectangle,
				rounded corners=1pt,
				inner sep=2pt
				] (ybox) {\texttt{y=1}};
			\end{tikzpicture}\\[-2.5pt]
			\begin{minipage}{0.9cm}
				\begin{lstlisting}[language=CustomPseudoCode,numbers=none,basicstyle=\tiny\ttfamily]
return y
				\end{lstlisting}
			\end{minipage}
		};
		
		\node[
		right=0.6cm of state3,
		draw=black,
		line width=0.8pt,
		fill=RedViolet!20,
		text=black,
		diamond,
		aspect=2,
		inner sep=2pt,
		scale=0.7,
		font=\Large
		] (resp0) {\texttt{0}};
		
		\node[
		right=0.6cm of state4,
		draw=black,
		line width=0.8pt,
		fill=RedViolet!20,
		text=black,
		diamond,
		aspect=2,
		inner sep=2pt,
		scale=0.7,
		font=\Large
		] (resp1) {\texttt{1}};
		
		\draw[->] (main) -- (state1);
		
		\draw[->] (state1) edge[loop above]
		node[above] {%
			\begin{tikzpicture}[baseline=(a.base)]
				\node[draw=black,line width=0.8pt,fill=blue!10,rectangle,rounded corners=1pt,inner sep=2pt] (a) {\texttt{X=1, L=1}};
			\end{tikzpicture}
			$\to$
			\begin{tikzpicture}[baseline=(b.base)]
				\node[draw=black,line width=0.8pt,fill=blue!10,rectangle,rounded corners=1pt,inner sep=2pt] (b) {\texttt{X=1, L=1}};
			\end{tikzpicture}
		} (state1);
		
		\draw[->] (state1) -- node[above] {%
			\begin{tikzpicture}[baseline=(a.base)]
				\node[draw=black,line width=0.8pt,fill=blue!10,rectangle,rounded corners=1pt,inner sep=2pt] (a) {\texttt{X=0, L=0}};
			\end{tikzpicture}
			$\to$
			\begin{tikzpicture}[baseline=(b.base)]
				\node[draw=black,line width=0.8pt,fill=blue!10,rectangle,rounded corners=1pt,inner sep=2pt] (b) {\texttt{X=1, L=1}};
			\end{tikzpicture}
		} (state2);
		
		\draw[->] ([yshift=4pt]state2.east) to[out=50,in=180]
		node[above,sloped] {%
			\begin{tikzpicture}[baseline=(a.base)]
				\node[draw=black,line width=0.8pt,fill=blue!10,rectangle,rounded corners=1pt,inner sep=2pt] (a) {\texttt{X=0, L=0}};
			\end{tikzpicture}
			$\to$
			\begin{tikzpicture}[baseline=(b.base)]
				\node[draw=black,line width=0.8pt,fill=blue!10,rectangle,rounded corners=1pt,inner sep=2pt] (b) {\texttt{X=0, L=0}};
			\end{tikzpicture}
		} (state3.west);
		
		\draw[->] ([yshift=-16pt]state2.east) to[out=-50,in=180]
		node[below,sloped] {%
			\begin{tikzpicture}[baseline=(a.base)]
				\node[draw=black,line width=0.8pt,fill=blue!10,rectangle,rounded corners=1pt,inner sep=2pt] (a) {\texttt{X=1, L=1}};
			\end{tikzpicture}
			$\to$
			\begin{tikzpicture}[baseline=(b.base)]
				\node[draw=black,line width=0.8pt,fill=blue!10,rectangle,rounded corners=1pt,inner sep=2pt] (b) {\texttt{X=0, L=0}};
			\end{tikzpicture}
		} (state4.west);
		
		\draw[->] (state3) -- (resp0);
		\draw[->] (state4) -- (resp1);
		
	\end{tikzpicture}
	
	\caption{The network system for interleaved executions of the program in Listing~\ref{lst:MotivatingExample3Ser}.}
\label{fig:code3ExampleNS}
\end{figure}


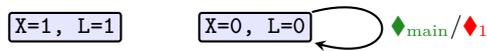
\begin{figure}  [!htbp]
	\centering

	\begin{tikzpicture}[
		->,>=stealth,
		thick,
		node distance=2.5cm,
		state/.style={
			draw=black,
			line width=0.8pt,
			fill=blue!10,
			rectangle,
			rounded corners=1pt,
			inner sep=2pt,
			font=\small
		},
		every node/.style={font=\small}
		]
		\node[state] (X1L1) {\texttt{X=1, L=1}};
		\node[state, right of=X1L1] (X0L0) {\texttt{X=0, L=0}};
		
		
		\draw[->] (X0L0) edge[loop right]
		node[right] {${\color{ForestGreen}\blacklozenge_{\mathrm{main}}}/{\color{red}\blacklozenge_1}$} (X0L0);
	\end{tikzpicture}
	
	\caption{The NFA for serial executions of the program in Listing~\ref{lst:MotivatingExample3Ser}.}
\label{fig:code3ExampleNFA}
\end{figure}

\begin{figure}[!htbp]
	\centering
	\includegraphics[width=0.8\textwidth]{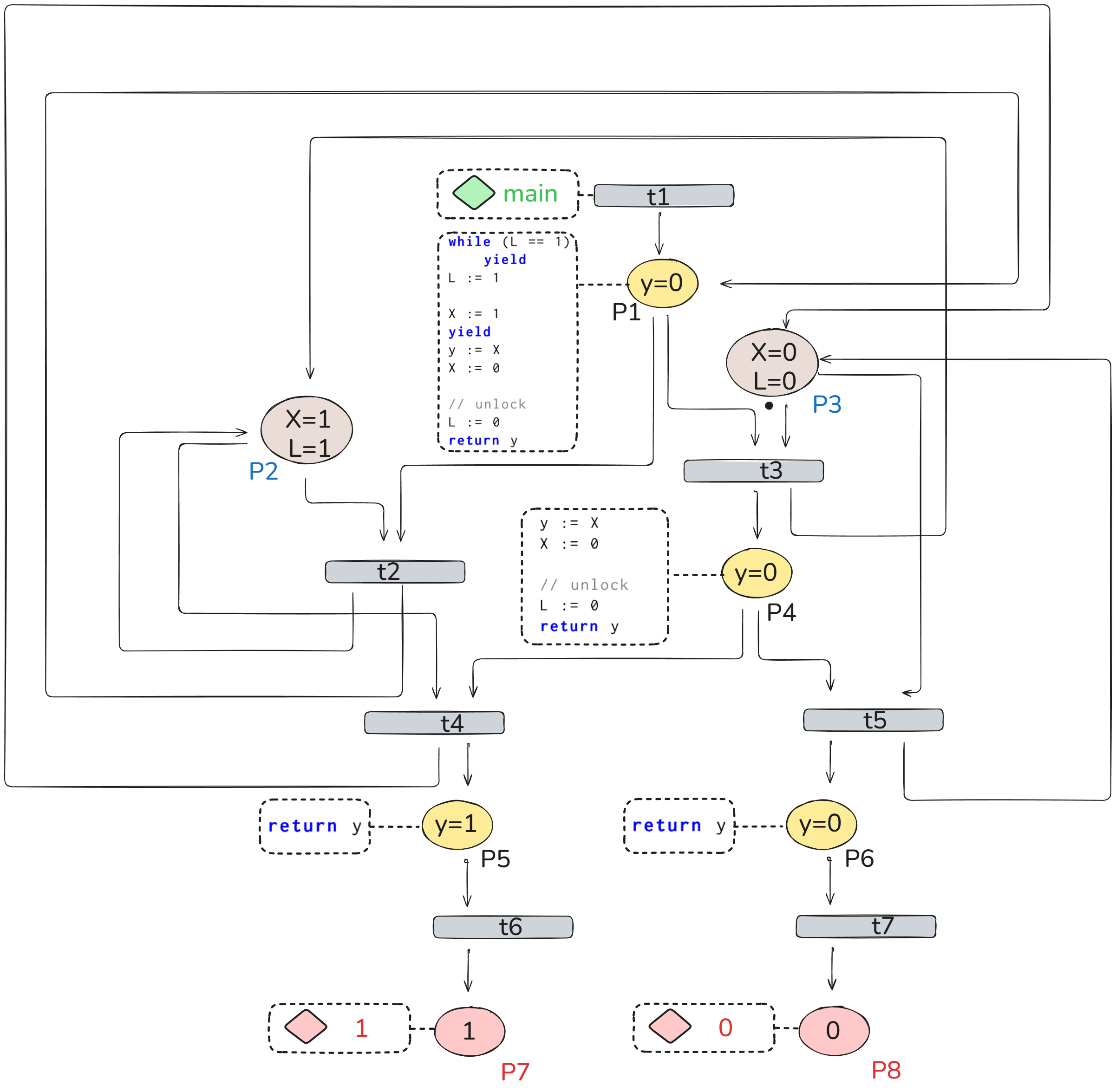}
	\caption{The Petri net for interleaved executions of the program in Listing~\ref{lst:MotivatingExample3Ser}.}
	\label{fig:code3ExamplePN}
\end{figure}


%% file: sections/8_appendix_NS_to_PN_formulation.tex
\clearpage

\section{Translating Network Systems to Petri Nets}
\label{appendix:NS-to-PN-formulation}

We denote with \(\mathbf0\) a zero vector of dimension \(|P|\), and with \(\mathbf1_{p}\) a \(|P|\)-sized indicator vector that has 0 in every coordinate except the one corresponding to place \(p \in P\), which has 1. 
The flow functions $\mathsf{pre},\mathsf{post}:T\to\{0,1\}^{|P|}$ assign to each transition $t$ a binary vector over $P$ whose $1$-entries mark the places from which tokens are consumed (for $\mathsf{pre}(t)$) and to which tokens are produced (for $\mathsf{post}(t)$) when $t$ fires.
A transition $t$ is enabled at $M$ iff $\mathsf{pre}(t)\le M$ (component-wise); firing yields
\(
M\xrightarrow{t}M' \quad\text{where}\quad M' = M-\mathsf{pre}(t)+\mathsf{post}(t)
\).	

\medskip
\textit{Construction.}
We generate the Petri net:
\[
N_{\mathrm{int}}(\mathcal S)
= (P,\,T,\,\mathsf{pre},\,\mathsf{post},\,M_0),
\]
where
\[
P
=
P_G \;\cup\; P_{REQ,L} \;\cup\; P_{REQ,RESP}
\]

for 
\[
\begin{aligned}
	P_G
	&= \{\,p_g \mid g\in G\},\quad
	P_{REQ,L}
	= \bigl\{\,p_{({\color{ForestGreen}\blacklozenge_{\mathit{req}}},\ell)}
	\mid {\color{ForestGreen}\blacklozenge_{\mathit{req}}}\in\mathit{REQ},\,\ell\in  L\bigr\},\\[1ex]
	P_{REQ,RESP}
	&= \bigl\{\,p_{({\color{ForestGreen}\blacklozenge_{\mathit{req}}}/{\color{red}\blacklozenge_{\mathit{resp}}})}
	\mid {\color{ForestGreen}\blacklozenge_{\mathit{req}}}\in\mathit{REQ},\,
	{\color{red}\blacklozenge_{\mathit{resp}}}\in\mathit{RESP}\bigr\}.
\end{aligned}
\]




with \(G\) being the set of global states, \(L\) being the set of local states (in the case of a \toolname-derived NS, this is the coupling of the local variable assignments of an in-flight request and its remaining \toolname{} program to execute), \(\mathit{REQ}\) denotes the request labels; and \(\mathit{RESP}\) denotes the response labels.

\medskip
Transitions are partitioned as:
\[
T = T_{\mathit{req}} \;\cup\; T_{\delta}\;\cup\;T_{\mathit{resp}}
\]
where


\begin{align*}
	T_{\mathit{req}}
	&= \{\,t_{({\color{ForestGreen}\blacklozenge_{\mathit{req}}},\ell)} \mid {(\color{ForestGreen}\blacklozenge_{\mathit{req}}},\ell)\in\mathit{req}\},\\[1ex]
	T_{\delta}
	&= \bigl\{\,t_{((\ell,g),(\ell',g'))} 
	\mid ((\ell,g),(\ell',g'))\in\delta\bigr\},\quad
	T_{\mathit{resp}}
	= \{\,t_{(\ell,{\color{red}\blacklozenge_{\mathit{resp}}})} \mid (\ell,{\color{red}\blacklozenge_{\mathit{resp}}})\in\mathit{resp}\}.
\end{align*}

Their \(\mathsf{pre}\) and \(\mathsf{post}\) flow functions are:
\[
\begin{alignedat}{3}
	\mathsf{pre}\bigl(t_{({\color{ForestGreen}\blacklozenge_{\mathit{req}}},\ell)}\bigr)
	&= \mathbf0, &
	\mathsf{post}\bigl(t_{({\color{ForestGreen}\blacklozenge_{\mathit{req}}},\ell)}\bigr)
	&= \mathbf1_{p_{({\color{ForestGreen}\blacklozenge_{\mathit{req}}},\ell)}}, 
	&&\text{for }({\color{ForestGreen}\blacklozenge_{\mathit{req}}},\ell)\in\mathit{req},\\
	\mathsf{pre}\bigl(t_{((\ell,g),(\ell',g'))}\bigr)
	&= \mathbf1_{p_{({\color{ForestGreen}\blacklozenge_{\mathit{req}}},\ell)}} + \mathbf1_{p_g}, &
	\mathsf{post}\bigl(t_{((\ell,g),(\ell',g'))}\bigr)
	&= \mathbf1_{p_{({\color{ForestGreen}\blacklozenge_{\mathit{req}}},\ell')}} + \mathbf1_{p_{g'}}, 
	&&\text{for }{{\color{ForestGreen}\blacklozenge_{\mathit{req}}}\in\mathit{REQ}}, ((\ell,g),(\ell',g'))\in\delta,\\
	\mathsf{pre}\bigl(t_{(\ell,{\color{red}\blacklozenge_{\mathit{resp}}})}\bigr)
	&= \mathbf1_{p_{({\color{ForestGreen}\blacklozenge_{\mathit{req}}},\ell)}}, &
	\mathsf{post}\bigl(t_{(\ell,{\color{red}\blacklozenge_{\mathit{resp}}})}\bigr)
	&= \mathbf1_{p_{({\color{ForestGreen}\blacklozenge_{\mathit{req}}}/{\color{red}\blacklozenge_{\mathit{resp}}})}}, 
	&&\text{for }{{\color{ForestGreen}\blacklozenge_{\mathit{req}}}\in\mathit{REQ},(\ell,\color{red}\blacklozenge_{\mathit{resp}}})\in\mathit{resp}
\end{alignedat}
\]

Where, for the last two cases, \({\color{ForestGreen}\blacklozenge_{\mathit{req}}}\) concerns requests that eventually give rise to a local state \(\ell \in L\) that originated downstream (during execution).

\medskip
The initial marking is a single token in the place representing the initial global state $g_0$ of the NS:
\[
M_0(p_{g_0}) = 1,
\quad
M_0(p) = 0 \text{ for all }p\neq p_{g_0},
\]

Define the projection \(\pi\) to solely include the markings of places representing completed request/response pairs.
Then, the multiset of all  (${{\color{ForestGreen}\blacklozenge_{\mathit{req}}}/{\color{red}\blacklozenge_{\mathit{resp}}}}$) pairs of the NS, obtained by \textit{any} interleaving, is:
\[
\mathsf{Int}(\mathcal S)
\;=\;
\bigl\{\;\pi(M)\;\bigm|\;M_0 \xrightarrow{}^{*} M\text{ in }N_{\mathrm{int}}(\mathcal S)\bigr\}.
\]

%% file: sections/8_appendix_serializable_program_proof_example.tex
\clearpage

\section{Example: Serializable Program}
\label{appendix:ns-serializable}

Now, we observe again the adjusted program with a spin-lock (as previously described in Listing~\ref{lst:MotivatingExample3Ser}), of which we depicted figures of the corresponding NS (Fig.~\ref{fig:code3ExampleNS}), Serializability NFA (Fig.~\ref{fig:code3ExampleNFA}), and Interleaving Petri net (Fig.~\ref{fig:code3ExamplePN}) in Appendix~\ref{appendix:MoreNsExamples}.
In this case, serializability corresponds to the Petri net being unable to reach a marking satisfying the same semilinear formula \(\mathcal {F}\) as in the non-serializable case described in the main text (subsec.~\ref{subsec:SerToNsTranslation}):

\[
\mathcal {F}:
\quad
P_1 = 0 \wedge 
\textcolor{blue}{P_2} \ge 0 \wedge \textcolor{blue}{P_3} \ge 0  \wedge P_4 = 0
\wedge P_5 = 0 \wedge P_6 = 0 \wedge \textcolor{red}{P_7} \ge 0 \wedge \textcolor{red}{P_8} \ge 1.
\]

%

In addition, although the target set is the same as in the previous example, the Petri net places $(P_1,\ldots,P_8)$ encode different states that correspond to the updated network system. For instance, now each place in the PN that encodes a global state accounts for two global variables, \texttt{X} and \texttt{L}, and the initial global state corresponds to the place encoding the initial assignment \textcolor{blue}{[X=0, L=0]}, etc.
Furthermore, unlike the case in Listing~\ref{lst:MotivatingExample2NonSer} (covered in subsec.~\ref{subsec:SerToNsTranslation}), this target set of markings (encoding request/response pairs of non-serial executions) is \textit{unreachable}, as witnessed by the inductive invariant:

\[
\begin{aligned}
	&(P_{1},\textcolor{blue}{P_{2}},\textcolor{blue}{P_{3}},P_{4},P_{5},P_{6},\textcolor{red}{P_{7}},\textcolor{red}{P_{8}})
	\;\mapsto\;\\
	&\quad
	\exists\,e_{0},\dots,e_{5}\ge0.\;
	\Bigl(
	e_{2}-e_{1}+\textcolor{blue}{P_{3}}-1=0\;\land\;
	e_{2}+P_{1}-e_{5}=0\;\land\;
	P_{5}-e_{1}+e_{4}=0\;\land\\
	&\qquad\quad
	-\,e_{4}+\textcolor{red}{P_{7}}=0\;\land\;
	P_{6}+e_{3}-e_{0}=0\;\land\;
	\textcolor{red}{P_{8}}-e_{3}=0\;\land\\
	&\qquad\quad
	-\,e_{2}+e_{1}+e_{0}+P_{4}=0\;\land\;
	-\,e_{2}+e_{1}+\textcolor{blue}{P_{2}}=0
	\Bigr)
	\;\land\;
	\bigl(P_{4}-1\ge0\;\lor\;\textcolor{blue}{P_{3}}-1\ge0\bigr).
\end{aligned}
\]

We then revert and project it on the request/response pairs of the network system.
We get the following inductive invariants for each of the two (reachable) global states:

\begin{proof}
	
	\medskip\noindent
	For global state \textcolor{blue}{[L=0,X=0]} the projected invariant is:
	\[
	\bigl(\,\text{\color{ForestGreen}$\blacklozenge_{\text{main}}$}/\text{\color{red}$\blacklozenge_{0}$},\;
	\text{\color{ForestGreen}$\blacklozenge_{\text{main}}$}/\text{\color{red}$\blacklozenge_{1}$}\bigr)
	\;\mapsto\;
	\exists\,e_{0},\dots,e_{5}\ge0.\;
	\begin{aligned}[t]
		& e_{2}-e_{1}=0,\quad
		e_{2}-e_{5}=0,\quad
		-e_{1}+e_{4}=0,\\
		& -e_{4}+\bigl(\text{\color{ForestGreen}$\blacklozenge_{\text{main}}$}/\text{\color{red}$\blacklozenge_{1}$}\bigr)=0,\quad
		-e_{0}+e_{3}=0,\\
		& -e_{3}+\bigl(\text{\color{ForestGreen}$\blacklozenge_{\text{main}}$}/\text{\color{red}$\blacklozenge_{0}$}\bigr)=0,\quad
		-e_{2}+e_{1}+e_{0}=0,\\
		& -e_{2}+e_{1}=0.
	\end{aligned}
	\]
	\noindent From 
	\[e_{1}=e_{2}=e_{4}=e_{5}=(\;
	\text{\color{ForestGreen}$\blacklozenge_{\text{main}}$}/\text{\color{red}$\blacklozenge_{1}$}),\;
	e_{0}=e_{3}=
	(\text{\color{ForestGreen}$\blacklozenge_{\text{main}}$}/\text{\color{red}$\blacklozenge_{0}$})
	\]
	
	it follows that \[-e_{2}+e_{1}+e_{0}=0\;\Longrightarrow\;e_{0}=0,\]
	
	thus: 
	\[
	(	\text{\color{ForestGreen}$\blacklozenge_{\text{main}}$}/\text{\color{red}$\blacklozenge_{0}$})
	=0 
	\]
	
	indicating that  (\(\text{\color{ForestGreen}$\blacklozenge_{\text{main}}$}/\text{\color{red}$\blacklozenge_{0}$}\)) cannot be obtained from the global state
	\textcolor{blue}{[L=0,X=0]}.
	
	\medskip\noindent
	In the second case, for the global state \textcolor{blue}{[L=1, X=1]}
	the projected invariant is:

	\[
	\bigl(\,\text{\color{ForestGreen}$\blacklozenge_{\text{main}}$}/\text{\color{red}$\blacklozenge_{0}$},\;
	\text{\color{ForestGreen}$\blacklozenge_{\text{main}}$}/\text{\color{red}$\blacklozenge_{1}$}\bigr)
	\;\mapsto\;
	\exists\,e_{0},\dots,e_{5}.\;\bot,
	\]
	which is unsatisfiable. Hence, no completed request/response pair, and in particular, no (\(\text{\color{ForestGreen}$\blacklozenge_{\text{main}}$}/\text{\color{red}$\blacklozenge_{0}$}\)) pair can be produced from this state via \textit{any} execution. Intuitively, this aligns with the fact that there cannot be any output generated via an interleaving, given that the spin-lock is acquired (\textcolor{blue}{[L=1]}).
	
	\medskip
	\noindent\textbf{Conclusion.}
	In every reachable state, no request/response pair of type	($	\text{\color{ForestGreen}$\blacklozenge_{\text{main}}$}/\text{\color{red}$\blacklozenge_{0}$})
	$
	can occur. Consequently, the only possible pairs are of type
	($	\text{\color{ForestGreen}$\blacklozenge_{\text{main}}$}/\text{\color{red}$\blacklozenge_{1}$})
	$,
	all of which lie within the NFA’s language for serial executions (Fig.~\ref{fig:code3ExampleNFA}).
	Hence, the program is serializable. Moreover, as proven in subsection~\ref{appendix:subsec:InductiveInvariantExample},
	these invariants are inductive: they hold in the initial state and are preserved under every transition.
\end{proof}



%


\subsection{Proof of Inductive Invariant}
\label{appendix:subsec:InductiveInvariantExample}

\begin{proof}
	
	Define the predicate
	\[
	\begin{aligned}
		I(P_{1},\dots,\textcolor{red}{P_{8}})
		:={}&
		(P_{1},\textcolor{blue}{P_{2}},\textcolor{blue}{P_{3}},P_{4},P_{5},P_{6},\textcolor{red}{P_{7}},\textcolor{red}{P_{8}})
		\;\mapsto\;\\
		&\quad
		\exists\,e_{0},\dots,e_{5}\ge0.\;
		\Bigl(
		e_{2}-e_{1}+\textcolor{blue}{P_{3}}-1=0\;\land\;
		e_{2}+P_{1}-e_{5}=0\;\land\;
		P_{5}-e_{1}+e_{4}=0\;\land\\
		&\qquad\quad
		-\,e_{4}+\textcolor{red}{P_{7}}=0\;\land\;
		P_{6}+e_{3}-e_{0}=0\;\land\;
		\textcolor{red}{P_{8}}-e_{3}=0\;\land\\
		&\qquad\quad
		-\,e_{2}+e_{1}+e_{0}+P_{4}=0\;\land\;
		-\,e_{2}+e_{1}+\textcolor{blue}{P_{2}}=0
		\Bigr)
		\;\land\;
		\bigl(P_{4}-1\ge0\;\lor\;\textcolor{blue}{P_{3}}-1\ge0\bigr).
	\end{aligned}
	\]

	\medskip\noindent
	\textbf{(1) Initialization.}
	The initial marking has $\textcolor{blue}{P_{3}}=1$ and $P_{1}=\textcolor{blue}{P_{2}}=P_{4}=P_{5}=P_{6}=\textcolor{red}{P_{7}}=\textcolor{red}{P_{8}}=0$.
	Choose $e_{0}=\cdots=e_{5}=0$.  Then
	\[
	e_{i}\ge0,\quad
	e_{2}-e_{1}+\textcolor{blue}{P_{3}}-1=0-0+1-1=0,\;\dots,\;-e_{2}+e_{1}+P_{2}=0,
	\]
	and 
	\[
	P_{4}-1\ge0\;\lor\;\textcolor{blue}{P_{3}}-1\ge0
	\;=\;-1\ge0\;\lor\;0\ge0
	\;=\;\texttt{FALSE}\;\lor\;\texttt{TRUE}
	\;=\;\texttt{TRUE}.
	\]
	Thus $I$ holds initially.
	
	\medskip\noindent
	\textbf{(2) Consecution.}
	One checks for each transition $t_{k}$ of the Petri net that
	\[
	I(M)\;\Longrightarrow\;I\bigl(t_{k}(M)\bigr).
	\]
	In each case, the same $(e_{0},\dots,e_{5})$ can be adjusted (per the \texttt{SMT} certificate) to show that the eight equalities and the disjunction remain valid. See our accompanying artifact~\cite{ArtifactRepository} for generating a full proof in the standard \texttt{SMT-LIB} format~\cite{BaStTi10}.
	
	\medskip\noindent
	\textbf{(3) Refutation of the property.}
	Suppose by contradiction that there exists a marking $P$ for which both $I(P)$ and $\mathcal {F}(P)$ hold:
	\[
	\mathcal {F}(P):\quad
	P_{1}=0,\;
	\textcolor{blue}{P_{2}}\ge0,\;
	\textcolor{blue}{P_{3}}\ge0,\;
	P_{4}=0,\;
	P_{5}=0,\;
	P_{6}=0,\;
	\textcolor{red}{P_{7}}\ge0,\;
	\textcolor{red}{P_{8}}\ge1.
	\] 
	
	\noindent
	From
	\[
	e_{2}-e_{1}+\textcolor{blue}{P_{3}}-1=0
	\quad\text{and}\quad
	-e_{2}+e_{1}+\textcolor{blue}{P_{2}}=0
	\]
	we get
	\[
	\textcolor{blue}{P_{2}}=1-\textcolor{blue}{P_{3}}.
	\]
	From
	\[
	\textcolor{red}{P_{8}}-e_{3}=0
	\quad\text{and}\quad
	P_{6}+e_{3}-e_{0}=0
	\]
	and from the assumption that $P_6=0$, we get
	\[
	e_{0}=e_{3}=\textcolor{red}{P_{8}}.
	\]

	\noindent
	Similarly, the invariant equalities 
	$(-\,e_{2}+e_{1}+e_{0}+P_{4}=0)$ and $(	-\,e_{2}+e_{1}+\textcolor{blue}{P_{2}}=0)$
	induce
	\[
	\textcolor{blue}{P_{2}}=P_{4}+e_{0}=P_{4}+\textcolor{red}{P_{8}},
	\]
	thus, and as we also assume that $P_4=0$, then:
	\[
	\textcolor{red}{P_{8}}=\textcolor{blue}{P_2}-P4=(1-\textcolor{blue}{P_{3}})-P_{4}=1-\textcolor{blue}{P_{3}}-0=1-\textcolor{blue}{P_{3}}
	\]

	\noindent
	However, $\mathcal {F}(P)$ also induces $\textcolor{blue}{P_{3}}\ge0$ and $\textcolor{red}{P_{8}}\ge1$, and hence $\textcolor{blue}{P_{3}}=0$.  
	Furthermore, as our invariant includes a conjunction with $\bigl(P_{4}-1\ge0\;\lor\;\textcolor{blue}{P_{3}}-1\ge0\bigr)$, then it necessarily holds that \(P_{4}\ge1\). This contradicts \(P_{4}=0\) as required for the semilinear set to be reachable.
	Thus, $I\land\mathcal {F}$ is unsatisfiable, i.e., 
	$
	I(P)\;\Longrightarrow\;\neg\mathcal {F}(P)$.
	This completes the proof that $I$ is an inductive invariant refuting property $\mathcal {F}$.
\end{proof}


%% file: sections/8_appendix_non_serializable_execution_counterexample.tex
\clearpage

\section{Non-Serializable Execution Counterexample}
\label{appendix:non-serializable-execution-example}
Continuing the running example presented in subsec.~\ref{subsec:SerToNsTranslation}, we present in Table~\ref{tab:PetriNetFiringCounterexample} a firing sequence of the Petri net (Fig.~\ref{fig:code2ExamplePN}) resulting in the marking \(M^*\) (satisfying $\mathcal {F}$):

\[
M^* = \{\textcolor{blue}{P_3}(1),\;\textcolor{red}{P_7}(1),\;\textcolor{red}{P_8}(1)\}
\]
%
%
%

%
%
%
%
%

\begin{table}[H]
	\centering
	\label{tab:reach-seq}
	\resizebox{0.9\textwidth}{!}{
		\begin{tabular}{c l c c c c c c}
			\toprule
			\textbf{Step} 
			& \textbf{Firing} 
			& \multicolumn{3}{c}{\textbf{Marking (after firing)}} 
			& \multicolumn{3}{c}{\textbf{Description (after firing)}} \\
			\cmidrule(lr){3-5} \cmidrule(lr){6-8}
			& 
			& \textbf{Global} 
			& \textbf{Local} 
			& \textbf{Responses} 
			& \textbf{Global state} 
			& \textbf{In-flight requests} 
			& \textbf{Responses} \\
			\midrule
			0 & --                                  
			& {\color{blue}$P_3$(1)}                  
			& --                                    
			& --                                    
			& {\color{blue}[X=0]}                   
			& --                          
			& --                                    \\
			1 & $\textcolor{ForestGreen}{t_1}$ 
			& {\color{blue}$P_3$(1)}                  
			& $P_1$(1)                                
			& --                                    
			& {\color{blue}[X=0]}                   
			& {\color{ForestGreen}$\blacklozenge_\text{main}$} 
			& --                                    \\
			2 & $\textcolor{ForestGreen}{t_1}$ 
			& {\color{blue}$P_3$(1)}                  
			& $P_1$(2)                                
			& --                                    
			& {\color{blue}[X=0]}                   
			& {\color{ForestGreen}$\blacklozenge_\text{main}$}, {\color{ForestGreen}$\blacklozenge_\text{main}$}  
			& --                                    \\
			3 & $t_3$                                  
			& {\color{blue}$P_2$(1)}                  
			& $P_1$(1),$P_4$(1)                          
			& --                                   
			&                                    {\color{blue}[X=1]}    
			&                                    {\color{black}$\blacklozenge_\text{until yield}$}, {\color{ForestGreen}$\blacklozenge_\text{main}$}   
			& --                                    \\
			4 & $t_2$                                  
			& {\color{blue}$P_2$(1)}                  
			& $P_4$(2)                                
			& --                                    
			&                                    {\color{blue}[X=1]}    
			&                                    {\color{black}$\blacklozenge_\text{until yield}$}, {\color{black}$\blacklozenge_\text{until yield}$}   
			& --                                    \\
			5 & $t_4$                                  
			& {\color{blue}$P_3$(1)}                  
			& $P_5$(1),$P_4$(1)                          
			& --                                    
			&                                   {\color{blue}[X=0]}     
			&                                    {\color{black}$\blacklozenge_\text{after yield}$}, {\color{black}$\blacklozenge_\text{until yield}$}   
			& --                                    \\
			6 & $\textcolor{red}{t_6}$                     
			& {\color{blue}$P_3$(1)}                  
			& $P_4$(1)                                
			& {\color{red}$P_7$(1)}                    
			&                                      	{\color{blue}[X=0]}  
			&                                    {\color{black}$\blacklozenge_\text{until yield}$}   
			&                                   {\color{red}$\blacklozenge_1$}     \\
			7 & $t_5$                                  
			& {\color{blue}$P_3$(1)}                  
			& $P_6$(1)                                
			& {\color{red}$P_7$(1)}                    
			&                                   {\color{blue}[X=0]}    
			&                                    {\color{black}$\blacklozenge_\text{after yield}$}      
			&                                   {\color{red}$\blacklozenge_1$}        \\
			8 & $\textcolor{red}{t_7}$                     
			& {\color{blue}$P_3$(1)}                                  
			& --                                    
			& {\color{red}$P_7$(1),\color{red}$P_8$(1)}    
			&                                   {\color{blue}[X=0]}    
			&                                   --    
			&                                   {\color{red}$\blacklozenge_0$}, {\color{red}$\blacklozenge_1$}       \\
			\bottomrule
		\end{tabular}
	}
	\caption{The firing sequence reaching marking $M^*$ which is in our target semilinear set $\mathcal {F}$. The marking $P_i(n_j)$ indicates that there are $n_j$ tokens in place $P_i$. The initial marking has a single token in place $\textcolor{blue}{P_3}$, encoding $g_0$ ($\textcolor{blue}{\texttt{[X=0]}}$).}
	\label{tab:PetriNetFiringCounterexample}
\end{table}


%% file: sections/8_appendix_bidirectional_optimization_proof.tex
\clearpage

\section{Proof: Bidirectional Slicing Correctness}
\label{appendix:BidirectionalProof}

%
%
%
%
%

\subsection{The Bidirectional Slicing Algorithm}

Let $N=(P,T,\Pre,\Post, M_0)$ be a Petri net and $S\subseteq\mathbb{N}^P$ be a target set.
By convention, we assume that $P$ and $T$ are disjoint.
 
\begin{definition}[Forward Over-Approximation]
	Define the operator $\mathcal{F}:\mathcal{P}(P\cup T)\to\mathcal{P}(P\cup T)$ by
	\[
	X \mapsto X
	~\cup~
	\{\,t\in T \mid \forall p\in P:\; \Pre(t,p)>0 \implies p\in X\}
	~\cup~
	\{\,p\in P \mid \exists t\in X\cap T,\ \Post(t,p)>0\}.
	\]
	Starting from $X_0 = \{\,p\mid M_0(p)>0\}$, iterate
	$X_{i+1} = \mathcal{F}(X_i)$ until a least fixed-point
	$X^*=\bigcup_i X_i$ is reached.  Call $X^*_P = X^*\cap P$ the set of
	forward-reachable places.
\end{definition}

\begin{definition}[Backward Over-Approximation]
	Let
	\[
	Y_0 = \{\,p\in P \mid \exists M\in S:\;M(p)\neq0\}
	\]
	be the places unconstrained to zero by the target.  Define
	$\mathcal{B}:\mathcal{P}(P\cup T)\to\mathcal{P}(P\cup T)$ by
	\[
	Y \mapsto Y
	~\cup~
	\{\,t\in T \mid \forall p\in P:\; \Post(t,p)>0 \implies p\in Y\}
	~\cup~
	\{\,p\in P \mid \exists t\in Y\cap T,\ \Pre(t,p)>0\}.
	\]
	Starting from $Y_0$, defined as the set of all places that are not constrained to zero in the target set $S$ and also have a token in $M_0$;
	iterate $Y_{i+1} = \mathcal{B}(Y_i)$ until a least fixed-point
	$Y^*=\bigcup_i Y_i$ is reached.  Call $Y^*_P = Y^*\cap P$ the set of
	backward-relevant places.
\end{definition}

\begin{definition}[Sliced Net]
  Let
  \begin{align*}
    P' &= X^*_P \;\cap\; Y^*_P,
    \\
    T' &= \{\,t\in T \mid
    \forall p:\;\Pre(t,p)>0\implies p\in P',\;
    \forall p:\;\Post(t,p)>0\implies p\in P'
    \}.
  \end{align*}
  If $M_0(p) > 0$ for any $p \not\in P'$, then the sliced subnet is undefined.
  Otherwise, the sliced subnet is
  \[
  N' = \bigl(P',\,T',\,\Pre|_{P'\times T'},\,\Post|_{P'\times T'}, M_0|_{P'}\bigr).
  \]
\end{definition}

\subsection{Invariant and Correctness}

Intuitively, $P'$ contains an over-approximation of all the places reachable by a firing sequence starting with marking $M_0$ and ending with a marking in $S$.

\begin{definition}[Witnessable Place]
	A place $p\in P$ is \emph{witnessable} if there exist firing
	sequences $\sigma_1,\sigma_2\in T^*$ and markings $M$ and $M'$ such that
	\[
	M_0 \xrightarrow{\sigma_1} M
	\quad\text{and}\quad
	M \xrightarrow{\sigma_2} M'
	\quad\text{with}\quad
	M(p)>0
	\quad\text{and}\quad
	M'\in S.
	\]
	In other words, $p$ can carry a token in some execution from $M_0$ to a marking in the target set $S$.
\end{definition}

\begin{theorem}[Slicing Invariant]
	\label{thm:invariant}
	If a place $p$ is witnessable, then $p\in P'$.
\end{theorem}

\begin{proof}
	We split the argument into two parts.
	
	\medskip
	\noindent
	\textbf{(1) Forward-reachability.}
	Suppose $p$ is witnessable.  Then there is a prefix
	$\sigma_1\in T^*$ such that $M_0\xrightarrow{\sigma_1}M$ and
	$M(p)>0$.  By standard Petri-net monotonicity, every place that
	receives a token in the course of $\sigma_1$ must appear in the
	forward fixed-point $X^*_P$.  Hence $p\in X^*_P$.
	
	\medskip
	\noindent
	\textbf{(2) Backward-relevance.}
	Again, since $p$ is witnessable, there is a suffix
	$\sigma_2\in T^*$ from $M$ to $M'\in S$ with $M(p)>0$.  Working
	backward from $S$, every place that can contribute to satisfying the
	semilinear constraints appears in the backward fixed-point $Y^*_P$.
	Thus $p\in Y^*_P$.
	
	\paragraph{Conclusion.}
	Combining (1) and (2) yields $p\in X^*_P\cap Y^*_P = P'$, as desired.
\end{proof}

\begin{corollary}
  If $M_0(p) > 0$ for any $p \not\in P'$ (\textit{i.e.}, if the sliced net is undefined), then $S$ is not reachable from $M_0$.
\end{corollary}

\begin{corollary}[Bidirectional Slicing Soundness]
	Let $N = (P, T, \Pre, \Post, M_0)$ be a Petri net and $S$ a target set.  
	Let $N' = (P',T',\,\Pre|_{P'\times T'},\,\Post|_{P'\times T'},\,M_0|_{P'})$ be the sliced net.  
	Then $S$ is reachable from $N$ iff it is reachable from $N'$.
\end{corollary}

\subsection{Termination and Complexity}

\begin{lemma}
	Each iteration of $\mathcal{F}$ and $\mathcal{B}$ strictly increases
	the set of included elements (unless already at the fixed point), and
	the total number of elements is finite.  Hence, both reach their
	fixed points in at most $|P|+|T|$ iterations each.
\end{lemma}

\begin{proof}
	Immediate from monotonicity and finiteness.
\end{proof}

\noindent
Therefore, the bidirectional slicing converges in polynomial time and preserves an over-approximation of the
places and transitions that \emph{may} appear in some firing sequence from
$M_0$, as part of a marking ending in the target semilinear set $S$.


\begin{figure}[H]
	\centering
	
	\begin{subfigure}[b]{0.45\textwidth}
		\centering
		\includegraphics[width=\textwidth]{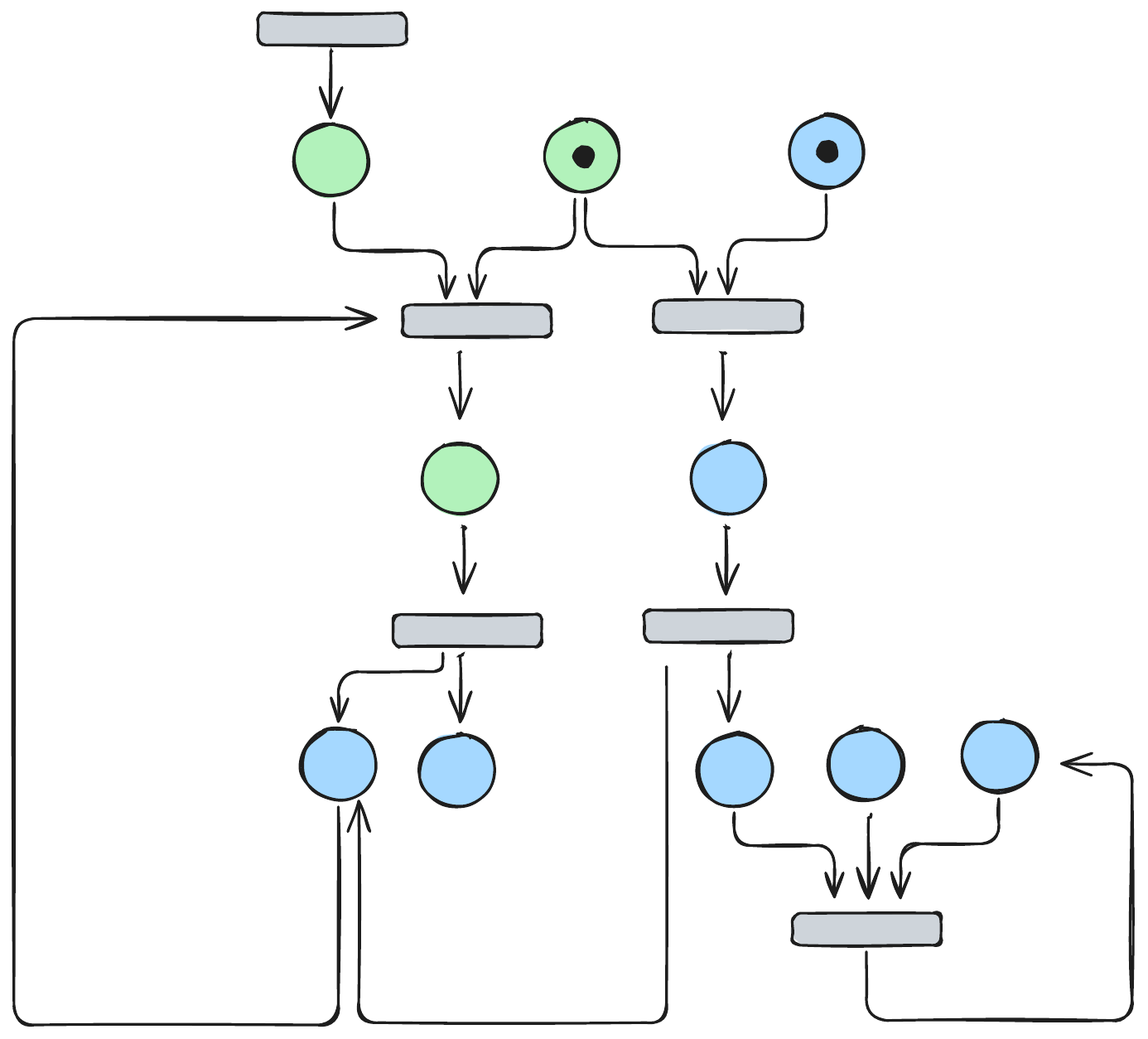}
		\caption{Step 0: initial Petri net, before slicing.}
		\label{fig:step:a}
	\end{subfigure}\hfill
	\begin{subfigure}[b]{0.45\textwidth}
		\centering
		\includegraphics[width=\textwidth]{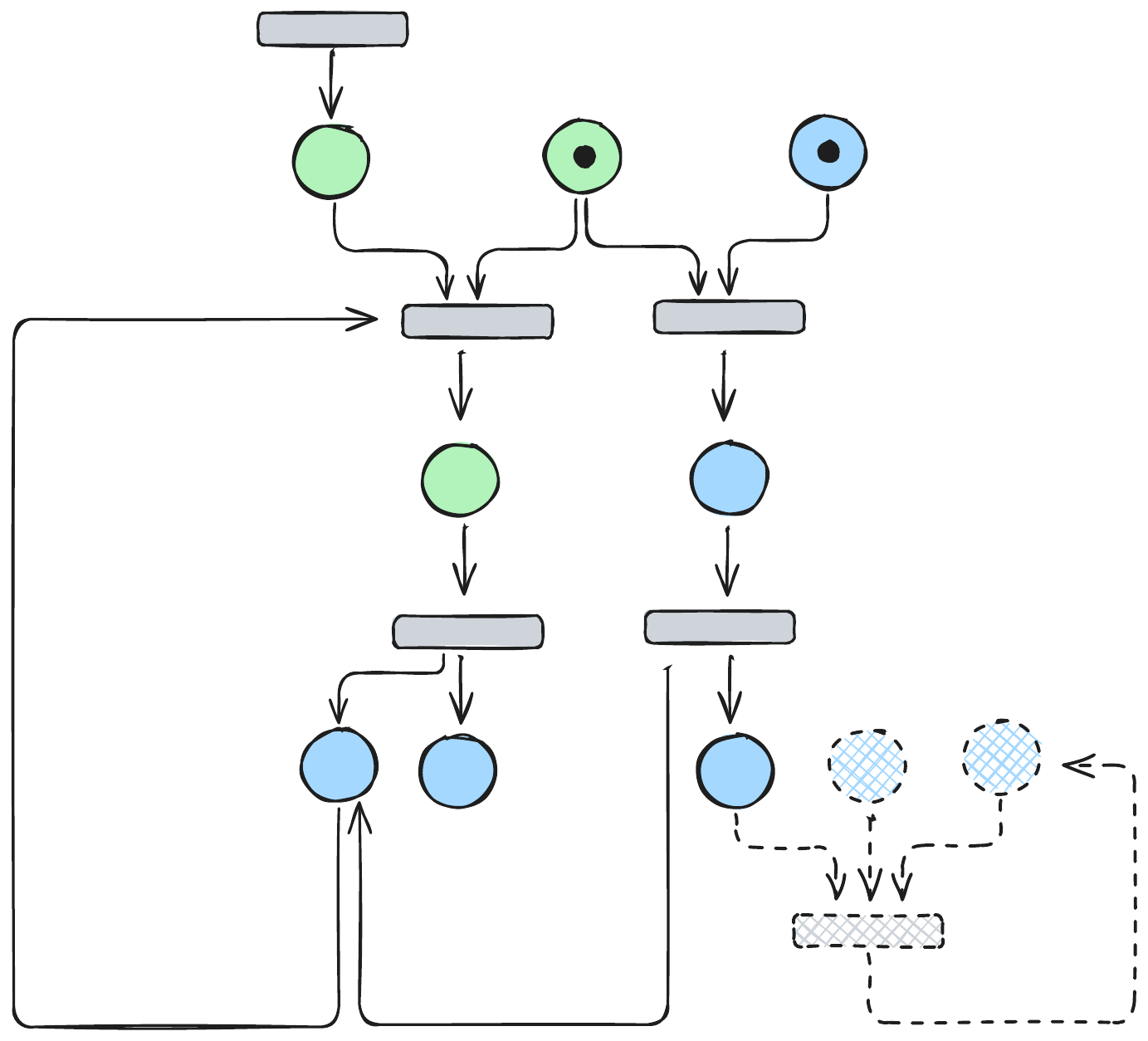}
		\caption{Step 1: first forward pass.}
		\label{fig:step:b}
	\end{subfigure}
	
	\vspace{1em}
	
	\begin{subfigure}[b]{0.30\textwidth}
		\centering
		\includegraphics[width=\textwidth]{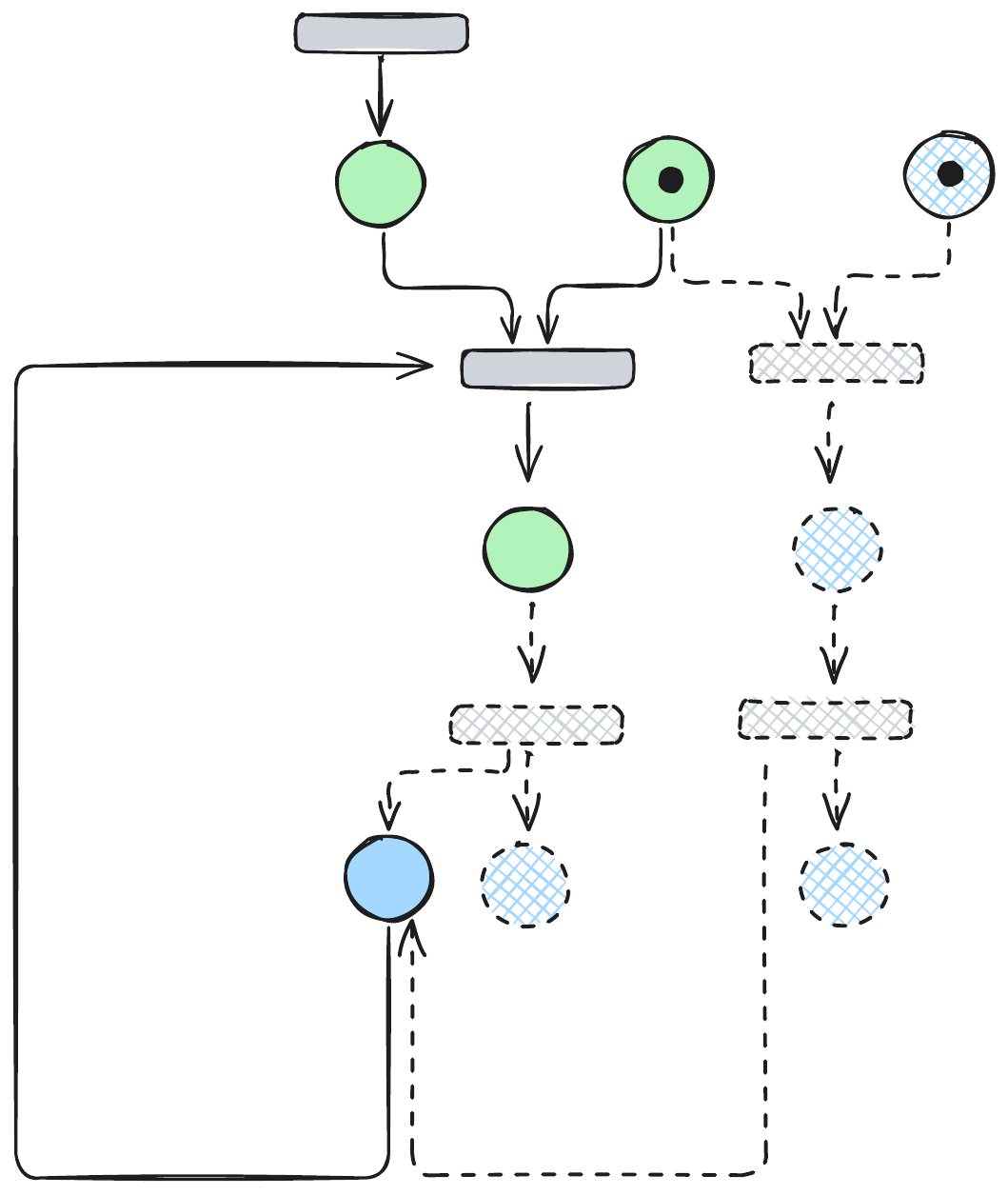}
		\caption{Step 2: first backward pass.}
		\label{fig:step:c}
	\end{subfigure}\hfill
	\begin{subfigure}[b]{0.23\textwidth}
		\centering
		\includegraphics[width=\textwidth]{plots/bidirectional_pruning_step_d_updated_2.pdf}
		\caption{Step 3: second forward pass.}
		\label{fig:step:d}
	\end{subfigure}\hfill
	\begin{subfigure}[b][\subfigheight][b]{0.23\textwidth}
		\centering
		\includegraphics[width=\textwidth]{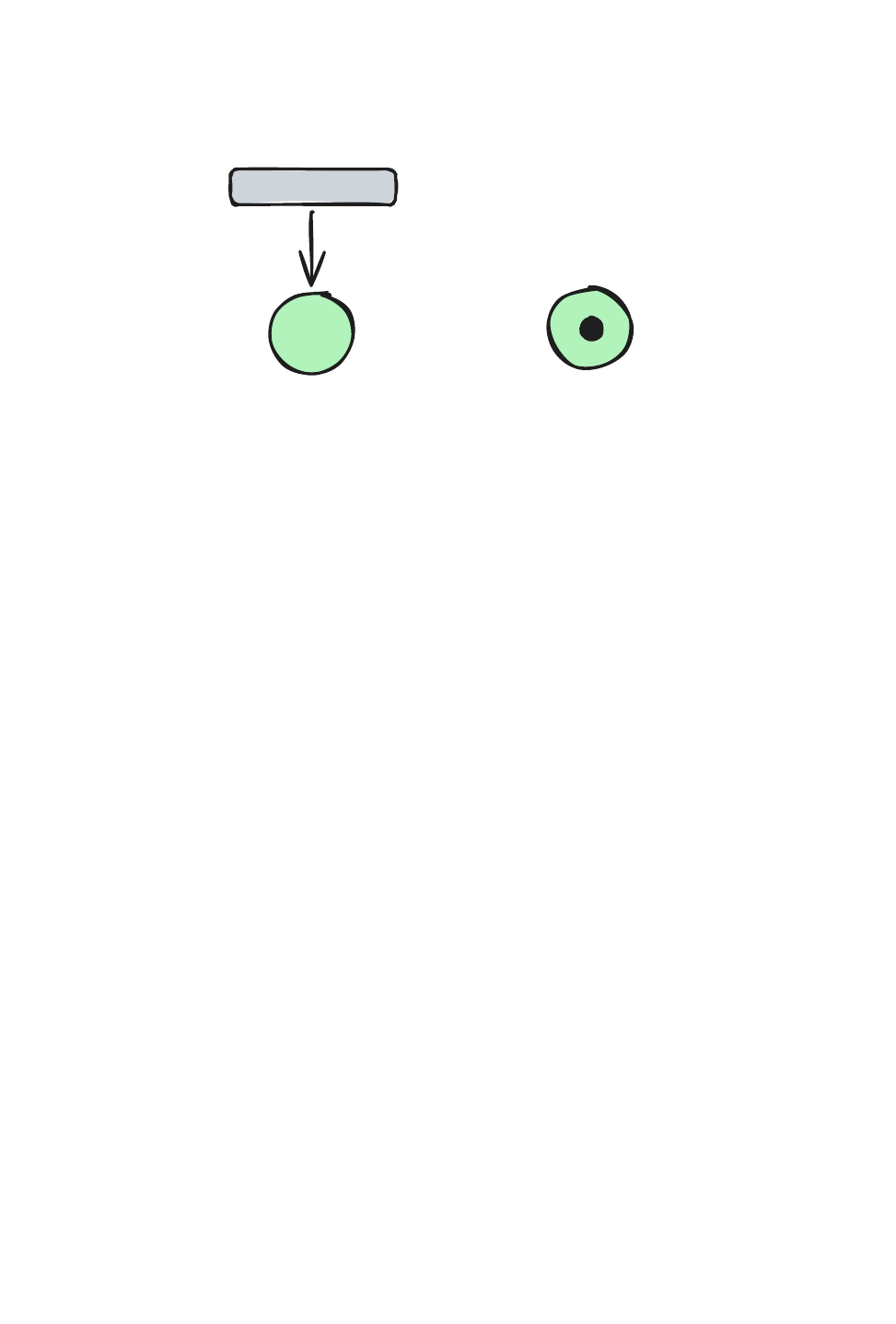}
		\caption{Step 4: final Petri net.}
		\label{fig:step:e}
	\end{subfigure}
	
	\caption{A Petri net during three rounds of bidirectional slicing: two forward passes and one backward pass. Black dots represent initial token markings; green places represent places that are allowed to be reachable in our constraints (i.e., aren't fixed to zero tokens in the final marking). Dashed shapes represent places and transitions that are identified as removable in the current iteration, and will be removed after it ends.}
	\label{fig:bidirectional_pruning}
\end{figure}


%% file: sections/8_appendix_big_table_and_more_results.tex
\clearpage

\section{Evaluation: Full Results}
\label{appendix:full_results}

See Table~\ref{tab:benchmarks-all}.

\begin{table}[htbp]
	\centering
	\input{tables/big_table_summary.tex}
	\caption{Overview of our benchmarks (\texttt{TIMEOUT} is $500$ seconds).}
	\label{tab:benchmarks-all}
\end{table}

\subsection{Optimization Analysis}
\label{subsec:optimization-results}

%

\subsubsection{Runtime optimization.}

We ran all benchmarks with each of the following six optimization configurations: 
(i) without any optimization (marked [\texttt{\textbf{\text{-}\text{-}\text{-}\text{-}}}] in Fig.~\ref{fig:timeout_cumulative_solved_log}); (ii) with bidirectional slicing (marked [\texttt{\textbf{\text{B}\text{-}\text{-}\text{-}}}]); (iii) with redundant constraint elimination (marked [\texttt{\textbf{\text{-}\text{R}\text{-}\text{-}}}]); (iv) with generation of fewer constraints (marked [\texttt{\textbf{\text{-}\text{-}\text{G}\text{-}}}]);
(v) with strategic Kleene elimination (marked [\texttt{\textbf{\text{-}\text{-}\text{-}\text{S}}}]);
and finally, (vi) with all optimizations altogether (marked [\texttt{\textbf{\text{B}\text{R}\text{G}\text{S}}}]).
The results of the aggregated runtimes are presented in Fig.~\ref{fig:timeout_cumulative_solved_log} and show that over $28\%$ more benchmarks are solved when using all optimizations compared to running without any optimization.
Not surprisingly, the best configuration is the one with all optimizations on. 
Furthermore, the best single-optimization configurations with regard to runtime are [\texttt{\textbf{\text{-}\text{-}\text{G}\text{-}}}] and [\texttt{\textbf{\text{B}\text{-}\text{-}\text{-}}}], solving over $74\%$ and $72\%$ of the benchmarks respectively. 
We also note that the two remaining optimizations, [\texttt{\textbf{\text{-}\text{R}\text{-}\text{-}}}] and [\texttt{\textbf{\text{-}\text{-}\text{-}\text{S}}}], performed slightly worse (although not significantly) than without the optimizations when counting overall timeouts.
However, when analyzing the redundant constraint optimization (  [\texttt{\textbf{\text{-}\text{R}\text{-}\text{-}}}]), we identified instances in which it still \textit{strictly} improves runtime.
For example, the optimization affords a speedup of between $72.2\%$ and $85.2\%$ for benchmarks \texttt{a3.ser} and \texttt{a7.ser},  when compared to the baseline.

\subsubsection{Space optimization.}
Our optimizations also reduce the space complexity of the two main components --- the Petri net and the semilinear set.

\noindent
(1) \textbf{Petri net.} Bidirectional slicing (Fig.~\ref{fig:petri_size_reduction}) 
eliminates the average number of places \textit{by roughly half} --- from $23.91$ down to $12.79$. This optimization proved even more effective on transitions, \textit{eliminating about two-thirds}: from $37.3$ down to $12.61$. 

\noindent
(2) \textbf{Semilinear sets.} We ran an ablation experiment in which we compared all optimizations against runs where each of the three semilinear optimizations (i.e., all but PN slicing) was disabled. The redundant-constraint elimination (with a negated effect in [\texttt{\textbf{\text{B}\textcolor{red}{-}\text{G}\text{S}}}]) and the fewer-constraint generation elimination (with a negated effect in [\texttt{\textbf{\text{B}\text{R}\textcolor{red}{-}\text{S}}}]) \textit{drastically} reduced component counts, with the latter being especially effective in reducing the \textit{maximal} number of components to be up to $\mathbf{931\times}$ smaller, and the \textit{average} number of components to be up to $\mathbf{223\times}$ smaller (Table~\ref{tab:semilinear-size-reduction}), when compared to the baseline executions configured with all optimizations on ([\texttt{\textbf{\text{B}\text{R}\text{G}\text{S}}}]). 
For fairness, we measured only benchmarks completed under all configurations, excluding cases where semilinear sets exploded beyond $2^{30}$ components and timed out. Thus, our reported improvements actually \textit{understate} the true impact of these optimizations on memory. Such blowups, 
render even simple programs intractable without these optimizations.

\begin{table}[!htbp]
	\centering
		\includegraphics[width=0.68\linewidth]{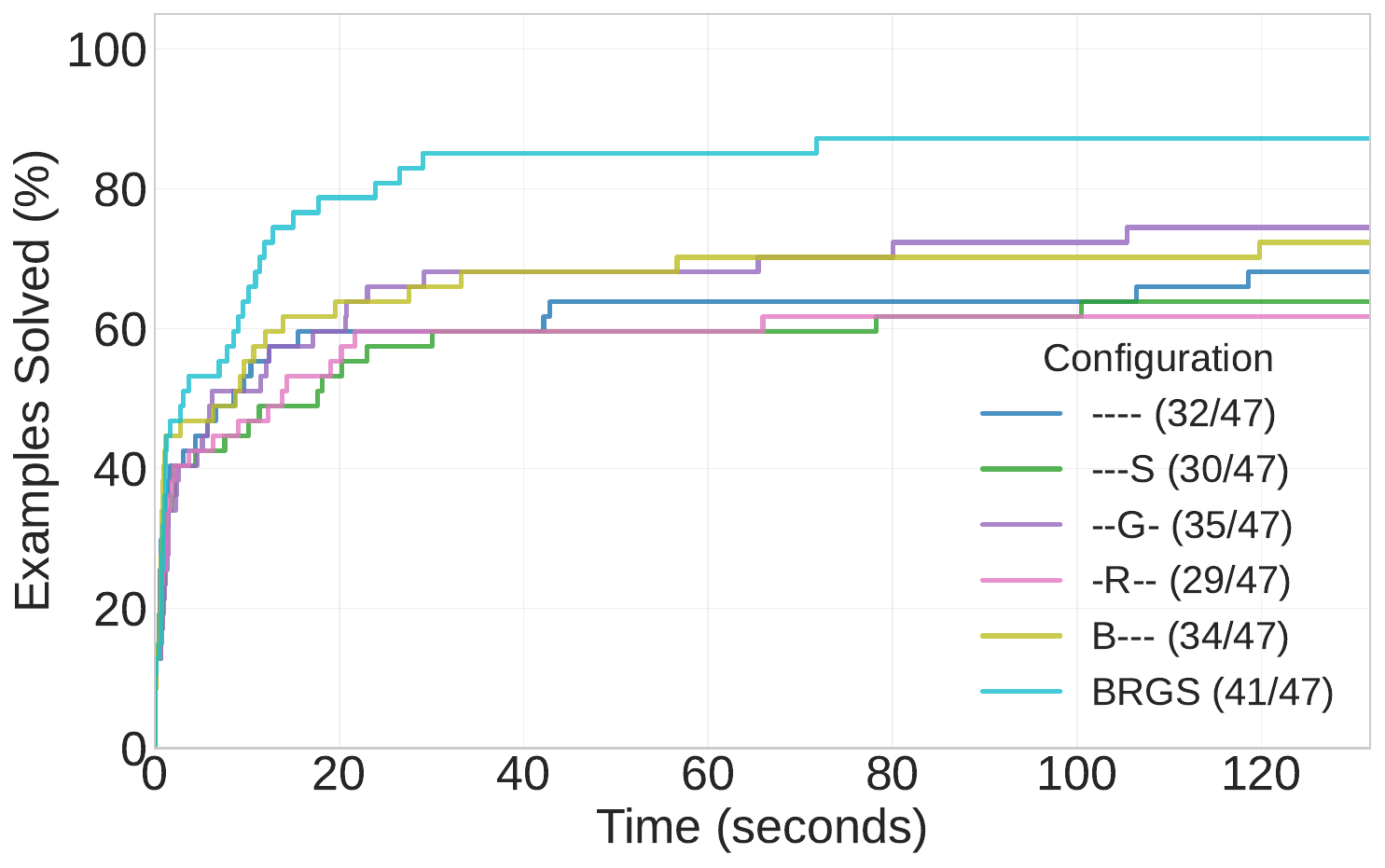}
	\captionof{figure}{Solved instances (\texttt{TIMEOUT} is $150$ seconds).}
	\label{fig:timeout_cumulative_solved_log}
\end{table}

\begin{table}[!htbp]
	\centering
	\includegraphics[width=0.68\linewidth]{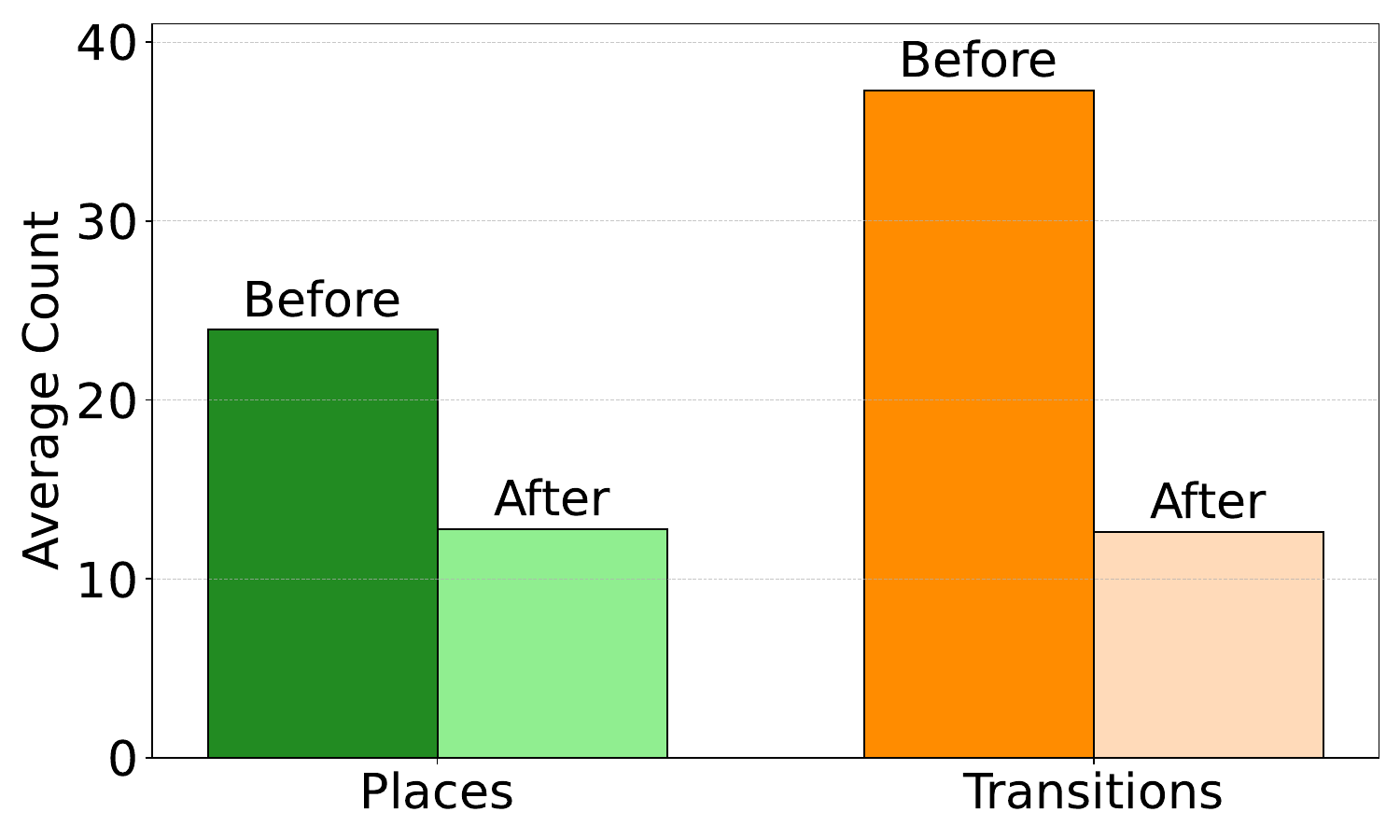}
	\captionof{figure}{PN size reduction via slicing.}
	\label{fig:petri_size_reduction}
\end{table}


\begin{table}[!htbp]
	\centering
	\input{tables/semilinear_size_reduction_compact.tex}
	\caption{Semilinear set size reduction via optimizations (baseline is [\texttt{\text{B}\text{R}\text{G}\text{S}}]).}
	\label{tab:semilinear-size-reduction}
\end{table}

%% file: tables/big_table_summary.tex
\begin{table}[H]
	\centering
	\small
	\setlength{\tabcolsep}{5pt}
	\renewcommand{\arraystretch}{0.9}
	\begin{tabular*}{\textwidth}{@{\extracolsep{\fill}}%
			p{1.5cm}   
			p{1.0cm} 
			c        
			c c c c c c 
			r r       
		}
		\toprule
		\multicolumn{2}{c}{\textbf{Benchmark}}
		& \textbf{Serializable}
		& \multicolumn{6}{c}{\textbf{Features}}
		& \multicolumn{2}{c}{\textbf{Runtime (ms)}} \\
		\cmidrule(lr){1-2} \cmidrule(lr){3-3} \cmidrule(lr){4-9} \cmidrule(lr){10-11}
		&
		&
		& If & While & \texttt{?} & Arith & Yield & Multi-req
		& Cert. & Total \\
		\midrule
		\multirow{7}{=}{Core expressions} & \texttt{a1.ser} & \greencmark &  & \cmark &  &  &       &   & 2 & 47 \\
		 & \texttt{a2.ser} & \xmark &  &        &  &  & \cmark &   & 280 & 296 \\
		 & \texttt{a3.ser} & \greencmark &  &        &  &  &       &   & 1 & 32 \\
		 & \texttt{a4.ser} & \greencmark &  &        &  &  & \cmark & \cmark & 637 & 1{,}071 \\
		 & \texttt{a5.ser} & \greencmark &  & \cmark &  &  & \cmark & \cmark & 3{,}234 & 13{,}624 \\
		 & \texttt{a6.ser} & \xmark &  &        &  &  & \cmark & \cmark & 757 & 775 \\
		 & \texttt{a7.ser} & \greencmark & \cmark & \cmark &  &  & \cmark &   & 4 & 33 \\
		\midrule
		\multirow{4}{=}{State machines} & \texttt{b1.json} & \greencmark & \cmark &        &  &  & \cmark & \cmark & 683 & 968 \\
		 & \texttt{b2.json} & \greencmark & \cmark &        &  &  & \cmark & \cmark & 2{,}063 & 7{,}802 \\
		 & \texttt{b3.json} & \greencmark & \cmark &        &  &  & \cmark & \cmark & 730 & 2{,}080 \\
		 & \texttt{b4.json} & \greencmark & \cmark &        &  &  & \cmark & \cmark & 660 & 1{,}909 \\
		\midrule
		\multirow{8}{=}{Mixed arithmetic} & \texttt{c1.ser} & \xmark &  & \cmark &  & \cmark & \cmark & \cmark & 356{,}195 & 356{,}299 \\
		 & \texttt{c2.ser} & \greencmark &  & \cmark &  & \cmark & \cmark & \cmark & 9{,}858 & 292{,}228 \\
		 & \texttt{c3.ser} & \greencmark &  & \cmark &  & \cmark & \cmark & \cmark & 1{,}886 & 2{,}397 \\
		 & \texttt{c4.ser} & \greencmark &  & \cmark &  & \cmark & \cmark & \cmark & 4{,}336 & 7{,}193 \\
		 & \texttt{c5.ser} & \xmark &  & \cmark &  & \cmark & \cmark & \cmark & 43{,}694 & 43{,}735 \\
		 & \texttt{c6.ser} & \xmark &  & \cmark &  & \cmark & \cmark & \cmark & 629 & 698 \\
		 & \texttt{c7.ser} & \xmark &  & \cmark &  & \cmark & \cmark & \cmark & 797 & 875 \\
		 & \texttt{c8.ser} & \greencmark &  & \cmark &  & \cmark & \cmark & \cmark & 4{,}357 & 8{,}931 \\
		\midrule
		\multirow{5}{=}{Circular increment} & \texttt{d1.ser} & \greencmark & \cmark & \cmark & \cmark &  & \cmark &   & 2{,}391 & 5{,}373 \\
		 & \texttt{d2.ser} & \xmark & \cmark &        & \cmark &  &   \cmark &   & 628 & 731 \\
		 & \texttt{d3.ser} & \greencmark & \cmark & \cmark & \cmark &  &  \cmark &   & 2{,}642 & 10{,}266 \\
		 & \texttt{d4.ser} & \greencmark & \cmark & \cmark & \cmark &  &     \cmark &   & 5{,}604 & 22{,}249 \\
		 & \texttt{d5.ser} & \xmark & \cmark &        &  &  & \cmark &   & 495 & 554 \\
		\midrule
		\multirow{7}{=}{Concurrency \& locking loops} & \texttt{e1.ser} & \greencmark &  & \cmark &  &  & \cmark &   & 351 & 502 \\
		 & \texttt{e2.ser} & \xmark & \cmark & \cmark &  & \cmark & \cmark & \cmark & \texttt{TIMEOUT} & \texttt{TIMEOUT} \\
		 & \texttt{e3.ser} & \xmark & \cmark & \cmark &  & \cmark &   \cmark & \cmark & 24{,}899 & 25{,}039 \\
		 & \texttt{e4.ser} & \xmark & \cmark & \cmark &  &  \cmark &   \cmark & \cmark & 273{,}062 & 273{,}351 \\
		 & \texttt{e5.ser} & \greencmark & \cmark & \cmark & \cmark &  & \cmark &   & 2 & 55 \\
		 & \texttt{e6.ser} & \greencmark & \cmark & \cmark & \cmark &  & \cmark &   & 10 & 114 \\
		 & \texttt{e7.ser} & \greencmark &  & \cmark &  &  &   \cmark &   & 299 & 444 \\
		\midrule
		\multirow{9}{=}{Non-determinism} & \texttt{f1.ser} & \greencmark & \cmark &    \cmark    & \cmark &  & \cmark &   & 388 & 494 \\
		 & \texttt{f2.ser} & \xmark & \cmark &   \cmark     & \cmark &  & \cmark &   & 612 & 676 \\
		 & \texttt{f3.ser} & \xmark &  &        &  & \cmark &   \cmark & \cmark & 653 & 716 \\
		 & \texttt{f4.ser} & \greencmark &  &     \cmark   &  & \cmark & \cmark & \cmark & 1{,}626 & 9{,}515 \\
		 & \texttt{f5.ser} & \greencmark & \cmark &        & \cmark &  &       &   & 7{,}401 & 11{,}301 \\
		 & \texttt{f6.ser} & \xmark & \cmark &        & \cmark &  & \cmark &   & 646 & 830 \\
		 & \texttt{f7.ser} & \xmark & \cmark &        & \cmark &  &  \cmark &   & 400 & 427 \\
		 & \texttt{f8.ser} & \xmark & \cmark &        & \cmark &  &   \cmark &   & 773 & 802 \\
		 & \texttt{f9.ser} & \greencmark & \cmark &        & \cmark &  &  \cmark &   & 10 & 94 \\
		\midrule
		\multirow{7}{=}{Network \& system protocols} & \texttt{g1.ser} & \xmark & \cmark & \cmark &  & \cmark & \cmark & \cmark & 59{,}312 & 74{,}539 \\
		 & \texttt{g2.ser} & \greencmark & \cmark & \cmark &  & \cmark & \cmark & \cmark & \texttt{TIMEOUT} & \texttt{TIMEOUT} \\
		 & \texttt{g3.ser} & \xmark & \cmark & \cmark & \cmark & \cmark & \cmark & \cmark & 20{,}557 & 20{,}954 \\
		 & \texttt{g4.ser} & \xmark & \cmark & \cmark & \cmark & \cmark & \cmark & \cmark & 6{,}859 & 7{,}047 \\
		 & \texttt{g5.ser} & \greencmark & \cmark & \cmark & \cmark & \cmark &   \cmark & \cmark & 3{,}047 & 12{,}324 \\
		 & \texttt{g6.ser} & \xmark & \cmark &        & \cmark & \cmark & \cmark &   & 8{,}193 & 8{,}285 \\
		 & \texttt{g7.ser} & \greencmark & \cmark &        & \cmark & \cmark &       &   & 6{,}886 & 252{,}752 \\
		\midrule
\bottomrule
	\end{tabular*}
\end{table}

%% file: tables/semilinear_size_reduction_compact.tex
\begin{table}[H]
	\centering
	\begin{tabular}{l c c c c}
		\toprule
		& \multicolumn{2}{c}{components} & \multicolumn{2}{c}{periods/component} \\
		\cmidrule(lr){2-3} \cmidrule(lr){4-5}
		& average & max & average & max \\
		\midrule
	\texttt{\text{B}\text{R}\text{G}\text{S}} & 2.91 & 22 & 1.33 & 4 \\
	\texttt{\text{B}\textcolor{red}{-}\text{G}\text{S}} & 8.79 & 194 & \textbf{1.64} & 11 \\
	\texttt{\text{B}\text{R}\textcolor{red}{-}\text{S}} & \textbf{651.41} & \textbf{20{,}484} & 1.28 & \textbf{15} \\
	\texttt{\text{B}\text{R}\text{G}\textcolor{red}{-}} & 2.91 & 22 & 1.35 & 4 \\
  \bottomrule
	\end{tabular}
\end{table}

%% file: sections/8_appendix_SMPT.tex
\clearpage

\section{Petri Net Model Checking}
\label{appendix:smpt}

\subsection{Petri Nets and VAS(S) Reachability}
	
	Our work builds on both theoretical and practical advances in 
	Petri net research, and specifically, \textit{Petri net model checking}~\cite{Mu89,Es96,Re12,EsNi24,DuLaSr25,HuScReAb17,AmBeDo14,PiHaRe20,Wo18}.
	Moreover, numerous studies (including~\cite{LiWaChSuZh02,Zu91,AkChDaJaSa17,AnPePe13,AnBeCh16,EtChRo16,AnBeCh14}, among others) have explored \textit{specific classes of Petri nets}, providing deeper insights into their structure, expressiveness, and verification challenges.
	%
	%

	\medskip
	While deciding reachability in a bounded Petri net may be straightforward (through exhaustive enumeration), the \textit{unbounded} case is highly nontrivial and was first solved by 
	Mayr~\cite{Ma81}, with subsequent improvements by Kosaraju~\cite{Ko82} and 
	Lambert~\cite{La92}. Recent work~\cite{CzWo22,Le22} has also established that this 
	problem is \texttt{Ackermann}-complete.
	These theoretical advances in Petri net reachability have given rise to a 
	plethora of practical tools, including \texttt{KReach}~\cite{DiLa20}, 
	\texttt{DICER}~\cite{XiZhLi21}, \texttt{MARCIE}~\cite{HeRoSc13}, and others. 
	Our implementation leverages \texttt{SMPT} (\emph{Satisfiability Modulo Petri Nets})~\cite{AmDa23}, a state-of-the-art model checker that combines \texttt{SMT}-solving with structural invariants~\cite{AmBeDa21,amat2022polyhedral}.

\subsection{SMPT}

\texttt{SMPT} incorporates a portfolio of symbolic model checking techniques --- including bounded model checking (BMC)~\cite{BiCiClZh99}, state equation reasoning~\cite{Mu77}, $k$-induction~\cite{BeDaWe18,ShSiSt20}, property directed reachability (PDR)~\cite{Br11,AmDaHu22,ViGu14,BjGa15,BlLa23,CiGrMoTo14,CiGrMoTo16,DuRo17}, and random state space exploration. It acts as a front-end to an \texttt{SMT} solver (\texttt{Z3}~\cite{DeBj08}, although other solvers could also be used, e.g., \texttt{cvc5}~\cite{BaCoDeHaJoKiReTi11,BaBaBrKrLaMaMoMoNiNo22}, \texttt{MathSAT}~\cite{CiGrScSe13}, etc.), while also incorporating domain-specific knowledge from Petri net theory, such as invariants and structural properties. \texttt{SMPT} has also participated in the last five editions of the \textit{Model Checking Contest} (MCC), an international competition for model-checking tools. In its most recent participation, it achieved a bronze medal and a confidence level score of  $\texttt{100\%}$, indicating it never returned an incorrect verdict~\cite{mcc:2025}.

\medskip
\texttt{SMPT} distinguishes itself from other tools in two ways that are particularly relevant to our setting and motivate its adoption. First, to the best of our knowledge, it is the only model checker for Petri nets that provides a proof of its verdict, regardless of the underlying verification technique. This means it either produces a witness trace when the property is reachable, or, more interestingly, a certificate of non-reachability~\cite{AmDaHu22} when the property is found to be unreachable.
The second distinguishing feature relates to our ongoing work on polyhedral reductions~\cite{AmBeDa21,amat2022polyhedral}, as elaborated in~\Cref{sec:discussion}.
